\documentclass[aps,prx,reprint,superscriptaddress]{revtex4-2}
\usepackage{silence}
\usepackage{amsmath}
\usepackage{amssymb}
\usepackage{physics}
\usepackage{mathtools}
\usepackage{enumitem}
\usepackage{algpseudocode}
\algrenewcommand\algorithmicrequire{\textbf{Input:}}
\algrenewcommand\algorithmicensure{\textbf{Output:}}
\algrenewcommand\algorithmicindent{2em}

\usepackage[caption=false]{subfig}
\usepackage{bm}
\usepackage{siunitx}
\ExplSyntaxOn
\msg_redirect_name:nnn { siunitx } { physics-pkg } { none }
\ExplSyntaxOff
\usepackage[unicode, colorlinks=true, linkcolor=blue, citecolor=blue, urlcolor=blue]{hyperref}
\usepackage{tikz}
\usepackage{amsthm}
\usepackage{thm-restate}

\usepackage{xfp}  
\usepackage{xparse}  

\NewDocumentCommand{\eformat}{m}{%
    \StrBefore{#1}{e}[\mantissa]%
    \StrBehind{#1}{e}[\exponent]%
    \mantissa \times 10^{\exponent}%
}

\usepackage{xstring}

\usepackage{graphicx}
\usepackage{booktabs}
\usepackage{multirow}
\usepackage{makecell}
\usepackage{array}
\usepackage{bm}
\usepackage{subfig}

\newcommand{\statefig}[1]{%
  \begin{minipage}[c]{0.1\textwidth}\centering
    \includegraphics[scale=0.18]{#1}
  \end{minipage}
}
\newcommand{\genfig}[1]{%
  \begin{minipage}[c]{0.17\textwidth}\centering
    \includegraphics[scale=1]{#1}
  \end{minipage}
}

\theoremstyle{definition}
\declaretheorem{theorem}
\newtheorem{apptheorem}{Theorem}[section]

\newtheorem{appcor}{Corollary}[section]

\newtheorem{proposition}{Proposition}

\newtheorem{cor}{Corollary}

\newcommand{\appropto}{\mathrel{\vcenter{
  \offinterlineskip\halign{\hfil$##$\cr
    \propto\cr\noalign{\kern2pt}\sim\cr\noalign{\kern-2pt}}}}}

\newcommand{\gaussprop}{\underset{\mathrm{G}}{\propto}}
\newcommand{\gausssim}{\underset{\mathrm{G}}{\sim}}
\newcommand{\gausssimprop}{\underset{\mathrm{G}}{\appropto}}
\newcommand{\signal}{\mathrm{s}}
\newcommand{\idler}{\mathrm{c}}
\newcommand{\choijam}{Choi--Jamio\l{}kowski}

\begin{document}
\title{Beyond Stellar Rank: Control Parameters for Scalable Optical Non-Gaussian State Generation}
\author{Fumiya Hanamura}
\affiliation{Department of Applied Physics, School of Engineering, The University of Tokyo, 7-3-1 Hongo, Bunkyo-ku, Tokyo 113-8656, Japan}
\author{Kan Takase}
\affiliation{Department of Applied Physics, School of Engineering, The University of Tokyo, 7-3-1 Hongo, Bunkyo-ku, Tokyo 113-8656, Japan}
\affiliation{Optical Quantum Computing Research Team, RIKEN Center for Quantum Computing, 2-1 Hirosawa, Wako, Saitama 351-0198, Japan}
\affiliation{OptQC Corporation, 1-21-7 Nishi-Ikebukuro, Toshima, Tokyo, Japan}
\author{Hironari Nagayoshi}
\affiliation{Department of Applied Physics, School of Engineering, The University of Tokyo, 7-3-1 Hongo, Bunkyo-ku, Tokyo 113-8656, Japan}
\author{Ryuhoh Ide}
\affiliation{Department of Applied Physics, School of Engineering, The University of Tokyo, 7-3-1 Hongo, Bunkyo-ku, Tokyo 113-8656, Japan}
\author{Warit Asavanant}
\affiliation{Department of Applied Physics, School of Engineering, The University of Tokyo, 7-3-1 Hongo, Bunkyo-ku, Tokyo 113-8656, Japan}
\affiliation{Optical Quantum Computing Research Team, RIKEN Center for Quantum Computing, 2-1 Hirosawa, Wako, Saitama 351-0198, Japan}
\affiliation{OptQC Corporation, 1-21-7 Nishi-Ikebukuro, Toshima, Tokyo, Japan}
\author{Kosuke Fukui}
\affiliation{Department of Applied Physics, School of Engineering, The University of Tokyo, 7-3-1 Hongo, Bunkyo-ku, Tokyo 113-8656, Japan}
\author{Petr Marek}
\author{Radim Filip}
\affiliation{Department of Optics, Palacky University, 17. listopadu 1192/12, Olomouc, 77146, Czech Republic.}
\author{Akira Furusawa}
\affiliation{Department of Applied Physics, School of Engineering, The University of Tokyo, 7-3-1 Hongo, Bunkyo-ku, Tokyo 113-8656, Japan}
\affiliation{Optical Quantum Computing Research Team, RIKEN Center for Quantum Computing, 2-1 Hirosawa, Wako, Saitama 351-0198, Japan}
\affiliation{OptQC Corporation, 1-21-7 Nishi-Ikebukuro, Toshima, Tokyo, Japan}
\date{\today}
\begin{abstract}
Advanced quantum technologies rely on non-Gaussian states of light, essential for universal quantum computation, fault-tolerant error correction, and quantum sensing. Their practical realization, however, faces hurdles: simulating large multi-mode generators is computationally demanding, and benchmarks such as the \emph{stellar rank} do not capture how effectively photon detections yield useful non-Gaussianity. We address these challenges by introducing the \emph{non-Gaussian control parameters} $(s_0,\delta_0)$, a continuous and operational measure that goes beyond stellar rank. Leveraging these parameters, we develop a universal optimization method that reduces photon-number requirements and greatly enhances success probabilities while preserving state quality. Applied to the Gottesman--Kitaev--Preskill (GKP) state generation, for example, our method cuts the required photon detections by a factor of three and raises the preparation probability by nearly $10^8$. Demonstrations across cat states, cubic phase states, GKP states, and even random states confirm broad gains in experimental feasibility. Our results provide a unifying principle for resource-efficient non-Gaussian state generation, charting a practical route toward scalable optical quantum technologies and fault-tolerant quantum computation.

\end{abstract}

\maketitle
\section{Introduction}
Optical quantum information processing holds strong promise for scalable quantum computing, owing to its intrinsic compatibility with traveling-wave architectures and potential for large-scale integration. Considerable progress has been made in the preparation and manipulation of Gaussian states~\cite{warit_cluster,mikkel_cluster,cluster_operation,mikkel_cluster_operation}, which are relatively easy to generate and control. Yet non-Gaussian states are indispensable for unlocking the full power of optical platforms, enabling universal quantum computation~\cite{gottesman_knill,nick_universal,cv_qc,wigneg_qc}, quantum error correction~\cite{qec_nogo,gkp,cat_code,binomial_code}, and advanced quantum sensing~\cite{single_mode_disp,metrology_nls,hanamura_gaussian_disp}.  

A common approach to generating non-Gaussianity is to prepare multimode Gaussian states and conditionally project selected modes using photon-number-resolving measurements~\cite{gbs_nongauss,gaussian_breeding,xanadu_architecture}. The recent optical realization of Gottesman--Kitaev--Preskill (GKP) states~\cite{xanadu_gkp} represents a major milestone, yet scaling toward fault tolerance demands larger photon numbers, higher success probabilities, and more complex architectures. Meeting these requirements makes optimization essential, but the search space expands exponentially with system size. In fact, its underlying structure is equivalent to Gaussian boson sampling, a problem known to be classically intractable~\cite{gbs,pan_gbs,chabaudResourcesBosonic}, highlighting the fundamental difficulty of systematic optimization. The central challenge, therefore, is to identify and optimize the true resources that drive non-Gaussianity.

To benchmark such resources, the \emph{stellar rank}~\cite{stellar_rank,stellar_rank_original,core_state,stellar_rank_heterodyne,radim_atom_ng} is often employed. Determined by the total detected photon number, it provides a simple measure directly tied to experimental resources. However, stellar rank does not provide the complete picture: it offers only an upper bound and fails to capture how effectively the non-Gaussianity of the photon-number measurement is harnessed into the generated state. As a result, states with the same rank can exhibit significantly different levels of operational non-Gaussianity, such as non-Gaussian squeezing~\cite{cps_nls,cat_nls,gkp_squeezing}.

To overcome this limitation, we introduce the \emph{non-Gaussian control parameters}, which quantify how efficiently photon-number measurements harness non-Gaussianity. They offer a continuous and operational description of non-Gaussian resources, forming the basis for systematic optimization and scalable state generation beyond stellar rank.

Building on this framework, we develop an optimization algorithm that systematically improves state generators by reducing photon detections and boosting success probabilities, while preserving high fidelity and essential non-Gaussian features such as squeezing. For GKP states, the method lowers the photon-number requirement by a factor of three and enhances the success probability by nearly $10^8$. Comparable gains are obtained for Schrödinger cat and cubic phase states, with photon-number requirements again reduced threefold and success rates improved by up to seven orders of magnitude. We further demonstrate applicability to random states, consistently finding substantial gains in experimental feasibility. Together, these results establish a new design principle for scaling state-of-the-art architectures~\cite{xanadu_gkp,mamoru_four_photon,konno_gkp} toward fault-tolerant quantum computation.

The structure of the paper is as follows. In Sec.~\ref{sec:overview}, we provide a preliminary definition of the non-Gaussian control parameters, giving intuition for their physical meaning. A rigorous derivation is presented in Sec.~\ref{sec:two-mode}, where we focus on two-mode non-Gaussian state generators (Fig.~\ref{fig:gbs}(b)) and introduce the \emph{control-mode representation}, a description of the generator via the first and second moments of the control mode. The non-Gaussian control parameters then naturally emerge from this framework. In Sec.~\ref{sec:cparam_conversion}, we investigate the resource-theoretic aspect of these parameters, showing that they quantify how efficiently the non-Gaussianity characterized by stellar rank is utilized by the state generator. In Sec.~\ref{sec:multi-mode}, we extend the formalism to general multi-mode generators, and in Sec.~\ref{sec:optimization}, we apply these results to formulate the optimization algorithm. Although our primary focus is on pure non-Gaussian states, the framework also extends to mixed states with Gaussian errors such as loss, as discussed in Appendix~\ref{sec:mixed_state}.

\section{Overview of non-Gaussian control parameters and two-mode examples}\label{sec:overview}
\begin{figure*}[tbp]
    \centering
    \input{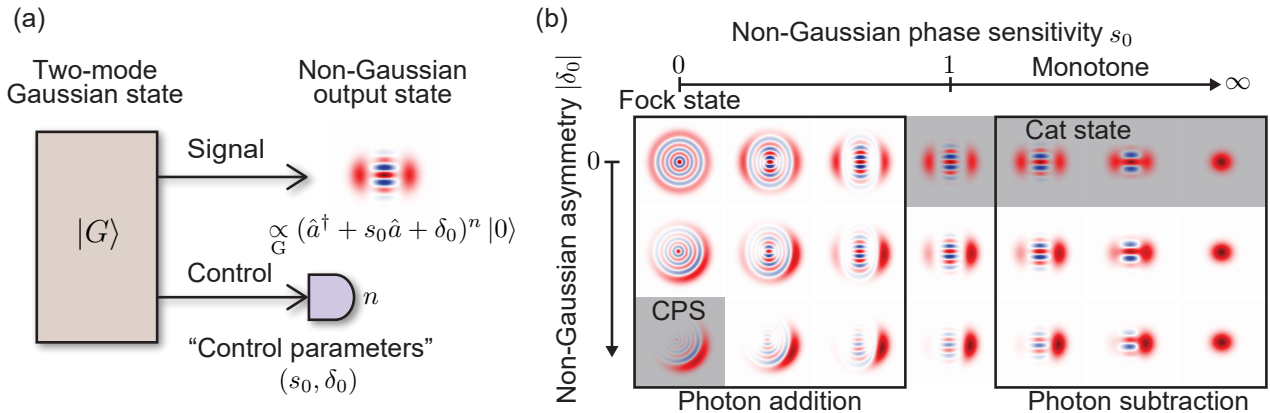}
    \caption{Non-Gaussian state generation via photon-number measurements and non-Gaussian control parameters. 
(a) Schematic of a two-mode non-Gaussian state generator, consisting of a two-mode Gaussian state $\ket{G}$ followed by a photon-number measurement. 
A measurement on the control mode heralds a non-Gaussian state in the signal mode, which can be characterized by the non-Gaussian control parameters $s_0$ and $\delta_0$, determined solely by the properties of the control mode. 
(b) Wigner functions of states generated by the two-mode state generator, organized by the values of the non-Gaussian phase sensitivity $s_0$ and the non-Gaussian asymmetry $\delta_0$. 
Although the output states are defined only up to Gaussian unitaries, we show the core state \cite{core_state}, corresponding to the ``particle form'' of Eq.~\eqref{eq:particle_form_intro}. 
The columns correspond to $s_0 = 0,\ 0.2,\ 0.5,\ 1,\ 1.5,\ 3,\ 100$, while the rows correspond to $\delta_0 = 0,\ 0.2 + 0.2i,\ 0.5 + 0.5i$, with the measured photon number fixed at $n = 6$. 
This simple setup already generates a diverse class of non-Gaussian states: $s_0 > 1$ corresponds to photon-subtracted states, while $0 \leq s_0 < 1$ corresponds to photon-added states (Sec.~\ref{sec:classification}). 
Notably, the case $\delta_0 = 0$, $s_0 > 1$ yields cat states, whereas $s_0 = 0$ and $|\delta_0| \sim 1$ yields cubic phase states (CPS). 
For $s_0 \geq 1$, the phase sensitivity parameter $s_0$ is non-decreasing under Gaussian maps (Sec.~\ref{sec:cparam_conversion}).}

    \label{fig:gps_landscape}
\end{figure*}
\begin{figure*}[tbp]
    \centering
    \input{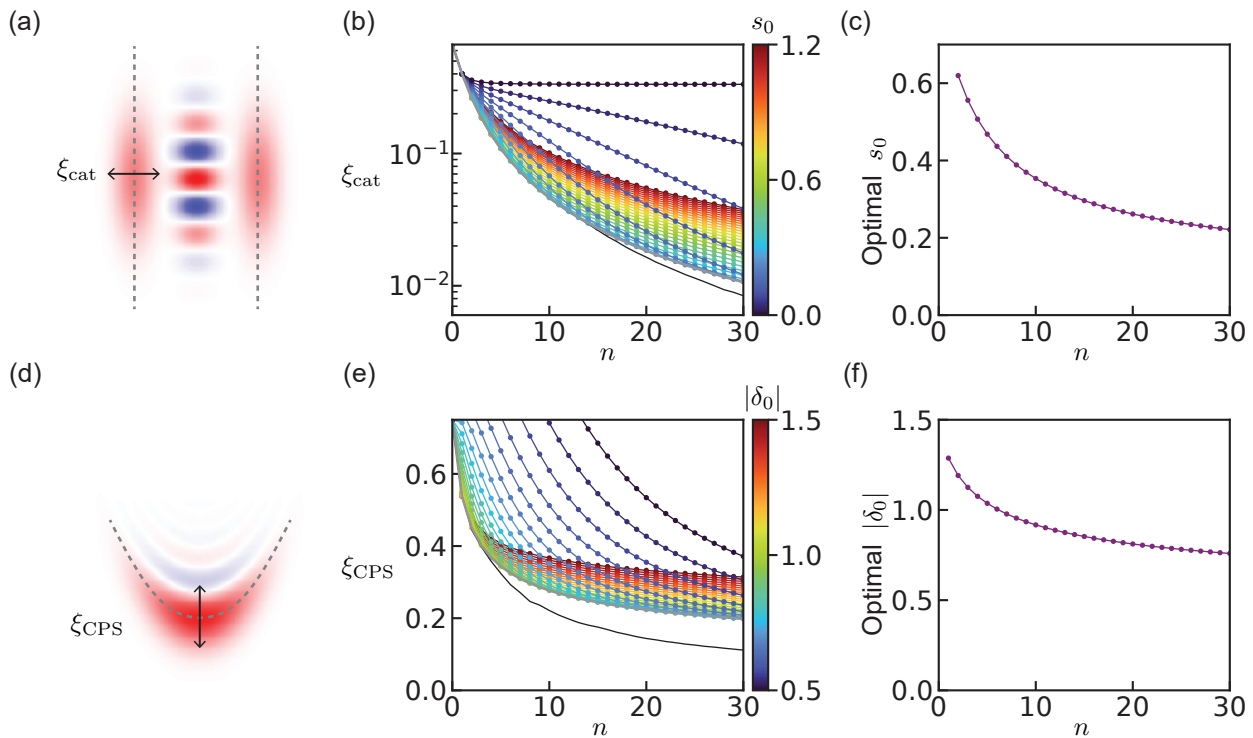}
    \caption{Characterization of cat states and CPS generated by two-mode non-Gaussian state generators. 
(a–c) Cat states ($s_0 > 0, \delta_0 = 0$). 
(a) Illustration of $x^2$-squeezing \cite{cat_nls}, $\xi_{\mathrm{cat}}$ (Eq.~\eqref{eq:cat_sqz_def}), used as a metric for cat states. This quantifies the degree of localization into two symmetric peaks in phase space. 
(b) $\xi_{\mathrm{cat}}$ of the state $\ket{\psi}^{(\mathrm{p})}_{s_0,0,n}$ as a function of measured photon number $n$, for different values of the non-Gaussian phase sensitivity $s_0$. Colored lines: different $s_0$; gray: optimized over $s_0$ for each $n$; black: theoretical lower bound imposed by $n$ \cite{cat_nls}. 
(c) Optimal $s_0$ yielding the minimum $\xi{\mathrm{cat}}$ for each $n$. Values for $n=0,1$ are not shown, as the state is independent of $s_0$ up to Gaussian unitaries.
(d–f) CPS ($s_0 = 0, |\delta_0| > 0$). 
(d) Illustration of cubic nonlinear squeezing \cite{cps_nls}, $\xi_{\mathrm{CPS}}$ (Eq.~\eqref{eq:cps_sqz_def}), used as a metric for CPS. This quantifies the degree of localization along a parabolic curve in phase space. 
(e) $\xi_{\mathrm{CPS}}$ of $\ket{\psi}^{(\mathrm{p})}_{0,\delta_0,n}$ as a function of $n$, for different values of the non-Gaussian asymmetry $\delta_0$. Colored lines: different $|\delta_0|$; gray: optimized over $|\delta_0|$; black: stellar-rank bound. 
(f) Optimal $|\delta_0|$ yielding the minimum $\xi{\mathrm{CPS}}$ for each $n$. The case $n=0$ is not shown, as the state is independent of $\delta_0$ up to Gaussian unitaries. }
\label{fig:gps_evaluate}
\end{figure*}
Non-Gaussian state generators can produce a wide variety of quantum states, whose properties may appear disparate when described in terms of Fock-state expansions or Wigner functions. A central unifying idea of this work is that these states can instead be characterized by a small set of continuous \emph{non-Gaussian control parameters}, denoted $(s_0,\delta_0)$. These parameters complement the discrete stellar-rank measure---defined as the maximum photon number in a finite Fock expansion followed by a Gaussian unitary, typically equal to the total number of detected photons~\cite{gbs_nongauss}. The non-Gaussian control parameters govern the symmetry and structure of the output states in phase space, providing a common language to compare distinct non-Gaussian resources.

To make their role concrete, we first describe the simplest yet general setting in which the control parameters naturally arise: a two-mode non-Gaussian state generator, shown in Fig.~\ref{fig:gps_landscape}(a). This serves as the fundamental building block for analyzing and optimizing the more general multi-mode non-Gaussian state generators introduced later. An entangled pure two-mode Gaussian state $\ket{G}$ is prepared. One mode, the \emph{control mode}, is measured in the photon-number basis, and conditional on the outcome $n$, a non-Gaussian state is probabilistically generated in the other mode, the \emph{signal mode}, through their entanglement. This setting already encompasses a wide range of known schemes, including photon subtraction~\cite{ps_cat_multi,ps_cat_multi_endo}, generalized photon subtraction~\cite{GPS,gps_exp}, photon addition~\cite{addition_exp,addition_theory}, coherent superpositions of subtraction and addition~\cite{sp_add_subtract,nonlinear_potential}, heralded Fock state generation from two-mode squeezed states~\cite{fock_exp,tatsuki_fock}, and CPS generation using displaced two-mode squeezed states~\cite{gkp,cps_konno}. 

In this setting, the non-Gaussian control parameters are defined as quantities that capture the structure of the output state and depend solely on the control mode, as formalized in the following proposition (restated and proven as Theorem~\ref{thm:particle_form} in Sec.~\ref{sec:c0_gps}).
\begin{proposition}\label{prop:particle_form}
The output state of any two-mode non-Gaussian state generator can be expressed in the following \emph{particle form}
\begin{align}
    \ket{\psi}_{s_0,\delta_0,n}
    &\gaussprop (\hat{a}^\dagger + s_0 \hat{a} + \delta_0)^n \ket{0}, \label{eq:particle_form_intro}
\end{align}
where $\hat a$ ($\hat a^\dagger$) denotes the annihilation (creation) operator satisfying $[\hat a,\hat a^\dagger]=1$, $\ket{0}$ is the corresponding vacuum state, and $\gaussprop$ denotes proportionality up to a Gaussian unitary operation. The measured photon number $n$ sets the \emph{stellar rank} \cite{stellar_rank,stellar_rank_original,core_state,stellar_rank_heterodyne,radim_atom_ng} of the output state. We introduce the parameters $s_0 \geq 0$ and $\delta_0 \in \mathbb{C}$, determined solely by the properties of the control mode, called \emph{non-Gaussian phase sensitivity} and \emph{non-Gaussian asymmetry}, respectively. Together, they are referred to as the \emph{non-Gaussian control parameters}.
\end{proposition}

By tuning these two parameters, even at a fixed stellar rank $n$, a surprisingly diverse class of non-Gaussian states can be engineered within this simple setup, as shown in Fig.~\ref{fig:gps_landscape}(b). The case $(s_0,\delta_0)=(0,0)$ yields a Fock state. Increasing $s_0$ gradually transforms it into a phase-sensitive state, while a nonzero $\delta_0$ biases the state, breaking the mirror symmetry in phase space. Notably, the regimes $s_0 > 1$ and $0 \leq s_0 < 1$ correspond to photon-subtracted and photon-added states, respectively (Sec.~\ref{sec:classification}). In addition, two important classes, approximations of Schrödinger cat states and cubic phase states (CPS), lie within this spectrum, providing intuition for the roles of the two parameters, as described below.
\subsection{Two distinctive examples: cat and cubic phase states}
The two control parameters each carry a distinct physical meaning. The parameter $s_0$ controls the balance between creation and annihilation operators, capturing phase-sensitive interference effects between separated components in phase space, as exemplified by cat states. By contrast, the parameter $\delta_0$ introduces an asymmetry that induces nonlinear structure in phase space, naturally leading to the cubic phase state (CPS) as an illustrative case. Examining these two states side by side reveals how the abstract parameters $(s_0,\delta_0)$ manifest in concrete, experimentally relevant non-Gaussian resources.

\subsubsection{Non-Gaussian phase sensitivity \texorpdfstring{$s_0$}{s₀}: Cat state}
When $s_0 > 1$ and $\delta_0 = 0$, the output state approximates a Schrödinger cat state \cite{cat_code,cat_state,cat_origin}
\begin{align}
    \ket{\mathrm{cat}_{\pm}(\alpha)} \propto \ket{\alpha} \pm \ket{-\alpha},\label{eq:cat_state}
\end{align}
with the correspondence
\begin{align}
    \ket{\psi}_{s_0,0,n} \gausssim \ket{\mathrm{cat}_{\mathrm{sgn}((-1)^n)}\qty(\sqrt{(n+1/2)/s_0})}.
\end{align}
(See Sec.~\ref{sec:classification} for the derivation.)  

The key feature of cat states is the coherent superposition of two well-separated, localized components in phase space, which produces the characteristic interference fringes in the Wigner function. This separation can be captured quantitatively by the so-called $x^2$-squeezing~\cite{cat_nls}, illustrated in Fig.~\ref{fig:gps_evaluate}(a), defined as
\begin{align}
    \xi_{\mathrm{cat}} = \min_{\lambda>0}\ev*{\qty(\tfrac{\hat{x}^2}{\lambda^2}-1)^2},\label{eq:cat_sqz_def}
\end{align}
where we define the quadrature operators as $\hat{x}=\hat{a}+\hat{a}^\dagger$ and $\hat{p}=-i\hat{a}+i\hat{a}^\dagger$ (see Sec.~\ref{sec:idler_rep_gps} for our convention). This measure captures the concentration of the quadrature distribution around two symmetric peaks in phase space. It is strictly positive and tends to zero for large-amplitude cat states, whereas Gaussian states satisfy the bound $\xi_{\mathrm{cat}} \geq 2/3$. For small $s_0$, the output resembles a Fock state, showing no visible double-peak structure, while for very large $s_0$ the peaks merge and coherence is lost. Between these extremes lies an optimal $s_0$, as shown in Figs.~\ref{fig:gps_evaluate}(b) and \ref{fig:gps_evaluate}(c), where the state exhibits maximum cat-like character and approaches the stellar-rank limit, nearly saturating the generator's ability to produce cat states.

\subsubsection{Non-Gaussian asymmetry \texorpdfstring{$\delta_0$}{δ₀}: Cubic phase state}

This example clarifies the role of non-Gaussian phase asymmetry $\delta_0$. When $s_0=0$ and $|\delta_0|>0$, the output state approximates a cubic phase state (CPS)~\cite{gkp,cps_nls,cps_konno,atsushi_cpg},
\begin{equation}
    \ket{\mathrm{CPS}} = e^{-i\hat{x}^3} \ket{p=0},\label{eq:cps}
\end{equation}
as derived in Sec.~\ref{sec:classification}.  

The hallmark of CPS is the cubic nonlinearity of its wavefunction. This feature can be quantified by the \emph{cubic nonlinear squeezing}~\cite{cps_nls}, illustrated in Fig.~\ref{fig:gps_evaluate}(d), defined as
\begin{align}
    \xi_{\mathrm{CPS}} = \min_{\lambda>0,d\in\mathbb{R}}\ev*{\qty(\lambda \hat{p}-\tfrac{\hat{x}^2}{4\sqrt{2}\lambda^2}-d)^2}.\label{eq:cps_sqz_def}
\end{align}
This quantity measures the extent to which the phase-space distribution is concentrated along a parabolic curve, reflecting the cubic structure of the state. For an ideal CPS, $\xi_{\mathrm{CPS}}$ approaches zero, while Gaussian states satisfy $\xi_{\mathrm{CPS}} \geq 3/4$~\cite{cps_nls,cps_konno}. For small $|\delta_0|$, the output resembles a symmetric, Fock-like state with little cubic character, whereas for very large $|\delta_0|$ the cubic coherence is washed out. Between these regimes lies an optimal $|\delta_0|$, highlighted in Figs.~\ref{fig:gps_evaluate}(e) and \ref{fig:gps_evaluate}(f), where the generated state exhibits its strongest cubic nonlinearity. 

Together, the cat and cubic phase states demonstrate how the control parameters $(s_0,\delta_0)$ map directly onto qualitatively distinct non-Gaussian resources. 
A more detailed classification of the states produced by two-mode generators is presented in Sec.~\ref{sec:classification}.

\subsection{Non-Gaussian control parameters as quantifier of non-Gaussianity}
The key insight from above examples is that the degree of non-Gaussianity is not determined by the stellar rank alone but depends crucially on the control parameters $(s_0,\delta_0)$. While all states share the same stellar rank $n$, their non-classical features vary widely depending on $(s_0,\delta_0)$. For instance, when $\delta_0=0$ and $s_0$ is too large, the cat state collapses to a trivial Gaussian, erasing non-Gaussianity. This shows that a state generator does not necessarily exploit the full non-Gaussian potential implied by its stellar rank $n$. We will explore the resource-theoretic role of this limitation in Sec.~\ref{sec:cparam_conversion}.

Building on this observation, we develop a general optimization framework. In certain regimes, by tailoring $(s_0,\delta_0)$ appropriately, one can approximate the same target state using fewer detected photons, thereby lowering the required stellar rank (Sec.~\ref{sec:approx_less_n}). Moreover, we identify transformations that preserve $(s_0,\delta_0)$ while modifying the success probability, enabling optimization of state generation rates without changing the output (Sec.~\ref{sec:damping_gps}). Generalizing these tools to multiple control modes, we propose in Sec.~\ref{sec:optimization} an algorithm that systematically optimizes non-Gaussian state generators, achieving nearly the same output states with reduced stellar rank and enhanced probability of success.
\section{Formal definition of control parameters in the two-mode non-Gaussian state generator}\label{sec:two-mode}
\begin{figure*}[tbp]
    \centering
    \input{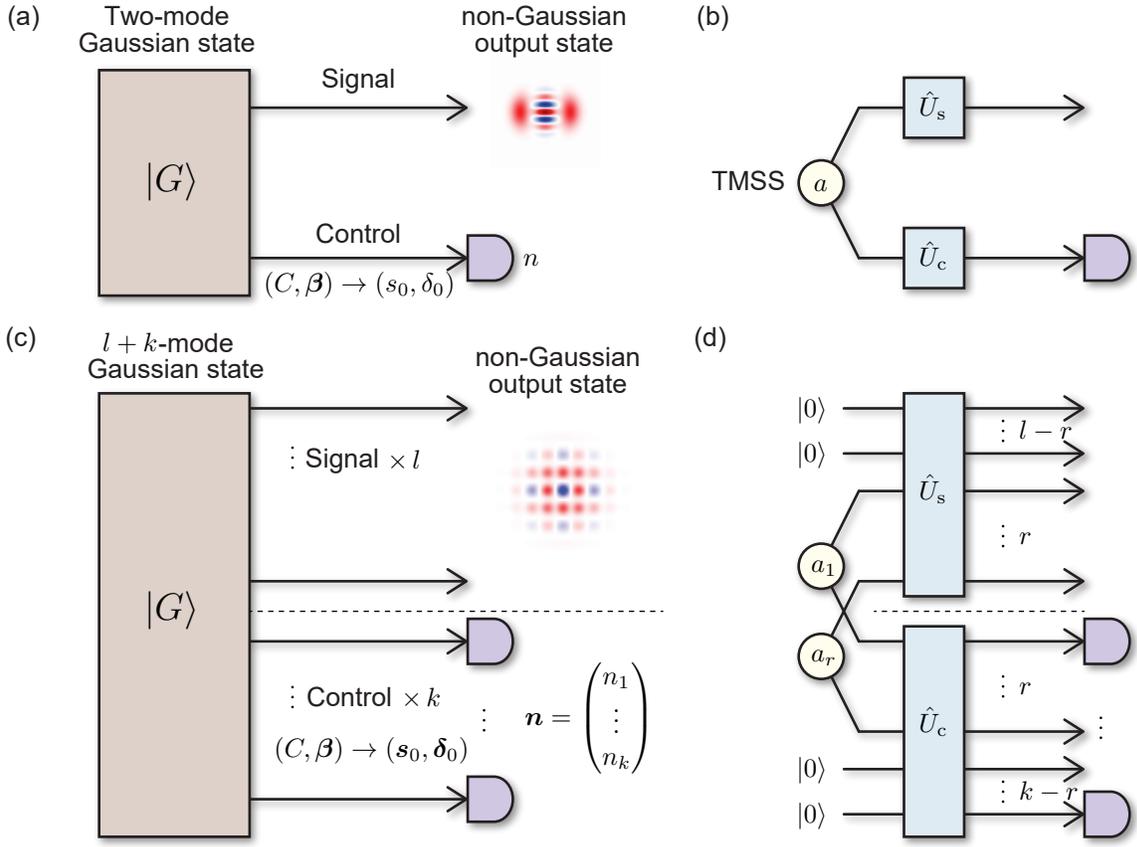}
    \caption{Two-mode and multi-mode non-Gaussian state generators and their canonical forms. 
(a) A two-mode non-Gaussian state generator, consisting of a control mode (measured in the photon-number basis) and a signal mode (where the non-Gaussian output state is generated). The output state and its success probability are determined solely by the control moments $(C,\bm\beta)$ and the measured photon number $n$ (Corollary~\ref{cor:idler_rep_gps}). 
(b) Canonical form of the two-mode generator, equivalent to the system in (a). The squeezing strength $a$ of the two-mode squeezed state and the Gaussian unitary $\hat{U}_\idler$ on the control mode depend only on the control moments $(C,\bm\beta)$ (Theorem~\ref{thm:purificationUniquenessGps}). 
(c) A general multi-mode non-Gaussian state generator with $k$ signal modes and $l$ control modes (Sec.~\ref{sec:multi-mode}). The $k$-mode output state and its success probability are determined solely by the control moments $(C,\bm\beta)$ and the measured photon numbers $\bm{n}$ (Corollary~\ref{cor:idler_rep}). 
(d) Canonical form of the multi-mode generator, equivalent to the system in (c). The parameters $a_1,\dots,a_k$ and the Gaussian unitary $\hat{U}_\idler$ depend only on the control moments $(C,\bm\beta)$ (Theorem~\ref{thm:purificationUniqueness}).}

\label{fig:gbs}
\end{figure*}
Having seen through examples that the non-Gaussian features of a state depend crucially on $(s_0,\delta_0)$, we now turn to a formal derivation, while still remaining in the simple two-mode case. For this, we introduce the \emph{control-mode representation}, which characterizes both the output state and its generation probability entirely in terms of the properties of the control mode. Within this framework, the parameters $s_0$ and $\delta_0$ naturally arise as a minimal parametrization of the output state. This two-mode case analysis provides the formal foundation for the multi-mode extension in Sec.~\ref{sec:multi-mode} and the optimization framework developed in Sec.~\ref{sec:optimization}.

\subsection{Control-mode representation}\label{sec:idler_rep_gps}
We start by giving a formal definition of a two-mode non-Gaussian state generator. A two-mode Gaussian state $\ket{G}$ is prepared, with modes denoted as the \emph{signal mode} and the \emph{control mode}. The control mode is subjected to a photon-number measurement, projecting it onto a Fock state $\ket{n}$ and thereby heralding a non-Gaussian state $\ket{\psi}\propto {}_{\idler}\!\braket{n}{G}$ in the signal mode, with success probability $p_n=\norm{{}_{\idler}\!\braket{n}{G}}^2$.

Let $\hat{\bm{q}}_{\signal}=(\hat{x}_{\signal},\hat{p}_{\signal})^T$ and $\hat{\bm{q}}_\idler=(\hat{x}_{\idler},\hat{p}_{\idler})^T$ denote the quadrature vectors of the signal and control modes, respectively. The full quadrature vector is then written as $\hat{\bm{q}}=(\hat{\bm{q}}_{\signal}^{T},\hat{\bm{q}}_{\idler}^{T})^{T}$.
Each mode satisfies the canonical commutation relation $\comm{\hat{x}}{\hat{p}}=2i$, corresponding to the convention $\hbar=2$, which will be used throughout the paper.

The Gaussian state $\ket{G}$ is completely characterized by the mean $\bm\gamma\in\mathbb{R}^{4}$ and covariance $\Sigma\in\mathbb{R}^{4\times 4}$, defined as \cite{gaussian_qi}:
\begin{align}
    \bm{\gamma}&\coloneqq \expval*{\hat{\bm q}}{G},\label{eq:mean_def}\\
    \Sigma &\coloneqq \tfrac{1}{2}\expval*{\acomm{\hat{\bm q}-\bm{\gamma}}{(\hat{\bm q}-\bm{\gamma})^{T}}}{G},\label{eq:cov_def}
\end{align}
where $\acomm{\cdot}{\cdot}$ denotes the anti-commutator. These can be written in block form as
\begin{align}
    \Sigma = \begin{pmatrix} A & B^T \\ B & C \end{pmatrix}, 
    \quad 
    \bm{\gamma} = \begin{pmatrix} \bm{\alpha} \\ \bm{\beta} \end{pmatrix},
    \label{eq:sigma_gps}
\end{align}
with $A,B,C\in\mathbb{R}^{2\times 2}$ and $\bm{\alpha},\bm{\beta}\in\mathbb{R}^{2}$.
The covariance matrix must satisfy the uncertainty relation
\begin{align}
    \Sigma \geq i\Omega^{\otimes 2},\label{eq:uncertainty_all_gps}
\end{align}
where
\begin{align}
    \Omega=\mqty(0&1\\-1&0)
\end{align}
is the symplectic form.

We now focus on the properties of the control mode, and denote the partial moments $(C,\bm\beta)$ as the \emph{control moments}. From Eq.~\eqref{eq:uncertainty_all_gps}, the control moments must also satisfy the uncertainty constraint
\begin{align}
    C\geq i\Omega.\label{eq:uncertainty_gps}
\end{align}
An important observation is that the state $\ket{G}$ can be regarded as a purification of the control mode. Corresponding to the well-known fact that a purification of a quantum state is unique up to unitaries acting on the auxiliary system~\cite{nielsen_chuang}, an analogous result holds for Gaussian states, where the unitary can be restricted to Gaussian unitaries~\cite{gaussian_qi,duan_simon_duan,duan_simon_simon}. Tailored to our setting, we have the following theorem.
\begin{restatable}[Canonical form]{theorem}{thmpurificationUniquenessGps}\label{thm:purificationUniquenessGps}
Any pure two-mode Gaussian state $\ket{G}$ can be expressed in the following \emph{canonical form} (Fig.~\ref{fig:gbs}(b)):
\begin{align}
    \ket{G} = \big(\hat{U}_\signal \otimes \hat{U}_\idler\big) \ket{\mathrm{TMSS}(a)},
\end{align}
where
\begin{align}
    \ket{\mathrm{TMSS}(a)} = \frac{2\sqrt{a}}{a+1}\sum_{j=0}^{\infty} \bigg(\frac{a-1}{a+1}\bigg)^{j} \ket{j}\ket{j}
\end{align}
is a two-mode squeezed state, and $\hat{U}_\signal$, $\hat{U}_\idler$ are Gaussian unitaries acting on the signal and control modes, respectively. The parameter $a>1$ depends only on $C$, and $\hat{U}_\idler$ can be chosen depending only on $(C,\bm{\beta})$.
\end{restatable}
\begin{proof}
See Appendix~\ref{sec:proof_Cmat}.
\end{proof}

A two-mode non-Gaussian state generator is specified by $\ket{G}$ together with the measurement outcome $n$, which sets the stellar rank of the output state. From Theorem~\ref{thm:purificationUniquenessGps}, $\ket{G}$ depends only on a Gaussian unitary $\hat{U}_\signal$ and the control moments $(C,\bm\beta)$. Hence, a two-mode non-Gaussian state generator is completely characterized by the set $(C,\bm{\beta},n,\hat{U}_\signal)$, as summarized below.
\begin{cor}[Control-mode representation]\label{cor:idler_rep_gps}
For a two-mode non-Gaussian state generator with measured photon number $n$ and control moments $(C,\bm\beta)$:
\begin{itemize}
    \item The output non-Gaussian state is determined by $(C,\bm{\beta},n)$, up to Gaussian unitaries.
    \item The success probability of the state generation is determined by $(C,\bm{\beta},n)$.
\end{itemize}
\end{cor}
\begin{proof}
The first statement follows from Theorem~\ref{thm:purificationUniquenessGps}. The second follows because the probability of measuring $n$ photons depends only on the control moments.
\end{proof}

Therefore, for a fixed stellar rank $n$ (i.e.\ a fixed measurement outcome $n$), the control moments $(C,\bm\beta)$ completely determine the non-Gaussian features of the output state---somewhat paradoxically, this information is contained entirely in the control-mode subsystem, even though the non-Gaussian state itself is generated in the signal mode. We will henceforth identify a state generator with this set of parameters and simply write ``non-Gaussian state generator $(C,\bm{\beta},n,\hat{U}_\signal)$,'' with the output state and the success probability denoted by $\ket{\psi}_{C,\bm{\beta},n,\hat{U}_\signal}$ and $p_n(C,\bm\beta)$, respectively. This representation in terms of the control moments will be called the \emph{control-mode representation}. Whenever the Gaussian unitary $\hat{U}_\signal$ is clear from the context or not of primary interest, we simply write the state as $\ket{\psi}_{C,\bm{\beta},n}$ and the generator as $(C,\bm\beta,n)$, ignoring the Gaussian degree of freedom. Moreover, for states $\ket{\psi}$ and $\ket{\psi'}$, we introduce the notations
\begin{align}
    \ket{\psi} &\gaussprop \ket{\psi'}, \\
    \ket{\psi} &\gausssimprop \ket{\psi'},
\end{align}
to denote that there exists a Gaussian unitary $U$ such that $\ket{\psi} \propto U\ket{\psi'}$, either exactly or approximately, respectively.
\subsection{Rotation and Damping}\label{sec:damping_gps}
Although the control-mode representation $(C, \bm{\beta}, n, \hat{U}_\signal)$ uniquely specifies the non-Gaussian output state, the corresponding state generator is not necessarily unique for a given output state. In fact, let us consider the phase-rotation operator
\begin{align}
    \hat{R}(\theta)=e^{-i\theta \hat{n}}
\end{align}
with $\theta \in \mathbb{R}$, which corresponds to a Gaussian unitary operation, and the damping operator
\begin{align}
    \hat{\Gamma}(\lambda)=e^{-\lambda \hat{n}}
\end{align}
with $\lambda \in \mathbb{R}$, which is a purity-preserving Gaussian filter \cite{gaussian_filter} also known as noiseless linear attenuation ($\lambda>0$) or amplification ($\lambda<0$) \cite{nl_amp,usugaNoisepoweredProbabilistic,zavattaHighfidelityNoiseless, kocsisHeraldedNoiseless,nl_attn,nl_amp_attn, marekCoherentstatePhase}. Since both operators are diagonal in the Fock basis, inserting either of these operators immediately before the photon-number measurement yields a new two-mode state generator that produces the same output state as the original state generator (Fig.~\ref{fig:opt_method_all}(a)). Importantly, this modification can be implemented in practice by appropriately adjusting the control moments $(C,\bm{\beta})$, without physically applying these operators in the laboratory. This observation leads to the following theorems.

\begin{theorem}[Rotation transformation]\label{thm:fock_equivalence_rot_gps}
For a two-mode non-Gaussian state generator $(C, \bm{\beta}, n,\hat{U}_\signal)$, the output state and the success probability remain invariant:
\begin{align}
    \ket{\psi}_{(C', \bm{\beta}', n,\hat{U}_\signal)}&\propto \ket{\psi}_{(C, \bm{\beta}, n,\hat{U}_\signal)},\\
    p_n(C', \bm{\beta}')&= p_n(C, \bm{\beta}),
\end{align}
under the transformation
\begin{align}
    C' &= O C O^T, \\
    \bm{\beta}' &= O \bm{\beta},
\end{align}
where $O \in SO(2)$.
\end{theorem}
\begin{proof}
This transformation of the control moments corresponds to inserting a rotation operator $\hat{R}(\theta)$ on the control mode. The theorem then follows from
\begin{align}
    \ket{\psi}_{(C', \bm{\beta}', n,\hat{U}_\signal)}
        &\propto {}_\idler\!\matrixel{n}{\hat{I}_\signal \otimes \hat{R}(\theta)_\idler}{G} \\
        &= e^{-i\theta n}\, {}_\idler\!\braket{n}{G} \\
        &\propto e^{-i\theta n}\ket{\psi}_{(C, \bm{\beta}, n,\hat{U}_\signal)}.
\end{align}
\end{proof}

\begin{restatable}[Damping transformation]{theorem}{thmfockEquivalenceDampGps}\label{thm:fock_equivalence_damp_gps}
Let $I$ denote the identity matrix.  For a two-mode non-Gaussian state generator $(C,\bm{\beta},n,\hat{U}_\signal)$ and any $t \in \mathbb{R}$ satisfying
\begin{align}
    t > 1 \quad \text{or} \quad tI < -C,
\end{align}
under the following \emph{damping transformation} of the control moments
\begin{align}
    C' &=\mathcal{D}_t(C)\coloneqq (tC + I)(C + tI)^{-1}, \label{eq:damping_trans_gps_C}\\
    \bm{\beta}' &=\mathcal{D}_t(\bm\beta)\coloneqq \sqrt{t^2 - 1}\,(C + tI)^{-1} \bm{\beta}, \label{eq:damping_trans_gps_beta}
\end{align}
there exists a Gaussian unitary $\hat{U}_\signal'$ such that
\begin{align}
    \ket{\psi}_{(C', \bm{\beta}', n,\hat{U}_\signal')}
    \propto \ket{\psi}_{(C, \bm{\beta}, n,\hat{U}_\signal)}.
\end{align}
\end{restatable}
\begin{proof}
This follows in the same way as Theorem~\ref{thm:fock_equivalence_rot_gps}, from the invariance of the output state under insertion of the damping operator $\hat{\Gamma}(\lambda)$ with $\lambda = \coth^{-1} t$. The complete proof is given in Appendix~\ref{sec:proof_fock_equivalence_damp}.
\end{proof}

Note that under the ``Cayley transform''
\begin{align}
    \tilde{C} &= (C+I)^{-1}(C-I), \\
    \tilde{\beta} &= (C+I)^{-1}\bm{\beta}, \\
    \tilde{t} &= (t+1)^{-1}(t-1),
\end{align}
Eqs.~\eqref{eq:damping_trans_gps_C} and \eqref{eq:damping_trans_gps_beta} simplify to
\begin{align}
    \tilde{C}' = \tilde{t}\,\tilde{C}, \qquad 
    \tilde{\beta}' = \sqrt{\tilde{t}}\,\tilde{\beta}.
\end{align}
(See Appendix~\ref{sec:cayley-transform} for the proof.) While the damping transformation preserves the output state, it generally modifies the success probability of state generation. This degree of freedom plays a crucial role in enhancing the success probability, as discussed in Sec.~\ref{sec:optimization}.
\subsection{Non-Gaussian control parameters \texorpdfstring{$s_0,\delta_0$}{s₀,δ₀}}\label{sec:c0_gps}
According to Corollary~\ref{cor:idler_rep_gps}, for a fixed measured photon number $n$, the output state of a two-mode non-Gaussian state generator can be fully characterized (up to Gaussian unitaries) by the control moments $(C,\bm{\beta})$, which together have five real degrees of freedom. By applying the rotation transformation (Theorem~\ref{thm:fock_equivalence_rot_gps}) and the damping transformation (Theorem~\ref{thm:fock_equivalence_damp_gps}), these degrees of freedom can be reduced to three parameters that uniquely characterize the non-Gaussian features of the output state.

Specifically, $C$ can always be diagonalized as
\begin{align}
    C = O^T \mathrm{diag}(c,d)\,O, \label{eq:diag_C}
\end{align}
with $c \geq d$ and $O \in SO(2)$. Using the rotation transformation (Theorem~\ref{thm:fock_equivalence_rot_gps}), $C$ can be brought to diagonal form $C = \mathrm{diag}(c,d)$.

Then, we find two invariant parameters under the damping transformation (Theorem~\ref{thm:fock_equivalence_damp_gps}), namely the \emph{non-Gaussian phase sensitivity}:
\begin{align}
s_0 = \frac{c-d}{cd-1}, \label{eq:reduced_cparam_c0}
\end{align}
and the \emph{non-Gaussian asymmetry}:
\begin{align}
\delta_0 = \frac{2}{\sqrt{cd-1}} \qty(\sqrt{\frac{d+1}{c+1}}\,\bar{\beta}_x - i \sqrt{\frac{c+1}{d+1}}\,\bar{\beta}_p), \label{eq:reduced_cparam_beta0}
\end{align}
where $(\bar{\beta}_x,\bar{\beta}_p)^T$ is defined by
\begin{align}
    \mqty(\bar{\beta}_x \\ \bar{\beta}_p) \coloneqq O \bm{\beta},
\end{align}
using the same $O$ as in Eq.~\eqref{eq:diag_C}.

These parameters $(s_0,\delta_0)$ provide a description of the output state equivalent to the control-mode representation $(C,\bm\beta)$, as summarized in the following theorem.
\begin{restatable}[Non-Gaussian control parameters]{theorem}{reducedCparam}\label{thm:reduced_cparam}
The output state of a two-mode non-Gaussian state generator $(C,\bm{\beta},n)$ is uniquely determined up to Gaussian unitary operations by the triplet $(s_0,\delta_0,n)$. Furthermore, two sets of control moments $(C,\bm{\beta})$ and $(C',\bm{\beta}')$ that yield the same $(s_0,\delta_0)$ can be transformed into each other via rotation and damping transformations.
\end{restatable}
\begin{proof}
See Appendix~\ref{sec:c0_proof}.
\end{proof}

Figure~\ref{fig:s0_cd} illustrates the relation between the eigenvalues $c \geq d$ of $C$ and the corresponding value of $s_0$, providing an intuitive visualization of this definition.
\begin{figure}[!htb]
    \input{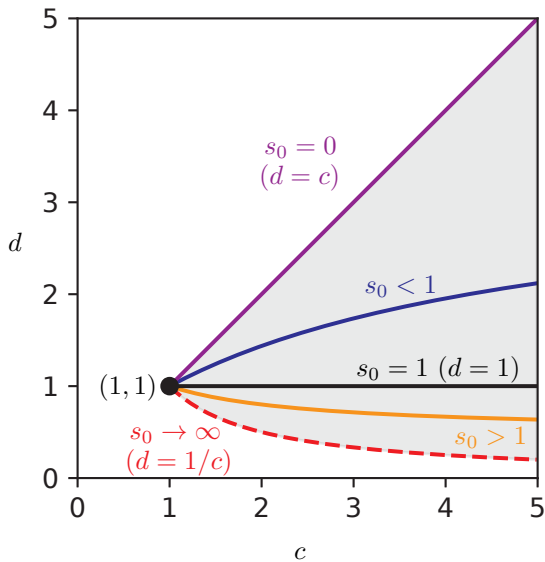}
    \caption{Relation between the eigenvalues $c \geq d$ of $C$ and the non-Gaussian phase sensitivity $s_0$ defined in Eq.~\eqref{eq:reduced_cparam_c0}. 
    Solid curves show contours of constant $s_0$, while the gray shaded region marks the physically allowed domain satisfying Eq.~\eqref{eq:uncertainty_gps} with $c \geq d$. 
    The point $(1,1)$ corresponds to the vacuum state in the control mode.}
    \label{fig:s0_cd}
\end{figure}

We call the pair $(s_0,\delta_0)$ the \emph{non-Gaussian control parameters}. These parameters fully characterize the output state up to a Gaussian unitary. When there is no risk of confusion, we further write $\ket{\psi}_{C,\bm\beta,n}$ simply as $\ket{\psi}_{s_0,\delta_0,n}$, again omitting the dependence on $\hat{U}_\signal$. The state can then be expressed explicitly in terms of $(s_0,\delta_0,n)$ in several equivalent forms.
\begin{restatable}[Particle form]{theorem}{particleForm}\label{thm:particle_form}
The output state of a two-mode non-Gaussian state generator with control parameters $(s_0,\delta_0)$ and measured photon number $n$ can be written as
\begin{align}
\ket{\psi}_{s_0,\delta_0, n} \gaussprop \ket{\psi}_{s_0, \delta_0, n}^{(\mathrm{p})} \coloneqq \left(\hat{a}^{\dagger} + s_0 \hat{a} + \delta_0 \right)^n \ket{0}.\label{eq:fock_form}
\end{align}
\end{restatable}
\begin{proof}
See Appendix~\ref{sec:proof_fock_wave_form}.
\end{proof}
\begin{restatable}[Wave form]{theorem}{waveForm}\label{thm:wave_form}
The output state of a two-mode non-Gaussian state generator with control parameters $(s_0,\delta_0)$ and measured photon number $n$ can be written as
\begin{align}
\begin{split}
\ket{\psi}_{s_0,\delta_0, n} \gaussprop & \ket{\psi}_{s_0, \delta_0, n}^{(\mathrm{w})} \\\coloneqq &
\exp\qty(-\frac{1}{2\sqrt{s_0 + 1}} \delta_{0x} \hat{p}) \\
& \cdot \exp\qty(\frac{\sqrt{s_0 + 1}}{2} \delta_{0p} \hat{x})\exp\qty(-\frac{s_0}{4} \hat{x}^2) \ket{n},
\end{split}
\label{eq:wave_form}
\end{align}
where $\ket{n} \propto \hat{a}^{\dagger n}\ket{0}$ is the Fock state.
\end{restatable}
\begin{proof}
See Appendix~\ref{sec:proof_fock_wave_form}.
\end{proof}
We refer to $\ket{\psi}_{s_0,\delta_0,n}^{(\mathrm{p})}$ and 
$\ket{\psi}_{s_0,\delta_0,n}^{(\mathrm{w})}$ as the \emph{particle form} 
and the \emph{wave form}, respectively. By construction, these two forms are related 
by a Gaussian unitary, which is detailed in Appendix \ref{sec:proof_fock_wave_form}. 
We remark that the relation between the particle form and the wave form is reminiscent of the wave--particle duality in quantum mechanics.
The particle form corresponds to the ``core state'' 
\cite{core_state,stellar_rank_original,stellar_rank}, a superposition of Fock 
states up to $n$ photons that explicitly displays its stellar rank $n$. This picture connects naturally to the intuition behind schemes such as photon subtraction~\cite{ps_cat_multi,ps_cat_multi_endo}, photon addition~\cite{addition_exp,addition_theory}, and coherent superpositions of them~\cite{sp_add_subtract,nonlinear_potential}, all of which fall within our general framework. Indeed, as we will see in Sec.~\ref{sec:classification}, the output states can be expressed as photon-subtracted, photon-added, or balanced superpositions of the two, applied to a Gaussian state, depending on the regime of $s_0$.

In contrast, the wave form represents the state as $\ket{n}$ acted on by Gaussian filters \cite{gaussian_filter} in the quadratures $\hat{x}$ and $\hat{p}$, producing a ``filtered'' wavefunction in the quadrature basis. This representation is 
particularly suited for analyzing coherent structures in phase space. 
Hence, wave-function engineering schemes~\cite{wf_engineer,fock_to_cat} and breeding protocols based on homodyne conditioning~\cite{cat_breeding_theory,gkp_breeding, gkp_kan} are naturally related to this picture and can be reduced to our framework by including the homodyne conditioning into the Gaussian state preparation (see Appendix~\ref{sec:opt_gkp_detail} for the example of GKP breeding \cite{gkp_breeding}). The wave form picture thus provides a natural interpretation of these schemes and also serves as the basis for the approximation of the output state in Sec.~\ref{sec:approx_less_n}.

Moreover, for $n \geq 2$, the parameters $(s_0,\delta_0)$ associated with a given output state are unique, up to the trivial transformation $\delta_0 \to -\delta_0$.

\begin{restatable}[Uniqueness of control parameters]{theorem}{gpsUniqueness}\label{thm:gps_uniqueness}
For $n \geq 2$, if $\ket{\psi}_{s_0,\delta_0,n} 
\gaussprop \ket{\psi}_{s_0',\delta_0',n}$,
then either $s_0 = s_0' = 0$ with $|\delta_0| = |\delta_0'|$, 
or $(s_0,\delta_0) = (s_0', \pm \delta_0')$.
\end{restatable}

\begin{proof}
    See Appendix \ref{sec:proof_uniqueness}.
\end{proof}

This shows that for $n \geq 2$, the control parameters $(s_0,\delta_0)$ can be regarded as intrinsic properties of the output states, even though they were originally defined through state generators. By contrast, when $n=1$, the particle form of the output state is
\begin{align}
    \ket{\psi}_{s_0,\delta_0,1}^{(\mathrm{p})}
    &\propto \ket{1} + \delta_0\ket{0} \\
    &\propto \hat{R}(\theta)\!\left(\ket{1} + |\delta_0|\ket{0}\right),
\end{align}
where $\hat{R}(\theta)$ is the phase-rotation operator. Thus, in this case the state is specified by a single real parameter $|\delta_0|$.
\subsection{Classification of the output state}\label{sec:classification}
We now classify the output states of two-mode non-Gaussian state generators according to the regime of the non-Gaussian control parameters $(s_0,\delta_0)$. Representative examples of non-Gaussian states for different values of $(s_0,\delta_0)$ are shown in Fig.~\ref{fig:gps_landscape}.
\subsubsection{\texorpdfstring{$s_0>1$}{s₀>1}: photon-subtracted state}
When $s_0>1$, there exists a Gaussian unitary $\hat{U}$ such that
\begin{align}
    \hat{U}^\dagger \hat{a} \hat{U} \propto \hat{a}^\dagger + s_0 \hat{a} + \delta_0.
\end{align}
Hence, from Theorem~\ref{thm:particle_form}, we obtain
\begin{align}
    \ket{\psi}_{s_0,\delta_0,n} \gaussprop \hat{a}^n \hat{U}\ket{0},
\end{align}
which shows that the state is, up to a Gaussian unitary, an $n$-photon-subtracted Gaussian state.

In particular, when $\delta_0=0$, the state reduces to a photon-subtracted squeezed vacuum state, which is well known to approximate the Schrödinger cat state
\begin{align}
    \ket{\mathrm{cat}_{\pm}(\alpha)} \propto \ket{\alpha} \pm \ket{-\alpha},
\end{align}
for large $n$~\cite{cat_state,GPS}. This can be seen from the wavefunction representation. For real $\alpha$, the $x$-wavefunction of the cat state $\ket{\mathrm{cat}_{\pm}(i\alpha)}$ is
\begin{align}
    \braket{x}{\mathrm{cat}_{+}(i\alpha)} &\propto \cos(\alpha x)e^{-x^2/4}, \\
    \braket{x}{\mathrm{cat}_{-}(i\alpha)} &\propto \sin(\alpha x)e^{-x^2/4}.
\end{align}
On the other hand, from Theorem~\ref{thm:wave_form}, the wavefunction of the generated state is
\begin{align}
    \braket{x}{\psi}_{s_0,0,n}^{(\mathrm{w})} \propto \phi_n(x)\, e^{-s_0 x^2/4},
\end{align}
where $\phi_n(x)$ denotes the wavefunction of the Fock state $\ket{n}$. 
Using the semiclassical approximation of $\phi_n(x)$ around $x=0$~\cite{hermite_approx},
\begin{align}
    \phi_n(x)\appropto \cos(\sqrt{n+1/2}x-\frac{n}{2}\pi),
\end{align}
we obtain
\begin{align}
    \ket{\psi}_{s_0,0,n}^{(\mathrm{w})} \sim \hat{S}\ket{\mathrm{cat}_{\mathrm{sgn}((-1)^n)}\qty(i\sqrt{\qty(n+1/2)/s_0})},\label{eq:cat_approx}
\end{align}
with the squeezing factor
\begin{align}
    \hat{S}^\dagger\mqty(\hat{x}\\\hat{p})\hat{S}=\mqty(1/\sqrt{s_0}&0\\0& \sqrt{s_0})\mqty(\hat{x}\\\hat{p}).\label{eq:cat_approx_sqz}
\end{align}
\subsubsection{\texorpdfstring{$0 \leq s_0 < 1$}{0 ≤ s₀ < 1}: photon-added state}
When $0\leq s_0<1$, there exists a Gaussian unitary $\hat{U}$ such that
\begin{align}
    \hat{U}^\dagger \hat{a}^\dagger \hat{U}=\hat{a}^\dagger +s_0 \hat{a}+\delta_0.
\end{align}
Hence, from Theorem \ref{thm:particle_form}, we obtain
\begin{align}
    \ket{\psi}_{s_0,\delta_0,n}\gaussprop \hat{a}^{\dagger n}\hat{U}\ket{0},
\end{align}
which shows that the state is, up to a Gaussian unitary, an $n$-photon-added Gaussian state. 
In particular, $s_0=\delta_0=0$ corresponds to the Fock state.

When $s_0=0$ and $|\delta_0|$ is sufficiently large, the state provides a good approximation of the cubic phase state (CPS)~\cite{gkp,cps_konno}, defined as
\begin{equation}
\ket{\mathrm{CPS}} = e^{-i \hat{x}^3} \ket{p=0},
\end{equation}
as illustrated in Fig.~\ref{fig:gps_landscape}(b). As detailed in Appendix~\ref{sec:cps_derive}, the asymptotic relation
\begin{align}
    \ket{\psi}_{0,\delta_{0},n}^{(\mathrm{w})}
    \appropto \hat{D}(\tfrac{i}{2}p_0)\hat{S}(\tfrac{1}{3}\ln\gamma)\ket{\mathrm{CPS}}
\end{align}
holds for large $n$, where
\begin{align}
    p_0 = 2\sqrt{n+\tfrac{1}{2}},\gamma = \frac{1}{24\sqrt{n+\tfrac{1}{2}}}.
\end{align}

\subsubsection{\texorpdfstring{$s_0=1$}{s₀=1}: critical state}
At $s_0 = 1$, the state lies on the boundary between photon addition and subtraction:
\begin{align}
    \ket{\psi}_{1,\delta_0,n}^{(\mathrm{p})} \propto (\hat{x} + \delta_0)^n \ket{0}.
\end{align}
Its wavefunction is given by
\begin{align}
    \braket{x}{\psi_{1,\delta_0,n}^{(\mathrm{p})}} \propto (x + \delta_0)^n e^{-x^2/4}.
\end{align}
In particular, when $\delta_0 = 0$, this state still provides a good approximation of a cat state for large $n$, while achieving a larger amplitude than the photon-subtracted cases with $s_0 > 1$:
\begin{align}
    \ket{\psi}_{1,0,n}^{(\mathrm{w})} 
    \sim \ket{\mathrm{cat}_{\mathrm{sgn}((-1)^n)}\qty(i\sqrt{n+1/2})}. 
    \label{eq:cat_approx_critical}
\end{align}

Ref.~\cite{GPS} derived the condition $s_0 = 1$ for the generation of large-amplitude cat states, referring to it as the \emph{critical condition}, in the special case of two orthogonally squeezed vacuum states mixed by a phase-preserving beamsplitter. They showed that Eq.~\eqref{eq:cat_approx_critical} holds asymptotically for large $n$. Ref.~\cite{xanadu_architecture} later referred to this as the ``monomial condition.'' As a consequence of Theorem~\ref{thm:reduced_cparam}, we find that this condition is not restricted to that special case but in fact valid for any two-mode non-Gaussian state generator. In practice, however, Figure~\ref{fig:gps_evaluate}(b) shows the corresponding optimal $s_0$, which is slightly shifted towards smaller values compared to the critical condition $s_0=1$.

\section{Non-Gaussian control parameters as a resource}\label{sec:cparam_conversion}
\begin{figure*}[tbp]
    \centering
    \input{figures/optimize_methods_all_annotated_labels_overlay.tex}
    \caption{Principles for optimizing non-Gaussian state generators. 
(a) Increasing the success probability of a multimode state generator without changing its output state (Sec.~\ref{sec:damping_gps} for the two-mode case and Sec.~\ref{sec:multimode_damping} for the multimode case). 
A non-Gaussian state generator $(C,\bm{\beta},\bm{n})$ can be transformed into another generator $(C',\bm{\beta}',\bm{n})$ by \emph{virtually} applying rotation and damping operators (Theorems~\ref{thm:fock_equivalence_rot_gps}, \ref{thm:fock_equivalence_damp_gps}, and \ref{thm:fock_equivalence_damp}). These transformations leave the output state invariant while potentially increasing the success probability.
(b) Reducing the measured photon number for a two-mode state generator while approximately preserving the output state (Sec.~\ref{sec:approx_less_n}). 
In the wave-form representation (Theorem~\ref{thm:wave_form}), the output wavefunction $\braket{x}{\psi}_{s_0,\delta_0,n}^{\mathrm{(w)}}$ (iv) is given by the Fock-state wavefunction $\phi_n(x)$ (ii) multiplied by a Gaussian envelope (i) determined by $(s_0,\delta_0)$. 
To approximate this with a lower photon number $n'$, the Fock-state wavefunction $\phi_n(x)$ is replaced by $\phi_{n'}(kx-d)$ (iii), which provides a good approximation near the center of the envelope. 
By rescaling the envelope through suitable choice of $s_0'$ and $\delta_0'$, the resulting wavefunction $\braket{x}{\psi}_{s_0',\delta_0',n'}^{\mathrm{(w)}}$ (v) closely approximates the original wavefunction (iv) up to extra Gaussian unitaries.
(c) Reducing the measured photon number for a multi-mode state generator. 
For each control mode $m$ with photon number $n_m$ and partial control moment $(C_m,\bm\beta_m)$, a corresponding two-mode generator $(C_m,\bm\beta_m,n_m)$ can be extracted as a subsystem (Theorem~\ref{thm:extract_cat}). 
By iteratively approximating each subsystem using the two-mode technique and applying the corresponding Gaussian filter $\hat{M}_m$ (Eq.~\eqref{eq:Mm}), an approximation of the entire multi-mode state generator is obtained.}
\label{fig:opt_method_all}
\end{figure*}

In this section, we show that the non-Gaussian phase sensitivity $s_0$, introduced in the previous section, serves as a non-Gaussian resource for two reasons.

First, in the regime $s_0 > 1$, we show that $s_0$ is non-decreasing under Gaussian maps. This monotonicity is a fundamental property in resource theory~\cite{nongaussian_resource}. Second, a state with large $s_0 > 1$ can be approximated by a state with smaller $s_0$ and a smaller photon number $n$. This demonstrates that $s_0$ quantifies the reduction in non-Gaussianity relative to the amount expected from the stellar rank of the state.

The following subsections elaborate on these two points.
\subsection{Monotonicity of \texorpdfstring{$s_0$}{s₀}}
To establish $s_0$ as a proper non-Gaussian resource, it must satisfy the monotonicity property~\cite{nongaussian_resource}: namely, $s_0$ cannot be reduced solely through Gaussian operations. In this subsection, we demonstrate that $s_0$ indeed satisfies this property in the regime $s_0 > 1$.

We derive a condition under which $\ket{\psi}_{s_0,\delta_0,n}$ can be converted into $\ket{\psi}_{s_0',\delta_0',n}$ using only Gaussian maps. Such a transformation is impossible with Gaussian unitaries alone, as shown in Theorem~\ref{thm:gps_uniqueness}, however becomes possible when allowing more general Gaussian maps.

For example, from Theorem~\ref{thm:wave_form} we have
\begin{align}
    e^{-a\hat{x}^2}\ket{\psi}_{s_0,0,n}^{(\mathrm{w})} \propto \ket{\psi}_{s_0',0,n}^{(\mathrm{w})},
\end{align}
with
\begin{align}
    s_0' = s_0 + 4a.
\end{align}
Here the operator $e^{-a\hat{x}^2}$ with $a>0$ corresponds to a physically realizable Gaussian filter \cite{gaussian_filter}, which can be implemented, for example, by mixing with a vacuum mode on a beamsplitter followed by homodyne conditioning. This shows that, for $\delta_0 = 0$, it is always possible to implement a transformation that increases $s_0$.

The necessary and sufficient conditions for general convertibility are summarized in the following theorem.
\begin{restatable}{theorem}{gpsIrreversible}\label{thm:gps_irreversible}
For $n\geq 2$, a state $\ket{\psi}_{s_0,\delta_0,n}$ can be (probabilistically) transformed into $\ket{\psi}_{s_0',\delta_0',n}$ by a Gaussian map, if and only if at least one of the following conditions is satisfied:
\begin{align}
\begin{cases}
0 \leq s_0 < 1, \\
s_0' > s_0, \\
s_0 = s_0' = 1 \ \text{ and } \ \delta_{0x} = 0.
\end{cases}
\end{align}
\end{restatable}
\begin{proof}
See Appendix~\ref{sec:proof_conversion_cond}.
\end{proof}
This structure of convertibility is depicted in Fig.~\ref{fig:gps_landscape}. The region $0\leq s_0< 1$ includes the Fock state $s_0=\delta_0 = 0$, and within this region, all states are mutually convertible. In contrast, when $s_0 \geq 1$, $s_0$ can only be converted to a state with an equal or larger $s_0$. In other words, $s_0$ is monotonic with respect to Gaussian maps in the regime $s_0 \leq 1$. This irreversible structure characterizes the role of $s_0$ as a non-Gaussian resource.

\subsection{Approximation with lower stellar rank}\label{sec:approx_less_n}

When $s_0\geq 1$ and $\delta_0 = 0$, the state closely approximates the cat state, as shown in Sec.~\ref{sec:classification}. The amplitude of the cat state is an indicator of non-Gaussianity, as it determines the number of interference fringes in the Wigner function. However, from Eq.~\eqref{eq:cat_approx}, this amplitude depends on both the measured photon number $n$ and $s_0$. This implies that the non-Gaussianity of the state cannot be fully captured by $n$ alone, which represents the stellar rank \cite{stellar_rank} of the state. To accurately quantify non-Gaussianity, both $s_0$ and $n$ must be taken into account.

The following proposition states that a state with large $s_0$ can be effectively approximated by a state with a smaller $s_0$ and a smaller photon number $n$.

\begin{proposition}\label{prop:c0_approx}
    For given non-Gaussian control parameters $(s_0,\delta_0)$, measured photon number $n$, and any $n' < n$, if $s_0$ is sufficiently large, then there exist parameters $s_0' < s_0$ and $\delta_0'$ such that
    \begin{align}
        \ket{\psi}_{s_0',\delta_0',n'} \gausssimprop \ket{\psi}_{s_0,\delta_0,n},
    \end{align}
    that is, $\ket{\psi}_{s_0',\delta_0',n'}$ provides a good approximation of $\ket{\psi}_{s_0,\delta_0,n}$ up to a Gaussian unitary.
\end{proposition}

Although Proposition \ref{prop:c0_approx} does not provide an exact bound on the precision of the approximation, here we provide evidences to support this proposition, analytically in the case of $\delta_0 = 0$, and numerically in the case of $\delta_0\neq 0$.

First we consider the cases without displacement ($\delta_0=0$). From the approximation with the cat state Eq.~\eqref{eq:cat_approx}, when $n$ and $n'$ share the same parity (i.e., $n - n'$ is even), using
\begin{align}
s_0' = \frac{2n'+1}{2n+1} s_0,
\end{align}
we have the approximation:
\begin{align}
    \ket{\psi}_{s_0,0,n}^{(\mathrm{w})} \sim \hat{S}\ket{\psi}_{s_0',0,n'}^{(\mathrm{w})},\label{eq:approx_cat_less_n}
\end{align}
where the squeezing factor is:
\begin{align}
    \hat{S}^\dagger\mqty(\hat{x}\\\hat{p})\hat{S} = \mqty(\sqrt{\frac{s_0'}{s_0}} & 0\\0 & \sqrt{\frac{s_0}{s_0'}})\mqty(\hat{x}\\\hat{p}).\label{eq:approx_cat_less_n_sqz}
\end{align}
This result indicates that a state with large $s_0$ can be approximated by smaller photon number by reducing $s_0$.

Using the wave form representation (Theorem~\ref{thm:wave_form}), this approximation extends to states with displacement ($\delta_0\neq 0$). Figure~\ref{fig:opt_method_all}(b) illustrates this mechanism. From the wave form, we have:
\begin{align}
    \begin{split}
    \ket{\psi}_{s_0, \delta_0, n}^{(\mathrm{w})} \propto 
    &\exp\qty(-\frac{1}{2\sqrt{s_0+1}}\delta_{0x} \hat{p})\\
    &\cdot\exp\left[-\frac{s_0}{4}\qty(\hat{x} - \frac{\sqrt{s_0+1}}{s_0}\delta_{0p})^2\right] \ket{n}.
    \end{split}\label{eq:wave_form_centered}
\end{align}
In the $x$-representation, the operator:
\begin{align}
\exp[-\frac{s_0}{4}\qty(\hat{x} - \frac{\sqrt{s_0+1}}{s_0}\delta_{0p})^2]
\end{align}
acts as a filter, effectively ``cutting out'' the wavefunction around
\begin{align}
    x = x_0 = \frac{\sqrt{s_0+1}}{s_0}\delta_{0p}.
\end{align}

Around this point, a semiclassical approximation \cite{wkb} yields
\begin{align}
    \phi_n(x) \appropto \phi_{n'}(kx - d), \label{eq:wavefunc_approximation}
\end{align}
where the parameters $k$ and $d$ are chosen so that the local behavior of $\phi_n(x)$ near $x_0$ is matched by $\phi_{n'}$. A heuristic construction of such $k$ and $d$ for arbitrary $x_0$ is given explicitly in Appendix~\ref{sec:approx_with_disp}. Using these values, we define the Gaussian unitary corresponding to the coordinate transformation as
\begin{align}
    \hat{U}_{s_0',\delta_0',n'\to s_0,\delta_0,n}
    = \hat{S}(\log k)\,\hat{D}(d), \label{eq:approx_unitary_def}
\end{align}
where $\hat{S}(\zeta)=\exp[\tfrac{1}{2}(\zeta^* \hat{a}^2 - \zeta \hat{a}^{\dagger 2})]$ and $\hat{D}(\alpha)=\exp(\alpha \hat{a}^\dagger - \alpha^* \hat{a})$ denote the squeezing and displacement operators, respectively.

After rescaling the envelope, we obtain the approximation
\begin{align}
    \ket{\psi}_{s_0,\delta_0,n}^{(\mathrm{w})} \appropto \hat{U}_{s_0',\delta_0',n'\to s_0,\delta_0,n}\ket{\psi}_{s_0',\delta_0',n'}^{(\mathrm{w})},\label{eq:approx_wave}
\end{align}
with transformed control parameters
\begin{align}
    s_0' &= \frac{1}{k^2} s_0,\\
    \delta_{0}' &=\sqrt{\frac{s_0+1}{s_0+k^2}}(k^2\delta_{0x}+i\delta_{0p})-i\frac{s_0 d}{k\sqrt{s_0+k^2}}.
\end{align}

In Appendix~\ref{sec:approx_with_disp}, we outline a general method for determining $k$ and $d$. In the specific case where $\delta_0 = 0$ and $n = n' \mod 2$, they can be analytically derived as:
\begin{align}
    k = \sqrt{\frac{2n+1}{2n'+1}}, \quad d = 0,
\end{align}
corresponding to the approximation Eq.~\eqref{eq:approx_cat_less_n}.

This approximation can also be expressed as
\begin{align}
    {}_\idler\!\braket{n}{G}\appropto{}_\idler\!\bra{n'}\qty(\hat{I}\otimes \hat{M}_{C,\bm{\beta},n\to s_0',\delta_0',n'})\ket{G},\label{eq:idler_operator}
\end{align}
using a Gaussian filter $\hat{M}_{C,\bm{\beta},n\to s_{0}',\delta_{0}',n'}$, whose detailed expression is derived in Appendix~\ref{sec:approx_with_disp}. The corresponding transformation of the control moments can be computed via the \choijam{} isomorphism (Appendix~\ref{sec:choi_mat}), which we denote concisely as
\begin{align}
(C,\bm{\beta})
\;\mapsto\;
\hat{M}_{C,\bm{\beta},n\to s_{0}',\delta_{0}',n'}(C,\bm{\beta}),
\end{align}
with a slight abuse of notation. This indicates that the necessary change of the control moments to achieve the approximation can be conveniently computed by virtually applying the Gaussian filter to the control mode. This representation will be useful in Sec.~\ref{sec:multi-mode}, where we discuss the general cases of multiple control modes.
\section{Multi-Mode Non-Gaussian State Generators}\label{sec:multi-mode}

Thus far, we have focused on two-mode non-Gaussian state generators. In this section, we generalize our framework to \emph{multi-mode} non-Gaussian state generators including multiple photon number measurements, as illustrated in Fig.~\ref{fig:gbs}(c). This setup is sometimes referred to as Gaussian Boson Sampling (GBS)~\cite{gbs} in the context of computational complexity, where the system exhibits exponential classical simulation cost.

In this scenario, an $(l + k)$-mode pure Gaussian state $\ket{G}$ is prepared, where the last $k$ modes (the control modes) are projected onto Fock states $\bigotimes_m \ket{n_m}$. This measurement probabilistically generates a non-Gaussian state in the first $l$ modes (the signal modes), whose stellar rank is at most $\sum_{m=1}^{k} n_m$ \cite{gbs_nongauss}. This scheme can produce a much broader range of non-Gaussian states compared to two-mode generators, enabling both multi-mode output states and single-mode states with more complex phase-space structures. An important single-mode example is the scheme for GKP state generation~\cite{xanadu_gkp,gaussian_breeding}, which will be discussed later in Sec.~\ref{sec:optimization}.

The control-mode representation and the non-Gaussian control parameters introduced in Sec.~\ref{sec:two-mode} extend naturally to this multi-mode setting, enabling systematic characterization and optimization of more complex non-Gaussian state generators.
\subsection{Control-mode representation}\label{sec:idler_rep_multi}
In the setup of Fig.~\ref{fig:gbs}(c), we define $\bm{q}_{\signal}=(x_{1\signal},p_{1\signal},\dots,x_{l\signal},p_{l\signal})^T$,$\bm{q}_\idler=(x_{1\idler},p_{1\idler},\dots,x_{k\idler},p_{k\idler})^T$ as the quadrature-operator vectors of the signal and control modes, respectively, and set $\bm{q} = (\bm{q}_\signal^T, \bm{q}_\idler^T)^T$. The mean vector $\bm\gamma \in \mathbb{R}^{2(l+k)}$ and covariance matrix 
$\Sigma \in \mathbb{R}^{2(l+k)\times 2(l+k)}$ are then defined as in Eqs.~\eqref{eq:mean_def} and~\eqref{eq:cov_def}. They admit the block decomposition
\begin{align}
    \Sigma = \begin{pmatrix} A & B^T \\ B & C \end{pmatrix}, 
    \qquad
    \bm{\gamma} = \begin{pmatrix} \bm{\alpha} \\ \bm{\beta} \end{pmatrix},
    \label{eq:sigma}
\end{align}
where $A \in \mathbb{R}^{2l\times 2l}$, $B \in \mathbb{R}^{2k\times 2l}$, $C \in \mathbb{R}^{2k\times 2k}$, and $\bm{\alpha}\in\mathbb{R}^{2l}$, $\bm{\beta}\in\mathbb{R}^{2k}$. Here, $C$ needs to satisfy the uncertainty relation
\begin{align}
    C\geq i\Omega^{\otimes k}.\label{eq:uncertainty}
\end{align}

Similarly to the case of two-mode non-Gaussian state generators, we can define the control-mode representation of this non-Gaussian state generator using the following theorem.
\begin{restatable}{theorem}{thmpurificationUniqueness}\label{thm:purificationUniqueness}
An arbitrary pure $l+k$-mode Gaussian state $\ket{G}$ can be expressed as the following \emph{canonical form}
\begin{align}
    \begin{split}
    \ket{G}&=\qty(\hat{U}_\signal\otimes \hat{U}_\idler)\\&\cdot\qty[\qty(\bigotimes_{m=1}^r\ket{\mathrm{TMSS}(a_m)}_{m,l+m}) \otimes \ket{0}_{\signal}^{\otimes l-r}\otimes \ket{0}_{\idler}^{\otimes k-r}],
    \end{split}\label{eq:standard_form}
\end{align}
where $a_1\geq\dots\geq a_r>1,\dots,1$ are the symplectic eigenvalues \cite{williamson_decomp} of $C$, which are unique for a given $C$ (Fig.~\ref{fig:gbs}(d)). $r$ is called Schmidt rank of $C$. Here,
\begin{align}
    \ket{\mathrm{TMSS}(a)}_{mm'}=\frac{2\sqrt{a}}{a+1}\sum_{j=0}^{\infty}\qty(\frac{a-1}{a+1})^{j}\ket{j}_m\ket{j}_{m'}\label{eq:tmss_def}
\end{align}
is a two-mode squeezed state, and $\hat{U}_\signal$ and $\hat{U}_\idler$ are Gaussian unitary operators acting on the signal modes and the control modes, respectively. $\hat{U}_\idler$ can be chosen depending only on $C,\bm\beta$.
\end{restatable}
\begin{proof}
    See Appendix \ref{sec:proof_Cmat}.
\end{proof}

Theorem \ref{thm:purificationUniqueness} allows us to represent the non-Gaussian state generator by the control-mode representation $(C, \bm{\beta}, \bm{n},\hat{U}_\signal)$, similarly to the two-mode case. We have the following corollary, which is the multi-mode analogue of Cor.~\ref{cor:idler_rep_gps}.
\begin{cor}\label{cor:idler_rep}
For a multi-mode non-Gaussian state generator $(C, \bm{\beta}, \bm{n}, \hat{U}_\signal)$:
\begin{itemize}
    \item The output non-Gaussian state is determined by $(C,\bm{\beta},\bm{n})$, up to Gaussian unitary operations.
    \item The success probability of the state generation is determined by $(C,\bm{\beta},\bm{n})$, which we denote as $p_{\bm{n}}(C,\bm\beta)$.
\end{itemize}
\end{cor}
\begin{proof}
    It follows from Theorem \ref{thm:purificationUniqueness}.
\end{proof}
Note that the number of output modes $l$ is not solely determined by $(C, \bm{\beta}, \bm{n})$, but the states with the same $(C, \bm{\beta}, \bm{n})$ and different $l$ are convertible by adding vacuum ancillary modes and applying a Gaussian unitary operation.
\subsection{Damping transformation}\label{sec:multimode_damping}
As shown by Theorem~\ref{thm:fock_equivalence_damp_gps}, inserting the damping operator Eq.~\eqref{eq:damping_trans} before the photon-number measurement does not change the output state but alters the probability of success. By generalizing this to multiple control modes, it can be seen that inserting the damping operators to all modes (Fig.~\ref{fig:opt_method_all}(a)) does not change the output state, only changing the success probability. The explicit form of the transformation of the state generator is given by the following theorem.

\begin{restatable}{theorem}{thmfockEquivalenceDamp}\label{thm:fock_equivalence_damp}
For a real vector $\bm{t}=(t_1,\dots,t_k)^T$, the following \emph{damping transformation} of the control-mode representation
\begin{align}
\begin{split}
    (C,\bm{\beta})\to &\mathcal{D}_{\bm{t}}(C,\bm{\beta}) \\
    :=(&T-\sqrt{T^2-1}(C+T)^{-1}\sqrt{T^2-1},\\
    &\sqrt{T^2-1}(C+T)^{-1}\bm{\beta})
\end{split}\label{eq:damping_trans}
\end{align}
does not change the output state up to a unitary Gaussian operation, while changing the success probability $p_{\bm{n}}(C,\bm\beta)$.

Here, $T$ is a diagonal matrix with paired entries
\begin{align}
    T=\mathrm{diag}(t_1,t_1,\dots,t_k,t_k).\label{eq:tdef}
\end{align}
\end{restatable}
The domain of $\bm{t}$ is given in Appendix \ref{sec:proof_fock_equivalence_damp}.
\begin{proof}
It suffices to show that the damping operation transforms $(C,\bm\beta)$ according to Eq.~\eqref{eq:damping_trans}. The detailed proof including the condition on $\bm{t}$ is given in Appendix \ref{sec:proof_fock_equivalence_damp}.
\end{proof}

Note that, under the Cayley transform
\begin{align}
\tilde{C}&=(C+I)^{-1}(C-I),\label{eq:cayley_1}\\
\tilde{\beta}&=(C+I)^{-1}\beta,\\
\tilde{T}&=(T+I)^{-1}(T-I),\label{eq:cayley_3}
\end{align}
Eq.~\eqref{eq:damping_trans} can be simplified as
\begin{align}
\mathcal{D}_{\bm{t}}(C,\bm\beta)&=(C',\bm{\beta}'),\label{eq:damping_cayley_1}\\
\tilde{C}'=\sqrt{\tilde{T}}&\tilde{C}\sqrt{\tilde{T}},\tilde{\beta}'=\sqrt{\tilde{T}}\tilde{\beta}.\label{eq:damping_cayley_2}
\end{align}
(See Appendix~\ref{sec:cayley-transform} for the proof.)

In Sec.~\ref{sec:optimization}, we use this degree of freedom to maximize the success probability for given photon numbers $\bm{n}$.
\subsection{Non-Gaussian control parameters}\label{subsec:control_multi}
Finally, we extend the definition of non-Gaussian control parameters, previously defined for two-mode state generators in Sec.\ref{sec:c0_gps}, to general multi-mode cases. Let $C_m, \bm\beta_m$ be the partial matrix and vector of $C,\bm\beta$, corresponding to the $m$-th control mode:
\begin{align}
    C_{m}&=(C_{ij})_{i=2m-1,2m,j=2m-1,2m},\label{eq:partial_cparam_C}\\
    \bm{\beta}_m&=(\bm{\beta}_{i})_{i=2m-1,2m}.\label{eq:partial_cparam_beta}
\end{align}
The non-Gaussian control parameters are defined as vectors $\bm{s}_{0}\in \mathbb{R}^{l}$ and $\bm{\delta}_{0}\in \mathbb{C}^{l}$, whose $m$-th components $s_{0m}$ and $\delta_{0m}$ are defined similarly to the two-mode case(Eqs.~\eqref{eq:reduced_cparam_c0} and \eqref{eq:reduced_cparam_beta0}):
\begin{align}
    s_{0m} &= \frac{d_m - c_m}{c_m d_m - 1},\\
    \delta_0 &= \frac{2}{\sqrt{c_m d_m-1}} \qty(\sqrt{\frac{d_m+1}{c_m+1}}\,\bar{\beta}_{mx} - i \sqrt{\frac{c_m+1}{d_m+1}}\,\bar{\beta}_{mp}). \label{eq:partial_reduced_cparam_beta0}
\end{align}

Here, $c_m \geq d_m$ are the eigenvalues of $C_m$, and we define $\bar{\beta}_{xm}$ and $\bar{\beta}_{pm}$ as
\begin{equation}
    O\bm\beta_m = \mqty(\bar{\beta}_{xm} \\ \bar{\beta}_{pm}),
\end{equation}
using the orthogonal matrix $O$ diagonalizing $C_m$ as
\begin{equation}
    C_m = O^T\mqty(c_m & 0 \\ 0 & d_m)O.
\end{equation}

Note that this definition of non-Gaussian control parameters is not invariant under the damping transformation. When invariance is required, one can instead adopt the \emph{invariant non-Gaussian control parameters} $\bm{\tilde{s}}_0, \bm{\tilde{\delta}}_0$, defined in Appendix~\ref{sec:invariant_s0}.

The following theorem is key to reducing measured photon numbers in multi-mode cases.
\begin{theorem}\label{thm:extract_cat}
The output non-Gaussian state of a multi-mode non-Gaussian state generator can be written as
\begin{equation}
    \left(\bigotimes_{m' \neq m}\bra{n_{m'}}_{m'}\right)\hat{U}\left(\bigotimes_{m' \neq m}\ket{0}_{m'}\right) \otimes \ket{\psi}_{c_{0m}, \beta_{0m}, n_m}^{(\mathrm{p})}\label{eq:extract_cat}
\end{equation}
as illustrated in Fig.~\ref{fig:opt_method_all}(c), where $\hat{U}$ is a Gaussian unitary operator acting only on the modes other than the $m$-th mode.
\end{theorem}
\begin{proof}
Let $\ket{G}$ be the Gaussian state for the multi-mode non-Gaussian state generator. By interpreting only the $m$-th control mode as the control mode and the remaining modes as signal modes and applying Theorem~\ref{thm:purificationUniqueness}, we obtain
\begin{align}
    \ket{G}=\qty(\hat{U}'\otimes \hat{U}_m) \otimes \qty(\ket{0}^{\otimes k+l-2}\otimes \ket{\mathrm{TMSS}(a_m)} )
\end{align}
where $\hat{U}_m$ is a Gaussian unitary operator acting only on the $m$-th control mode, and $\hat{U}'$ acts on the other modes. Since
\begin{align}
    \bra{n_m}\hat{U}_m\ket{\mathrm{TMSS}(a_m)}\propto \hat{U}_\signal\ket{\psi}_{c_{0m}, \beta_{0m}, n_m}^{(\mathrm{p})}
\end{align}
for some $\hat{U}_\signal$, by defining $\hat{U}=\hat{U}'\hat{U}_\signal$, Eq.~\eqref{eq:extract_cat} can be obtained.
\end{proof}

Theorem~\ref{thm:extract_cat} states that the two-mode state generator with non-Gaussian phase sensitivity parameters $s_{0m},\delta_{0m}$ can be extracted as a subsystem of the total state generator. Thus, applying the approximation of $\ket{\psi}_{s_{0m},\delta_{0m},n_m}^{(\mathrm{p})}$ obtained in Sec.~\ref{sec:approx_less_n}, an approximation of $\ket{\psi}_{C,\bm\beta,\bm{n},\hat{U}_\signal}$ can be obtained. Note that the closeness of the approximation of $\ket{\psi}_{s_{0m},\delta_{0m},n_m}^{(\mathrm{p})}$ does not guarantee the closeness of $\ket{\psi}_{C,\bm\beta,\bm{n},\hat{U}_\signal}$, as the remaining system involves projections to Fock states, which are non-unitary processes. Nevertheless, this approximation is practically useful, as demonstrated in the examples in Sec.~\ref{sec:optimization}.

When the approximation for the $m$-th control mode is applied, the corresponding transformation of the system can be described by a (virtual) application of the Gaussian filter
\begin{equation}
    \hat{M}_m=\hat{M}_{C_m, \bm\beta_m, n_m \to s_{0m}', \delta_{0m}', n_m'}\label{eq:Mm}
\end{equation}
defined in Eq.~\eqref{eq:idler_operator} to the $m$-th control mode,whose exact form can be computed via the \choijam{} isomorphism (Appendix~\ref{sec:choi_mat}). Again, note that this change can be experimentally implemented simply by modifying the initial Gaussian state $\ket{G}$ prior to the photon-number measurement, without actually performing the Gaussian filter operation \cite{gaussian_filter}.
By sequencially applying this ``approximation filter'' for all the control modes, we can obtain a new non-Gaussian state generator with reduced measured photon numbers $\bm{n}'$, but approximating the original output state. Note that the the result of this process may depend on the order of the control modes in which the approximations are applied.
\section{Optimization of non-Gaussian state generators}\label{sec:optimization}
As an application of our formalism based on the non-Gaussian control parameters $(s_0,\delta_0)$, we propose a method to find a non-Gaussian state generator that produces approximately the same output state while requiring detection of fewer photons and achieving a higher success probability. Our method applies to general multi-mode non-Gaussian state generators as in Fig.~\ref{fig:gbs}(c).
The optimization procedure is summarized in Fig.~\ref{fig:algorithm}, and its key principles are illustrated in Fig.~\ref{fig:opt_method_all}.

\begin{figure}[!htb]
  \centering
  \hrule
  \vspace{0.5ex}
  \begin{algorithmic}[0]
    \Require{\\Initial state generator $(C,\bm{\beta},\bm{n},\hat{U}_{\mathrm{s}})$.\\Target photon numbers $\bm{n}'$.}
    \Ensure{\\Optimized state generator $(C',\bm{\beta}',\bm{n}',\hat{U}'_{\mathrm{s}})$}
    \Statex

    \State \textbf{Step 1: Reduce detected photon numbers}
    \For{$m = 1$ \textbf{to} $k$}
      \State Get $(s_{0m}, \delta_{0m})$ \Comment{Sec.~\ref{subsec:control_multi}}
      \State Choose $(s_{0m}', \delta_{0m}')$
      \State \hspace{1em} s.t.\ $\ket{\psi}^{(\mathrm{w})}_{s_{0m},\delta_{0m},n_m}
        \gausssim \ket{\psi}^{(\mathrm{w})}_{s_{0m}',\delta_{0m}',n_m'}$ \Comment{Sec.~\ref{sec:approx_less_n}}
      \State $(C,\bm{\beta}) \gets
        \hat{M}_{C_m,\bm{\beta}_m,\,n_m \to s_{0m}',\,\delta_{0m}',\,n_m'}(C,\bm{\beta})$ \Comment{Eq.~\eqref{eq:idler_operator}}
    \EndFor
    \Statex

    \State \textbf{Step 2: Maximize success probability}
    \State Find $\bm{t}^* \gets \arg\max_{\bm{t}} p_{\bm{n}}\qty(\mathcal{D}_{\bm{t}}(C,\bm{\beta}))$ \Comment{Eq.~\eqref{eq:damping_trans}}
    \State $(C',\bm{\beta}') \gets \mathcal{D}_{\bm{t}^*}(C,\bm{\beta})$
    \Statex

    \State \textbf{Return} $(C',\bm{\beta}',\bm{n}',\hat{U}'_{\mathrm{s}})$
  \end{algorithmic}
  \vspace{0.5ex}
  \hrule
  \caption{Algorithm for optimizing non-Gaussian state generators.}
  \label{fig:algorithm}
\end{figure}

The algorithm takes as input an arbitrary non-Gaussian state generator $(C,\bm\beta,\bm{n},\hat{U}_\signal)$ of the form in Fig.~\ref{fig:gbs}(c) together with target photon numbers $\bm{n}'$ satisfying $n_m'\leq n_m, m=1,\dots,k$. It then outputs a new state generator $(C',\bm\beta',\bm{n}',\hat{U}'_\signal)$ with the target photon numbers and an improved success probability. The output state of the new generator closely approximates that of the original.

Our protocol consists of two parts: the reduction of detected photon numbers, based on the approximation in Sec.~\ref{sec:approx_less_n}, and the maximization of success probability, based on the damping operation in Sec.~\ref{sec:multimode_damping}, as described below. Importantly, these procedure does not require a full simulation of the state generation process. Instead, it only involves manipulating the control moments, particularly the submatrix associated with a single control mode. This makes our approach significantly more tractable than conventional brute-force methods that explore the full parameter space.
\subsection{Reduction of measured photon numbers}
The first step reduces the measured photon numbers from $\bm{n}$ to the target $\bm{n}'$, while preserving the output state as faithfully as possible. 
For each control mode $m = 1,\dots,k$, we first compute its non-Gaussian control parameters $(s_{0m},\delta_{0m})$ as defined in Sec.~\ref{subsec:control_multi}. 
By Theorem~\ref{thm:extract_cat}, a two-mode non-Gaussian state generator with parameters $(s_{0m},\delta_{0m})$ can be extracted as a subsystem, as shown in Fig.~\ref{fig:opt_method_all}(c). 

Given this subsystem and the target photon number $n_m' \leq n_m$ for the corresponding control mode, we choose new parameters $(s_{0m}',\delta_{0m}')$ such that the resulting wavefunction approximates the original one up to a Gaussian unitary,:
\begin{align}
\ket{\psi}_{s_{0m},\delta_{0m},n_m}
\gausssimprop
\ket{\psi}_{s_{0m}',\delta_{0m}',n_m'},
\end{align}
using the correspondence established in Sec.~\ref{sec:approx_less_n} (see Fig.~\ref{fig:opt_method_all}(b)). 
The corresponding two-mode generator is then replaced by its reduced counterpart with $(s_{0m}',\delta_{0m}',n_m')$, which amounts to updating the control moments via Eq.~\eqref{eq:idler_operator}, thereby modifying $(C,\bm{\beta})$. The new control moment is calculated via the transformation
\begin{align}
    (C,\bm\beta)\mapsto\hat{M}_m (C,\bm\beta),
\end{align} 
defined in Eq.~\eqref{eq:Mm}.
Iterating this procedure over all control modes yields a reduced state generator $(C'',\bm{\beta}'',\bm{n}',\hat{U}''_\signal)$ whose output state approximates that of the original generator but requires fewer detected photons.

\subsection{Maximization of success probability}
The second step maximizes the success probability $p_{\bm{n}'}$ of detecting the reduced photon number $\bm{n}'$. 
This is accomplished using the multimode damping transformation $\mathcal{D}_{\bm{t}}$ (Sec.~\ref{sec:multimode_damping}, Fig.~\ref{fig:opt_method_all}(a)). 
We optimize over the damping parameters $\bm{t}$ to find
\begin{align}
\bm{t}^* = \arg\max_{\bm{t}} \; p_{\bm{n}'}\!\left(\mathcal{D}_{\bm{t}}(C'',\bm{\beta}'')\right),
\end{align}
and apply the optimal transformation, yielding updated control moments
\begin{align}
(C',\bm{\beta}') = \mathcal{D}_{\bm{t}^*}(C'',\bm{\beta}'').
\end{align}

The resulting generator $(C',\bm{\beta}',\bm{n}',\hat{U}'_\signal)$ produces exactly the same output state as $(C'',\bm{\beta}'',\bm{n}',\hat{U}''_\signal)$, but with strictly higher success probability. 

Thus, the algorithm faithfully preserves the target state while making the generation process experimentally more efficient.

\subsection{Examples}\label{sec:optimization_example}
In this section, we demonstrate the performance of our optimization algorithm on several representative and practically important non-Gaussian state generation tasks. The results are summarized in Fig.~\ref{fig:optimize_results}. The key metrics include the measured photon number $\bm{n}$ and the success probability $p_{\bm{n}}$ before and after optimization, the corresponding change in the non-Gaussian control parameters $(s_0,\delta_0)$, and the fidelity between the initial and optimized output states. For cat states, CPS, and GKP states, we also evaluate tailored metrics collectively referred to as the non-Gaussian squeezing $\xi$ \cite{cat_nls,cps_nls,gkp_squeezing}. Numerical simulations of the state generation circuits were carried out using the open-source software MrMustard~\cite{mrmustard}.

In all cases, we observe a reduction of the measured photon number and an enhancement of the success probability, while maintaining a high fidelity ($>90\%$) and keeping the non-Gaussian squeezing nearly invariant. The details of the construction of the examples, the definitions of the metrics, and the optimization procedure are provided in Appendix~\ref{sec:opt_detail}, which also includes the intermediate state generator obtained after the first step of photon-number reduction. Notably, as a general trend, we find that a noticeable enhancement of the success probability already occurs in this first step, almost without increasing the required squeezing. This behavior can be attributed to the exponentially decaying nature of thephoton-number distribution of Gaussian states \cite{fiurasekMaximumHeralding}.

\begin{figure*}[tbp]
  \centering
  \renewcommand{\arraystretch}{1.25}
    \begin{tabular}{c c c c c c c c c}
        \toprule
         & \makecell{State\\generator} & \multicolumn{2}{c}{\makecell{Output\\(Wigner)}} & \makecell{Photon number\\$\bm{n}$} & \makecell{Probability\\$p_{\bm{n}}$} & \makecell{Control\\parameters\\$(s_0,\delta_0)$} & \makecell{Non-Gaussian\\squeezing\\$\xi$} &\makecell{Fidelity\\$F$} \\
        \midrule

        \multirow{4}{*}{\makecell{Cat\\(odd)}}
        & \multirow{4}{*}{\genfig{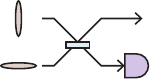}}
        & Bef. 
        & \statefig{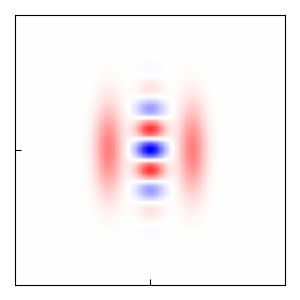}
        & $15$ & $\eformat{1.77e-6}$ & $(3.12,0)$ & $0.158$ & \multirow{4}{*}{\makecell{$0.9986$}} \\
        &   
        & Aft.
        & \statefig{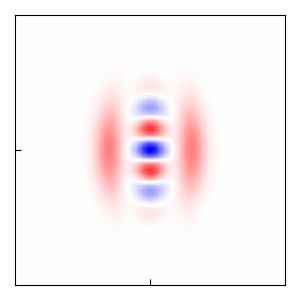}
        & $5$ & $\eformat{4.58e-2}$ & $(1.11,0)$ & $0.165$ &  \\
        \midrule
        \multirow{4}{*}{\makecell{Cat\\(even)}}
        & \multirow{4}{*}{\genfig{figures/optimization/cat_system.pdf}}
        & Bef. 
        & \statefig{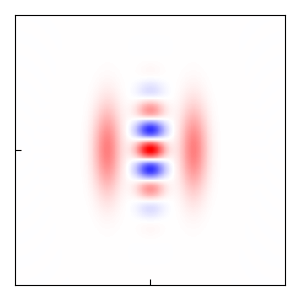}
        & $16$ & $\eformat{8.29e-7}$ & $(3.12,0)$ & $0.151$ & \multirow{4}{*}{\makecell{$0.9983$}} \\
        &   
        & Aft.
        & \statefig{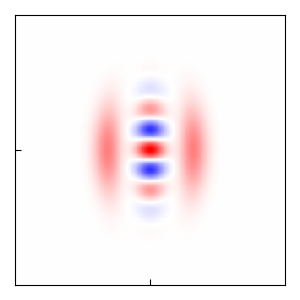}
        & $6$ & $\eformat{3.84e-2}$ & $(1.23,0)$ & $0.155$ &  \\
        \midrule
        \multirow{4}{*}{\makecell{CPS}}
        & \multirow{4}{*}{%
  \begin{minipage}[c]{0.17\textwidth}\centering
    \input{figures/optimization/cps_system_annotated_labels_overlay.tex}
  \end{minipage}
}
        & Bef. 
        & \statefig{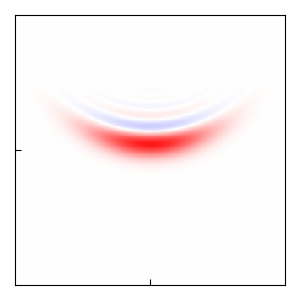}
        & $20$ & $\eformat{2.19e-8}$ & $(0,1.41i)$ & $0.315$ & \multirow{4}{*}{\makecell{$0.9964$}} \\
        &   
        & Aft.
        & \statefig{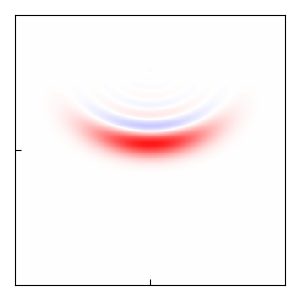}
        & $7$ & $\eformat{7.43e-2}$ & $(0,1.19i)$ & $0.330$ &  \\
        \midrule
        \multirow{4}{*}{\makecell{GKP}}
        & \multirow{4}{*}{\genfig{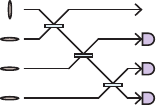}}
        & Bef. 
        & \statefig{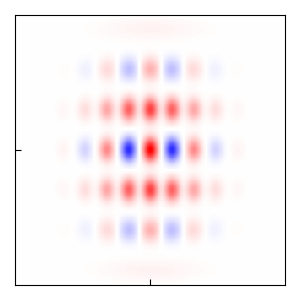}
        & $(18,18,18)$ & $\eformat{1.75e-12}$ & $(3.37,0)$ & $0.429$ & \multirow{4}{*}{\makecell{$0.993$}} \\
        &   
        & Aft.
        & \statefig{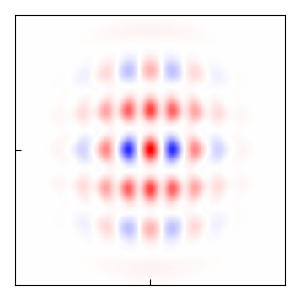}
        & $(6,6,6)$ & $\eformat{1.44e-4}$ & $(1.14,0)$ & $0.426$ &  \\
        \midrule
        \multirow{4}{*}{\makecell{Random}}
        & \multirow{4}{*}{%
  \begin{minipage}[c]{0.17\textwidth}\centering
    \input{figures/optimization/random_system_annotated_labels_overlay.tex}
  \end{minipage}
}
        & Bef. 
        & \statefig{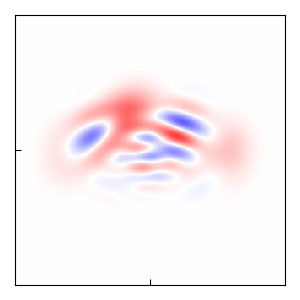}
        & $(14,14,14,14)$ & $\eformat{1.79e-30}$ & \textendash & \textendash & \multirow{4}{*}{\makecell{$0.911$}} \\
        &   
        & Aft.
        & \statefig{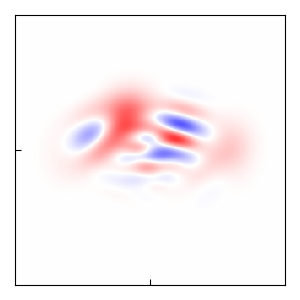}
        & $(9,14,2,12)$ & $\eformat{4.50e-6}$ & \textendash & \textendash &  \\

        \bottomrule
    \end{tabular}
    \caption{Examples of optimizing non-Gaussian state generators using the proposed algorithm (``Bef.'' = before optimization; ``Aft.'' = after optimization). 
For the GKP case, the non-Gaussian control parameters $(s_0,\delta_0)$ are shown only for one mode, since all control modes take the same value due to system symmetry. 
To quantify the quality of the states before and after optimization, suitable metrics are chosen for each case (except the random example): $x^2$-squeezing \cite{cat_nls} for cat states, cubic nonlinear squeezing \cite{cps_nls} for CPS, and GKP nonlinear squeezing \cite{gkp_squeezing} for GKP states. 
These values are shown in the column labeled “Non-Gaussian squeezing $\xi$.” 
The fidelity indicates the overlap between the initial and optimized states. 
See Appendix~\ref{sec:opt_detail} for details of the system construction as well as the parameters and performance.}

    \label{fig:optimize_results}
\end{figure*}
\subsubsection{Cat State}

Cat states, defined in Eq.~\eqref{eq:cat_state}, are widely used in quantum metrology and quantum error correction owing to their non-classical features. Experimentally, they are commonly generated either by photon subtraction from squeezed vacuum states~\cite{ps_cat_multi,mamoru_four_photon} or, more recently, through the generalized photon subtraction (GPS) scheme using two-mode squeezed states~\cite{gps_exp}. Both approaches fall within the class of schemes addressed by our method. In the following, we consider a GPS scheme as a starting point and examine whether further improvements can be achieved.

As in Eq.~\eqref{eq:cat_approx}, different parity of the cat state is obtained depending on the parity of the detected photon number. Both odd and even cases are examined. As an input to the algorithm, we take a GPS system requiring 15 or 16-photon detection to generate a cat states. After optimization, we obtain a new configuration needing only 5 or 6-photon detection, achieving more than $10^4$-fold enhancement of probability, with the fidelity $>99.8\%$ for both cases. The $x^2$-squeezing \cite{cat_nls} defined in Eq.\eqref{eq:cat_sqz_def} is used as a tailored metric for evaluating the quality of the cat states, and only minor degradation is observed. Further details are provided in Appendix~\ref{sec:opt_detail}. It is worth noting that in the first step of optimization, as well as the reduction of photon number, the success probability already improves by $10^2$, although the required squeezing level of the initial Gaussian state went down from (\SI{5.00}{dB}, \SI{5.00}{dB}) to (\SI{3.46}{dB}, \SI{-4.22}{dB}), taking the odd-cat case as an example (see Appendix~\ref{sec:opt_detail}, Fig.~\ref{fig:cat_opt}).

These improvements are attributed to the reduction in the non-Gaussian phase sensitivity parameter $s_0$, getting closer to the optimal value in Fig.~\ref{fig:gps_evaluate}(c).

\subsubsection{Cubic Phase State}
CPS, defined in Eq.~\eqref{eq:cps}, is a crucial resource for universal continuous-variable quantum computation. We start from the CPS generation protocol proposed by Gottesman et al.~\cite{gkp}, which uses a displaced two-mode squeezed state. With our algorithm, the photon number is reduced from 20 to 7, as well as enhancing the probability by $10^6$, while maintaining the fidelity $>99\%$. The cubic nonlinear squeezing \cite{cps_nls} defined in Eq.\eqref{eq:cps_sqz_def} is used as a tailored metric for evaluating the quality of the CPS, and only minor degradation is observed. Further details are provided in Appendix~\ref{sec:opt_detail}.

This enhancement is attributed to the reduction of $\delta_0$, getting closer to the optimal value in Fig.~\ref{fig:gps_evaluate}(f).

\subsubsection{GKP State}

The Gottesman-Kitaev-Preskill (GKP) state~\cite{gkp}, essential for bosonic quantum error correction, is defined as a superposition of displaced squeezed states:
\begin{equation}
\ket{\mathrm{GKP}} = \sum_k \hat{D}(k\Delta) \hat{S}(r) \ket{0}.
\end{equation}
It can be generated via the cat breeding protocol~\cite{cat_breeding_theory,cat_breeding_exp}, which combines multiple cat states using beamsplitters and homodyne conditioning.

We start from a Gaussian breeding protocol that is also employed by the recent experimental demonstration \cite{xanadu_gkp}, which can be obtained as a equivalent circuit of cat breeding protocol. The detail of this reduction is found in Appendix~\ref{sec:opt_detail}.

Applying our algorithm to this system, the photon number is reduced three times, and the probability enhances by $10^8$, while maintaining the fidelity $>99\%$. GKP non-Gaussian squeezing \cite{gkp_squeezing} is used as a tailored metric for evaluating the quality of the GKP state, and only minor degradation is observed (see Appendix~\ref{sec:opt_detail} for the definition).

This improvement can be attributed to the reduction of $s_{0m}$. A more detailed discussion of the trade-off between the control parameters and the photon number is provided in Appendix~\ref{sec:opt_detail}.

\subsubsection{Random state}
To illustrate the general applicability of our method, we also optimize a randomly constructed non-Gaussian state generator. 
The details of the construction are provided in Appendix~\ref{sec:opt_detail}. 
In this case as well, we observe a significant reduction in the required photon numbers together with an enhancement in the success probability, while maintaining the fidelity $>91\%$.

\section{Conclusion \& Discussion}

In summary, we have developed a framework to tackle the challenge of multi-mode non-Gaussian state generation, where brute-force exploration of the exponentially large parameter space is infeasible and traditional benchmarks such as stellar rank capture only part of the picture. To move beyond these limitations, we introduced the \emph{non-Gaussian control parameters} $(s_0,\delta_0)$, which reveal hidden inefficiencies in resource use and provide both a continuous description of non-Gaussianity beyond stellar rank and a practical handle for systematic optimization.

Using this framework, we constructed an optimization method that achieves striking gains in feasibility. Photon-number requirements are reduced by up to a factor of three—depending on the quality of the initial state generator and potentially exceeding this value for less efficient cases—while heralding probabilities increase by several orders of magnitude across key resource states including Schrödinger cat states, CPS, and GKP states. The approach also generalizes to randomly generated targets, underscoring its universality. In particular, our framework provides a practical pathway to improving state-of-the-art experiments, such as the recent demonstration of GKP state generation~\cite{xanadu_gkp}, by reducing photon-number requirements and enhancing success probabilities in a tractable way. Taken together, these results establish non-Gaussian control parameters as a versatile tool for diagnosing and enhancing non-Gaussian state generators, paving the way for more resource-efficient experiments and scalable optical quantum technologies capable of supporting fault-tolerant computation.

Looking ahead, two natural extensions of this work appear particularly promising. First, generalizing the framework to account for errors such as loss, and thereby to mixed states, is essential for realistic experimental considerations on the path toward resource-efficient and fault-tolerant optical quantum computing. In Appendix~\ref{sec:mixed_state}, we provide a preliminary analysis in this direction, showing that our scheme extends naturally to scenarios involving mixed Gaussian states.  

Second, optimizing non-Gaussian state generators is not only experimentally valuable but also closely related to questions of computational complexity~\cite{supremacy_review,boson_sampling,gbs, chabaudResourcesBosonic}. While simulating their full quantum behavior is widely believed to be classically intractable, as exemplified by Gaussian boson sampling~\cite{gbs}, our ability to optimize such generators efficiently highlights an important distinction between simulation and optimization. This perspective offers new insight into the relationship between physical resources and computational hardness, helping to clarify which features of a system drive intractability and where the boundary lies between classically simulable regimes and those that exhibit genuine quantum advantage.

Finally, we note that there are several complementary approaches sharing the goal of systematically optimizing non-Gaussian state generation~\cite{gbs_nongauss,aralovPhotonCatalysis}, as well as refining stellar-rank-based characterizations of non-Gaussianity~\cite{approx_stellar,meleSymplecticRank}. Exploring the fundamental connections between these approaches and our framework based on non-Gaussian control parameters represents an interesting direction for future research.

\begin{acknowledgments}
This work was partly supported by JST [Moonshot R\&D][Grant No.~JPMJMS2064], JSPS KAKENHI (Grant No.~23KJ0498), UTokyo Foundation, and donations from Nichia Corporation. 
PM acknowledges support of the Czech Science Foundation (project 25-17472S), and European Union's HORIZON Research and Innovation Actions under Grant Agreement no. 101080173 (CLUSTEC). P.M. also acknowledges the Programme Johannes Amos Comenius under the Ministry of Education, Youth and Sports of the Czech Republic reg. no. CZ.02.01.01/00/22\_008/0004649 (QUEENTEC).
R.F. acknowledges funding from the project 21-13265X of the Czech Science Foundation. 
F.H. acknowledges the support of Steven Touzard, his team, and CQT during the final preparation of the manuscript, and thanks Ni-Ni Huang for helpful feedback on the clarity and organization of the manuscript.
\end{acknowledgments}
\bibliography{main.bib}

@article{gbs_nongauss,
  title = {Conversion of Gaussian states to non-Gaussian states using photon-number-resolving detectors},
  author = {Su, Daiqin and Myers, Casey R. and Sabapathy, Krishna Kumar},
  journal = {Phys. Rev. A},
  volume = {100},
  issue = {5},
  pages = {052301},
  numpages = {32},
  year = {2019},
  month = {Nov},
  publisher = {American Physical Society},
  doi = {10.1103/PhysRevA.100.052301},
  url = {https://link.aps.org/doi/10.1103/PhysRevA.100.052301}
}

@article{gbs,
  title = {Gaussian Boson Sampling},
  author = {Hamilton, Craig S. and Kruse, Regina and Sansoni, Linda and Barkhofen, Sonja and Silberhorn, Christine and Jex, Igor},
  journal = {Phys. Rev. Lett.},
  volume = {119},
  issue = {17},
  pages = {170501},
  numpages = {5},
  year = {2017},
  month = {Oct},
  publisher = {American Physical Society},
  doi = {10.1103/PhysRevLett.119.170501},
  url = {https://link.aps.org/doi/10.1103/PhysRevLett.119.170501}
}

@article{gaussian_qi,
  title = {Gaussian quantum information},
  author = {Weedbrook, Christian and Pirandola, Stefano and Garc\'{\i}a-Patr\'on, Ra\'ul and Cerf, Nicolas J. and Ralph, Timothy C. and Shapiro, Jeffrey H. and Lloyd, Seth},
  journal = {Rev. Mod. Phys.},
  volume = {84},
  issue = {2},
  pages = {621--669},
  numpages = {0},
  year = {2012},
  month = {May},
  publisher = {American Physical Society},
  doi = {10.1103/RevModPhys.84.621},
  url = {https://link.aps.org/doi/10.1103/RevModPhys.84.621}
}

@article{GPS,
  title = {Generation of optical Schr\"odinger cat states by generalized photon subtraction},
  author = {Takase, Kan and Yoshikawa, Junichi and Asavanant, Warit and Endo, Mamoru and Furusawa, Akira},
  journal = {Phys. Rev. A},
  volume = {103},
  issue = {1},
  pages = {013710},
  numpages = {8},
  year = {2021},
  month = {Jan},
  publisher = {American Physical Society},
  doi = {10.1103/PhysRevA.103.013710},
  url = {https://link.aps.org/doi/10.1103/PhysRevA.103.013710}
}

@article{stellar_rank,
  title = {Stellar Representation of Non-Gaussian Quantum States},
  author = {Chabaud, Ulysse and Markham, Damian and Grosshans, Fr\'ed\'eric},
  journal = {Phys. Rev. Lett.},
  volume = {124},
  issue = {6},
  pages = {063605},
  numpages = {6},
  year = {2020},
  month = {Feb},
  publisher = {American Physical Society},
  doi = {10.1103/PhysRevLett.124.063605},
  url = {https://link.aps.org/doi/10.1103/PhysRevLett.124.063605}
}

@article{gps_exp,
  title={Boosting generation rate of squeezed single-photon states by generalized photon subtraction},
  author={Tomoda, Hiroko and Machinaga, Akihiro and Takase, Kan and Harada, Jun and Kashiwazaki, Takahiro and Umeki, Takeshi and Miki, Shigehito and China, Fumihiro and Yabuno, Masahiro and Terai, Hirotaka and others},
  journal={arXiv preprint arXiv:2404.19304},
  year={2024}
}

@article{ps_cat_multi,
  title = {Generation of optical coherent-state superpositions by number-resolved photon subtraction from the squeezed vacuum},
  author = {Gerrits, Thomas and Glancy, Scott and Clement, Tracy S. and Calkins, Brice and Lita, Adriana E. and Miller, Aaron J. and Migdall, Alan L. and Nam, Sae Woo and Mirin, Richard P. and Knill, Emanuel},
  journal = {Phys. Rev. A},
  volume = {82},
  issue = {3},
  pages = {031802},
  numpages = {4},
  year = {2010},
  month = {Sep},
  publisher = {American Physical Society},
  doi = {10.1103/PhysRevA.82.031802},
  url = {https://link.aps.org/doi/10.1103/PhysRevA.82.031802}
}

@article{ps_cat_multi_endo,
author = {Mamoru Endo and Ruofan He and Tatsuki Sonoyama and Kazuma Takahashi and Takahiro Kashiwazaki and Takeshi Umeki and Sachiko Takasu and Kaori Hattori and Daiji Fukuda and Kosuke Fukui and Kan Takase and Warit Asavanant and Petr Marek and Radim Filip and Akira Furusawa},
journal = {Opt. Express},
number = {8},
pages = {12865--12879},
publisher = {Optica Publishing Group},
title = {Non-Gaussian quantum state generation by multi-photon subtraction at the telecommunication wavelength},
volume = {31},
month = {Apr},
year = {2023},
url = {https://opg.optica.org/oe/abstract.cfm?URI=oe-31-8-12865},
doi = {10.1364/OE.486270},
}

@article{gkp,
  title = {Encoding a qubit in an oscillator},
  author = {Gottesman, Daniel and Kitaev, Alexei and Preskill, John},
  journal = {Phys. Rev. A},
  volume = {64},
  issue = {1},
  pages = {012310},
  numpages = {21},
  year = {2001},
  month = {Jun},
  publisher = {American Physical Society},
  doi = {10.1103/PhysRevA.64.012310},
  url = {https://link.aps.org/doi/10.1103/PhysRevA.64.012310}
}

@article{gkp_kan,
  title = {Generation of flying logical qubits using generalized photon subtraction with adaptive Gaussian operations},
  author = {Takase, Kan and Hanamura, Fumiya and Nagayoshi, Hironari and Bourassa, J. Eli and Alexander, Rafael N. and Kawasaki, Akito and Asavanant, Warit and Endo, Mamoru and Furusawa, Akira},
  journal = {Phys. Rev. A},
  volume = {110},
  issue = {1},
  pages = {012436},
  numpages = {10},
  year = {2024},
  month = jul,
  publisher = {American Physical Society},
  doi = {10.1103/PhysRevA.110.012436},
  url = {https://link.aps.org/doi/10.1103/PhysRevA.110.012436}
}

@book{hermite_approx,
  title={Orthogonal polynomials},
  author={Szeg, Gabor},
  volume={23},
  year={1939},
  publisher={American Mathematical Soc.}
}

@article{gaussian_breeding,
  author = {Takase, Kan and Fukui, Kosuke and Kawasaki, Akito and Asavanant, Warit and Endo, Mamoru and Yoshikawa, Junichi and van Loock, Peter and Furusawa, Akira},
  title = {Gottesman-Kitaev-Preskill qubit synthesizer for propagating light},
  journal = {npj Quantum Information},
  year = {2023},
  volume = {9},
  number = {1},
  pages = {98},
  doi = {10.1038/s41534-023-00772-y},
  url = {https://doi.org/10.1038/s41534-023-00772-y}
}

@article{williamson_decomp,
 ISSN = {00029327, 10806377},
 URL = {http://www.jstor.org/stable/2371062},
 author = {John Williamson},
 journal = {American Journal of Mathematics},
 number = {1},
 pages = {141--163},
 publisher = {Johns Hopkins University Press},
 title = {On the Algebraic Problem Concerning the Normal Forms of Linear Dynamical Systems},
 urldate = {2024-06-20},
 volume = {58},
 year = {1936}
}

@article{gottesman_knill,
  title = {Efficient Classical Simulation of Continuous Variable Quantum Information Processes},
  author = {Bartlett, Stephen D. and Sanders, Barry C. and Braunstein, Samuel L. and Nemoto, Kae},
  journal = {Phys. Rev. Lett.},
  volume = {88},
  issue = {9},
  pages = {097904},
  numpages = {4},
  year = {2002},
  month = {Feb},
  publisher = {American Physical Society},
  doi = {10.1103/PhysRevLett.88.097904},
  url = {https://link.aps.org/doi/10.1103/PhysRevLett.88.097904}
}

@book{nielsen_chuang,
  title={Quantum computation and quantum information},
  author={Nielsen, Michael A and Chuang, Isaac L},
  year={2010},
  publisher={Cambridge university press}
}

@book{wkb,
  title={Quantum optics in phase space},
  author={Schleich, Wolfgang P},
  year={2015},
  publisher={John Wiley \& Sons}
}

@article{cps_konno,
  title = {Nonlinear Squeezing for Measurement-Based Non-Gaussian Operations in Time Domain},
  author = {Konno, Shunya and Sakaguchi, Atsushi and Asavanant, Warit and Ogawa, Hisashi and Kobayashi, Masaya and Marek, Petr and Filip, Radim and Yoshikawa, Junichi and Furusawa, Akira},
  journal = {Phys. Rev. Appl.},
  volume = {15},
  issue = {2},
  pages = {024024},
  numpages = {9},
  year = {2021},
  month = {Feb},
  publisher = {American Physical Society},
  doi = {10.1103/PhysRevApplied.15.024024},
  url = {https://link.aps.org/doi/10.1103/PhysRevApplied.15.024024}
}

@article{cat_breeding_exp,
  author    = {Sychev, Demid V. and Ulanov, Alexander E. and Pushkina, Anastasia A. and Richards, Matthew W. and Fedorov, Ilya A. and Lvovsky, Alexander I.},
  title     = {Enlargement of optical Schrödinger's cat states},
  journal   = {Nature Photonics},
  year      = {2017},
  volume    = {11},
  number    = {6},
  pages     = {379--382},
  month     = jun,
  doi       = {10.1038/nphoton.2017.57},
  url       = {https://doi.org/10.1038/nphoton.2017.57},
  issn      = {1749-4893},
  id        = {Sychev2017}
}

@article{cat_breeding_theory,
  title = {Conditional production of superpositions of coherent states with inefficient photon detection},
  author = {Lund, A. P. and Jeong, H. and Ralph, T. C. and Kim, M. S.},
  journal = {Phys. Rev. A},
  volume = {70},
  issue = {2},
  pages = {020101},
  numpages = {4},
  year = {2004},
  month = aug,
  publisher = {American Physical Society},
  doi = {10.1103/PhysRevA.70.020101},
  url = {https://link.aps.org/doi/10.1103/PhysRevA.70.020101}
}

@Article{xanadu_gkp,
author={Larsen, M. V.
and Bourassa, J. E.
and Kocsis, S.
and Tasker, J. F.
and Chadwick, R. S.
and Gonz{\'a}lez-Arciniegas, C.
and Hastrup, J.
and Lopetegui-Gonz{\'a}lez, C. E.
and Miatto, F. M.
and Motamedi, A.
and Noro, R.
and Roeland, G.
and Baby, R.
and Chen, H.
and Contu, P.
and Di Luch, I.
and Drago, C.
and Giesbrecht, M.
and Grainge, T.
and Krasnokutska, I.
and Menotti, M.
and Morrison, B.
and Puviraj, C.
and Rezaei Shad, K.
and Hussain, B.
and McMahon, J.
and Ortmann, J. E.
and Collins, M. J.
and Ma, C.
and Phillips, D. S.
and Seymour, M.
and Tang, Q. Y.
and Yang, B.
and Vernon, Z.
and Alexander, R. N.
and Mahler, D. H.},
title={Integrated photonic source of Gottesman--Kitaev--Preskill qubits},
journal={Nature},
year={2025},
month={Jun},
day={04},
doi={10.1038/s41586-025-09044-5},
url={https://doi.org/10.1038/s41586-025-09044-5}
}

@article{nongaussian_resource,
  title = {Convex resource theory of non-Gaussianity},
  author = {Takagi, Ryuji and Zhuang, Quntao},
  journal = {Phys. Rev. A},
  volume = {97},
  issue = {6},
  pages = {062337},
  numpages = {14},
  year = {2018},
  month = {Jun},
  publisher = {American Physical Society},
  doi = {10.1103/PhysRevA.97.062337},
  url = {https://link.aps.org/doi/10.1103/PhysRevA.97.062337}
}

@article{cat_state,
  title = {Generating Schr\"odinger-cat-like states by means of conditional measurements on a beam splitter},
  author = {Dakna, M. and Anhut, T. and Opatrn\'y, T. and Kn\"oll, L. and Welsch, D.-G.},
  journal = {Phys. Rev. A},
  volume = {55},
  issue = {4},
  pages = {3184--3194},
  numpages = {0},
  year = {1997},
  month = apr,
  publisher = {American Physical Society},
  doi = {10.1103/PhysRevA.55.3184},
  url = {https://link.aps.org/doi/10.1103/PhysRevA.55.3184}
}

@article{gkp_squeezing,
  title = {Ground State Nature and Nonlinear Squeezing of Gottesman-Kitaev-Preskill States},
  author = {Marek, Petr},
  journal = {Phys. Rev. Lett.},
  volume = {132},
  issue = {21},
  pages = {210601},
  numpages = {6},
  year = {2024},
  month = {May},
  publisher = {American Physical Society},
  doi = {10.1103/PhysRevLett.132.210601},
  url = {https://link.aps.org/doi/10.1103/PhysRevLett.132.210601}
}

@article{duan_simon_duan,
  title = {Inseparability Criterion for Continuous Variable Systems},
  author = {Duan, Lu-Ming and Giedke, G. and Cirac, J. I. and Zoller, P.},
  journal = {Phys. Rev. Lett.},
  volume = {84},
  issue = {12},
  pages = {2722--2725},
  numpages = {0},
  year = {2000},
  month = {Mar},
  publisher = {American Physical Society},
  doi = {10.1103/PhysRevLett.84.2722},
  url = {https://link.aps.org/doi/10.1103/PhysRevLett.84.2722}
}

@article{duan_simon_simon,
  title = {Peres-Horodecki Separability Criterion for Continuous Variable Systems},
  author = {Simon, R.},
  journal = {Phys. Rev. Lett.},
  volume = {84},
  issue = {12},
  pages = {2726--2729},
  numpages = {0},
  year = {2000},
  month = {Mar},
  publisher = {American Physical Society},
  doi = {10.1103/PhysRevLett.84.2726},
  url = {https://link.aps.org/doi/10.1103/PhysRevLett.84.2726}
}

@article{cps_nls,
  title = {Implementation of a quantum cubic gate by an adaptive non-Gaussian measurement},
  author = {Miyata, Kazunori and Ogawa, Hisashi and Marek, Petr and Filip, Radim and Yonezawa, Hidehiro and Yoshikawa, Junichi and Furusawa, Akira},
  journal = {Phys. Rev. A},
  volume = {93},
  issue = {2},
  pages = {022301},
  numpages = {10},
  year = {2016},
  month = {Feb},
  publisher = {American Physical Society},
  doi = {10.1103/PhysRevA.93.022301},
  url = {https://link.aps.org/doi/10.1103/PhysRevA.93.022301}
}

@misc{cat_nls,
      title={Nonlinear Squeezing of Superpositions of Quadrature Eigenstates}, 
      author={Vojtěch Kuchař and Petr Marek},
      year={2025},
      eprint={2506.17437},
      archivePrefix={arXiv},
      primaryClass={quant-ph},
      url={https://arxiv.org/abs/2506.17437}, 
}

@article{gkp_breeding,
author = {H. M. Vasconcelos and L. Sanz and S. Glancy},
journal = {Opt. Lett.},
number = {19},
pages = {3261--3263},
publisher = {Optica Publishing Group},
title = {All-optical generation of states for ``Encoding a qubit in an oscillator''},
volume = {35},
month = oct,
year = {2010},
url = {https://opg.optica.org/ol/abstract.cfm?URI=ol-35-19-3261},
doi = {10.1364/OL.35.003261},
}

@article{xanadu_architecture,
  author = {Aghaee Rad, H. and Ainsworth, T. and Alexander, R. N. and Altieri, B. and Askarani, M. F. and Baby, R. and Banchi, L. and Baragiola, B. Q. and Bourassa, J. E. and Chadwick, R. S. and Charania, I. and Chen, H. and Collins, M. J. and Contu, P. and D’Arcy, N. and Dauphinais, G. and De Prins, R. and Deschenes, D. and Di Luch, I. and Duque, S. and Edke, P. and Fayer, S. E. and Ferracin, S. and Ferretti, H. and Gefaell, J. and Glancy, S. and González-Arciniegas, C. and Grainge, T. and Han, Z. and Hastrup, J. and Helt, L. G. and Hillmann, T. and Hundal, J. and Izumi, S. and Jaeken, T. and Jonas, M. and Kocsis, S. and Krasnokutska, I. and Larsen, M. V. and Laskowski, P. and Laudenbach, F. and Lavoie, J. and Li, M. and Lomonte, E. and Lopetegui, C. E. and Luey, B. and Lund, A. P. and Ma, C. and Madsen, L. S. and Mahler, D. H. and Mantilla Calderón, L. and Menotti, M. and Miatto, F. M. and Morrison, B. and Nadkarni, P. J. and Nakamura, T. and Neuhaus, L. and Niu, Z. and Noro, R. and Papirov, K. and Pesah, A. and Phillips, D. S. and Plick, W. N. and Rogalsky, T. and Rortais, F. and Sabines-Chesterking, J. and Safavi-Bayat, S. and Sazhaev, E. and Seymour, M. and Rezaei Shad, K. and Silverman, M. and Srinivasan, S. A. and Stephan, M. and Tang, Q. Y. and Tasker, J. F. and Teo, Y. S. and Then, R. B. and Tremblay, J. E. and Tzitrin, I. and Vaidya, V. D. and Vasmer, M. and Vernon, Z. and Villalobos, L. F. S. S. M. and Walshe, B. W. and Weil, R. and Xin, X. and Yan, X. and Yao, Y. and Zamani Abnili, M. and Zhang, Y.},
  title = {Scaling and networking a modular photonic quantum computer},
  journal = {Nature},
  year = {2025},
  volume = {638},
  number = {8052},
  pages = {912--919},
  doi = {10.1038/s41586-024-08406-9}
}

@misc{mrmustard,
  author       = {XanaduAI},
  title        = {MrMustard},
  howpublished = {\url{https://github.com/XanaduAI/MrMustard}},
  year         = {2021}
}

@misc{cps_wave_airy,
      title={Nonlinear Phase Gates as Airy Transforms of the Wigner Function}, 
      author={Darren W. Moore and Radim Filip},
      year={2025},
      eprint={2504.04851},
      archivePrefix={arXiv},
      primaryClass={quant-ph},
      url={https://arxiv.org/abs/2504.04851}, 
}

@article{hermite_airy,
title = {Asymptotic Formula for Quantum Harmonic Oscillator Tunneling Probabilities},
journal = {Reports on Mathematical Physics},
volume = {76},
number = {2},
pages = {149-158},
year = {2015},
issn = {0034-4877},
doi = {https://doi.org/10.1016/S0034-4877(15)30025-2},
url = {https://www.sciencedirect.com/science/article/pii/S0034487715300252},
author = {Arkadiusz Jadczyk}
}

@article{warit_cluster,
author = {Warit Asavanant  and Yu Shiozawa  and Shota Yokoyama  and Baramee Charoensombutamon  and Hiroki Emura  and Rafael N. Alexander  and Shuntaro Takeda  and Junichi Yoshikawa  and Nicolas C. Menicucci  and Hidehiro Yonezawa  and Akira Furusawa },
title = {Generation of time-domain-multiplexed two-dimensional cluster state},
journal = {Science},
volume = {366},
number = {6463},
pages = {373-376},
year = {2019},
doi = {10.1126/science.aay2645}
}

@article{mikkel_cluster,
  title={Deterministic multi-mode gates on a scalable photonic quantum computing platform},
  author={Larsen, Mikkel V and Guo, Xueshi and Breum, Casper R and Neergaard-Nielsen, Jonas S and Andersen, Ulrik L},
  journal={Nature Physics},
  volume={17},
  pages={1018-1023},
  year={2021},
  publisher={Nature Publishing Group}
}

@article{cluster_operation,
  title = {Time-Domain-Multiplexed Measurement-Based Quantum Operations with 25-{MHz} Clock Frequency},
  author = {Asavanant, Warit and Charoensombutamon, Baramee and Yokoyama, Shota and Ebihara, Takeru and Nakamura, Tomohiro and Alexander, Rafael N. and Endo, Mamoru and Yoshikawa, Junichi and Menicucci, Nicolas C. and Yonezawa, Hidehiro and Furusawa, Akira},
  journal = {Phys. Rev. Appl.},
  volume = {16},
  issue = {3},
  pages = {034005},
  numpages = {10},
  year = {2021},
  month = sep,
  publisher = {American Physical Society},
  doi = {10.1103/PhysRevApplied.16.034005},
  url = {https://link.aps.org/doi/10.1103/PhysRevApplied.16.034005}
}

@misc{mamoru_four_photon,
      title={High-Rate Four Photon Subtraction from Squeezed Vacuum: Preparing Cat State for Optical Quantum Computation}, 
      author={Mamoru Endo and Takefumi Nomura and Tatsuki Sonoyama and Kazuma Takahashi and Sachiko Takasu and Daiji Fukuda and Takahiro Kashiwazaki and Asuka Inoue and Takeshi Umeki and Rajveer Nehra and Petr Marek and Radim Filip and Kan Takase and Warit Asavanant and Akira Furusawa},
      year={2025},
      eprint={2502.08952},
      archivePrefix={arXiv},
      primaryClass={quant-ph},
      url={https://arxiv.org/abs/2502.08952}, 
}

@article{mikkel_cluster_operation,
  title = {Deterministic Multi-Mode Gates on a Scalable Photonic Quantum Computing Platform},
  author = {Larsen, Mikkel V. and Guo, Xueshi and Breum, Casper R. and {Neergaard-Nielsen}, Jonas S. and Andersen, Ulrik L.},
  year = {2021},
  month = sep,
  journal = {Nature Physics},
  volume = {17},
  number = {9},
  pages = {1018--1023},
  publisher = {Nature Publishing Group},
  issn = {1745-2481},
  doi = {10.1038/s41567-021-01296-y},
  urldate = {2025-08-27},
  copyright = {2021 The Author(s), under exclusive licence to Springer Nature Limited},
  langid = {english}
}

@article{qec_nogo,
  title={No-go theorem for Gaussian quantum error correction},
  author={Niset, Julien and Fiur{\'a}{\v{s}}ek, Jarom{\'\i}r and Cerf, Nicolas J},
  journal={Physical review letters},
  volume={102},
  number={12},
  pages={120501},
  year={2009},
  publisher={APS}
}

@article{nick_universal,
  title = {Universal {{Quantum Computation}} with {{Continuous-Variable Cluster States}}},
  author = {Menicucci, Nicolas C. and {van Loock}, Peter and Gu, Mile and Weedbrook, Christian and Ralph, Timothy C. and Nielsen, Michael A.},
  year = {2006},
  month = sep,
  journal = {Physical Review Letters},
  volume = {97},
  number = {11},
  pages = {110501},
  publisher = {American Physical Society},
  doi = {10.1103/PhysRevLett.97.110501},
  urldate = {2025-08-27}
}

@article{single_mode_disp,
  title = {Single-Mode Displacement Sensor},
  author = {Duivenvoorden, Kasper and Terhal, Barbara M. and Weigand, Daniel},
  year = {2017},
  month = jan,
  journal = {Physical Review A},
  volume = {95},
  number = {1},
  pages = {012305},
  issn = {2469-9926, 2469-9934},
  doi = {10.1103/PhysRevA.95.012305},
  urldate = {2025-08-25},
  copyright = {http://link.aps.org/licenses/aps-default-license},
  langid = {english}
}

@article{hanamura_gaussian_disp,
  title = {Estimation of {{Gaussian}} Random Displacement Using Non-{{Gaussian}} States},
  author = {Hanamura, Fumiya and Asavanant, Warit and Fukui, Kosuke and Konno, Shunya and Furusawa, Akira},
  year = {2021},
  month = dec,
  journal = {Physical Review A},
  volume = {104},
  number = {6},
  pages = {062601},
  issn = {2469-9926, 2469-9934},
  doi = {10.1103/PhysRevA.104.062601},
  urldate = {2025-08-25},
  langid = {english}
}

@article{cat_code,
  title = {Macroscopically Distinct Quantum-Superposition States as a Bosonic Code for Amplitude Damping},
  author = {Cochrane, P. T. and Milburn, G. J. and Munro, W. J.},
  year = {1999},
  month = apr,
  journal = {Physical Review A},
  volume = {59},
  number = {4},
  pages = {2631--2634},
  publisher = {American Physical Society},
  doi = {10.1103/PhysRevA.59.2631},
  urldate = {2025-08-27}
}

@article{binomial_code,
  title = {New {{Class}} of {{Quantum Error-Correcting Codes}} for a {{Bosonic Mode}}},
  author = {Michael, Marios H. and Silveri, Matti and Brierley, R. T. and Albert, Victor V. and Salmilehto, Juha and Jiang, Liang and Girvin, S. M.},
  year = {2016},
  month = jul,
  journal = {Physical Review X},
  volume = {6},
  number = {3},
  pages = {031006},
  publisher = {American Physical Society},
  doi = {10.1103/PhysRevX.6.031006},
  urldate = {2025-08-27}
}

@article{metrology_nls,
  title = {Metrological {{Nonlinear Squeezing Parameter}}},
  author = {Gessner, Manuel and Smerzi, Augusto and Pezz{\`e}, Luca},
  year = {2019},
  month = mar,
  journal = {Physical Review Letters},
  volume = {122},
  number = {9},
  pages = {090503},
  publisher = {American Physical Society},
  doi = {10.1103/PhysRevLett.122.090503},
  urldate = {2025-08-27}
}

@article{wigneg_qc,
  title = {Positive {{Wigner Functions Render Classical Simulation}} of {{Quantum Computation Efficient}}},
  author = {Mari, A. and Eisert, J.},
  year = {2012},
  month = dec,
  journal = {Physical Review Letters},
  volume = {109},
  number = {23},
  pages = {230503},
  publisher = {American Physical Society},
  doi = {10.1103/PhysRevLett.109.230503},
  urldate = {2025-08-27}
}

@article{cv_qc,
  title = {Quantum {{Computation}} over {{Continuous Variables}}},
  author = {Lloyd, Seth and Braunstein, Samuel L.},
  year = {1999},
  month = feb,
  journal = {Physical Review Letters},
  volume = {82},
  number = {8},
  pages = {1784--1787},
  publisher = {American Physical Society},
  doi = {10.1103/PhysRevLett.82.1784},
  urldate = {2025-08-27}
}

@article{pan_gbs,
  title = {Quantum Computational Advantage Using Photons},
  author = {Zhong, Han-Sen and Wang, Hui and Deng, Yu-Hao and Chen, Ming-Cheng and Peng, Li-Chao and Luo, Yi-Han and Qin, Jian and Wu, Dian and Ding, Xing and Hu, Yi and Hu, Peng and Yang, Xiao-Yan and Zhang, Wei-Jun and Li, Hao and Li, Yuxuan and Jiang, Xiao and Gan, Lin and Yang, Guangwen and You, Lixing and Wang, Zhen and Li, Li and Liu, Nai-Le and Lu, Chao-Yang and Pan, Jian-Wei},
  year = {2020},
  month = dec,
  journal = {Science},
  volume = {370},
  number = {6523},
  pages = {1460--1463},
  publisher = {American Association for the Advancement of Science},
  doi = {10.1126/science.abe8770},
  urldate = {2025-08-27}
}

@article{radim_atom_ng,
  title = {Quantum {{Non-Gaussianity}} of {{Multiphonon States}} of a {{Single Atom}}},
  author = {Podhora, L. and Lachman, L. and Pham, T. and Le{\v s}und{\'a}k, A. and {\v C}{\'i}p, O. and Slodi{\v c}ka, L. and Filip, R.},
  year = {2022},
  month = jun,
  journal = {Physical Review Letters},
  volume = {129},
  number = {1},
  pages = {013602},
  publisher = {American Physical Society},
  doi = {10.1103/PhysRevLett.129.013602},
  urldate = {2025-08-27}
}

@article{stellar_rank_heterodyne,
  title = {Certification of {{Non-Gaussian States}} with {{Operational Measurements}}},
  author = {Chabaud, Ulysse and Roeland, Gana{\"e}l and Walschaers, Mattia and Grosshans, Fr{\'e}d{\'e}ric and Parigi, Valentina and Markham, Damian and Treps, Nicolas},
  year = {2021},
  month = jun,
  journal = {PRX Quantum},
  volume = {2},
  number = {2},
  pages = {020333},
  publisher = {American Physical Society},
  doi = {10.1103/PRXQuantum.2.020333},
  urldate = {2025-08-27}
}

@article{atsushi_cpg,
  title = {Nonlinear Feedforward Enabling Quantum Computation},
  author = {Sakaguchi, Atsushi and Konno, Shunya and Hanamura, Fumiya and Asavanant, Warit and Takase, Kan and Ogawa, Hisashi and Marek, Petr and Filip, Radim and Yoshikawa, Junichi and Huntington, Elanor},
  year = {2023},
  journal = {Nature Communications},
  volume = {14},
  number = {1},
  pages = {3817},
  publisher = {Nature Publishing Group UK London},
  urldate = {2025-08-25}
}

@article{cat_origin,
  title = {Even and Odd Coherent States and Excitations of a Singular Oscillator},
  author = {Dodonov, V. V. and Malkin, I. A. and Man'ko, V. I.},
  year = {1974},
  month = mar,
  journal = {Physica},
  volume = {72},
  number = {3},
  pages = {597--615},
  issn = {0031-8914},
  doi = {10.1016/0031-8914(74)90215-8},
  urldate = {2025-08-27}
}

@article{konno_gkp,
author = {Shunya Konno  and Warit Asavanant  and Fumiya Hanamura  and Hironari Nagayoshi  and Kosuke Fukui  and Atsushi Sakaguchi  and Ryuhoh Ide  and Fumihiro China  and Masahiro Yabuno  and Shigehito Miki  and Hirotaka Terai  and Kan Takase  and Mamoru Endo  and Petr Marek  and Radim Filip  and Peter van Loock  and Akira Furusawa },
title = {Logical states for fault-tolerant quantum computation with propagating light},
journal = {Science},
volume = {383},
number = {6680},
pages = {289-293},
year = {2024},
doi = {10.1126/science.adk7560},
URL = {https://www.science.org/doi/abs/10.1126/science.adk7560},
eprint = {https://www.science.org/doi/pdf/10.1126/science.adk7560}
}

@article{phase_conj_channel,
  title = {Operator-Sum Representation for Bosonic {{Gaussian}} Channels},
  author = {Ivan, J. Solomon and Sabapathy, Krishna Kumar and Simon, R.},
  year = {2011},
  month = oct,
  journal = {Physical Review A},
  volume = {84},
  number = {4},
  pages = {042311},
  publisher = {American Physical Society},
  doi = {10.1103/PhysRevA.84.042311},
  urldate = {2025-08-27}
}

@article{stellar_rank_original,
  title = {Faithful {{Hierarchy}} of {{Genuine}} \$n\$-{{Photon Quantum Non-Gaussian Light}}},
  author = {Lachman, Luk{\'a}{\v s} and Straka, Ivo and Hlou{\v s}ek, Josef and Je{\v z}ek, Miroslav and Filip, Radim},
  year = {2019},
  month = jul,
  journal = {Physical Review Letters},
  volume = {123},
  number = {4},
  pages = {043601},
  publisher = {American Physical Society},
  doi = {10.1103/PhysRevLett.123.043601},
  urldate = {2025-09-01}
}

@article{core_state,
  title = {Gaussian-Optimized Preparation of Non-{{Gaussian}} Pure States},
  author = {Menzies, David and Filip, Radim},
  year = {2009},
  month = jan,
  journal = {Physical Review A},
  volume = {79},
  number = {1},
  pages = {012313},
  publisher = {American Physical Society},
  doi = {10.1103/PhysRevA.79.012313},
  urldate = {2025-09-01}
}

@article{fock_to_cat,
  title = {Generation of Optical `{{Schr{\"o}dinger}} Cats' from Photon Number States},
  author = {Ourjoumtsev, Alexei and Jeong, Hyunseok and {Tualle-Brouri}, Rosa and Grangier, Philippe},
  year = {2007},
  month = aug,
  journal = {Nature},
  volume = {448},
  number = {7155},
  pages = {784--786},
  issn = {0028-0836, 1476-4687},
  doi = {10.1038/nature06054},
  urldate = {2025-09-02},
  langid = {english}
}

@article{wf_engineer,
  title = {Wave-Function Engineering via Conditional Quantum Teleportation with a Non-{{Gaussian}} Entanglement Resource},
  author = {Asavanant, Warit and Takase, Kan and Fukui, Kosuke and Endo, Mamoru and Yoshikawa, Junichi and Furusawa, Akira},
  year = {2021},
  month = apr,
  journal = {Physical Review A},
  volume = {103},
  number = {4},
  pages = {043701},
  publisher = {American Physical Society},
  doi = {10.1103/PhysRevA.103.043701},
  urldate = {2025-09-02}
}

@article{nonlinear_potential,
  title = {Nonlinear Potential of a Quantum Oscillator Induced by Single Photons},
  author = {Park, Kimin and Marek, Petr and Filip, Radim},
  year = {2014},
  month = jul,
  journal = {Physical Review A},
  volume = {90},
  number = {1},
  pages = {013804},
  publisher = {American Physical Society},
  doi = {10.1103/PhysRevA.90.013804},
  urldate = {2025-09-01}
}

@article{sp_add_subtract,
  title = {Quantum State Engineering by a Coherent Superposition of Photon Subtraction and Addition},
  author = {Lee, Su-Yong and Nha, Hyunchul},
  year = {2010},
  month = nov,
  journal = {Physical Review A},
  volume = {82},
  number = {5},
  pages = {053812},
  publisher = {American Physical Society},
  doi = {10.1103/PhysRevA.82.053812},
  urldate = {2025-09-02}
}

@article{addition_exp,
  title = {Quantum-to-{{Classical Transition}} with {{Single-Photon-Added Coherent States}} of {{Light}}},
  author = {Zavatta, Alessandro and Viciani, Silvia and Bellini, Marco},
  year = {2004},
  month = oct,
  journal = {Science},
  volume = {306},
  number = {5696},
  pages = {660--662},
  publisher = {American Association for the Advancement of Science},
  doi = {10.1126/science.1103190},
  urldate = {2025-09-02}
}

@article{addition_theory,
  title = {Engineering Quantum Operations on Traveling Light Beams by Multiple Photon Addition and Subtraction},
  author = {Fiur{\'a}{\v s}ek, Jarom{\'i}r},
  year = {2009},
  month = nov,
  journal = {Physical Review A},
  volume = {80},
  number = {5},
  pages = {053822},
  publisher = {American Physical Society},
  doi = {10.1103/PhysRevA.80.053822},
  urldate = {2025-09-02}
}

@article{fock_exp,
  title = {Quantum {{State Reconstruction}} of the {{Single-Photon Fock State}}},
  author = {Lvovsky, A. I. and Hansen, H. and Aichele, T. and Benson, O. and Mlynek, J. and Schiller, S.},
  year = {2001},
  month = jul,
  journal = {Physical Review Letters},
  volume = {87},
  number = {5},
  pages = {050402},
  publisher = {American Physical Society},
  doi = {10.1103/PhysRevLett.87.050402},
  urldate = {2025-09-02}
}

@article{tatsuki_fock,
  title = {Generation of Multi-Photon {{Fock}} States at Telecommunication Wavelength Using Picosecond Pulsed Light},
  author = {Sonoyama, Tatsuki and Takahashi, Kazuma and Sano, Tomoki and Suzuki, Takumi and Nomura, Takefumi and Yabuno, Masahiro and Miki, Shigehito and Terai, Hirotaka and Takase, Kan and Asavanant, Warit and Endo, Mamoru and Furusawa, Akira},
  year = {2024},
  month = aug,
  journal = {Optics Express},
  volume = {32},
  number = {18},
  pages = {32387},
  issn = {1094-4087},
  doi = {10.1364/OE.530164},
  urldate = {2025-09-02},
  langid = {english}
}

@article{gaussian_filter,
  title = {Efficient Representation of Purity-Preserving {{Gaussian}} Quantum Filters},
  author = {Fiur{\'a}{\v s}ek, Jarom{\'i}r},
  year = {2013},
  month = may,
  journal = {Physical Review A},
  volume = {87},
  number = {5},
  pages = {052301},
  publisher = {American Physical Society},
  doi = {10.1103/PhysRevA.87.052301},
  urldate = {2025-09-02}
}

@article{supremacy_review,
  title = {Quantum Sampling Problems, {{BosonSampling}} and Quantum Supremacy},
  author = {Lund, A. P. and Bremner, Michael J. and Ralph, T. C.},
  year = {2017},
  month = apr,
  journal = {npj Quantum Information},
  volume = {3},
  number = {1},
  pages = {15},
  publisher = {Nature Publishing Group},
  issn = {2056-6387},
  doi = {10.1038/s41534-017-0018-2},
  urldate = {2025-09-02},
  copyright = {2017 The Author(s)},
  langid = {english}
}

@article{boson_sampling,
  title = {The {{Computational Complexity}} of {{Linear Optics}}},
  author = {Aaronson, Scott and Arkhipov, Alex},
  year = {2013},
  month = feb,
  journal = {Theory of Computing},
  volume = {9},
  number = {4},
  pages = {143--252},
  publisher = {Theory of Computing},
  doi = {10.4086/toc.2013.v009a004},
  urldate = {2025-09-02}
}

@article{nl_amp_attn,
  title = {Gaussian Postselection and Virtual Noiseless Amplification in Continuous-Variable Quantum Key Distribution},
  author = {Fiur{\'a}{\v s}ek, Jarom{\'i}r and Cerf, Nicolas J.},
  year = {2012},
  month = dec,
  journal = {Physical Review A},
  volume = {86},
  number = {6},
  pages = {060302},
  issn = {1050-2947, 1094-1622},
  doi = {10.1103/PhysRevA.86.060302},
  urldate = {2026-01-18},
  copyright = {http://link.aps.org/licenses/aps-default-license},
  langid = {english}
}

@article{nl_attn,
  title = {Noiseless {{Loss Suppression}} in {{Quantum Optical Communication}}},
  author = {Mi{\v c}uda, M. and Straka, I. and Mikov{\'a}, M. and Du{\v s}ek, M. and Cerf, N. J. and Fiur{\'a}{\v s}ek, J. and Je{\v z}ek, M.},
  year = {2012},
  month = nov,
  journal = {Physical Review Letters},
  volume = {109},
  number = {18},
  pages = {180503},
  issn = {0031-9007, 1079-7114},
  doi = {10.1103/PhysRevLett.109.180503},
  urldate = {2026-01-18},
  copyright = {http://link.aps.org/licenses/aps-default-license},
  langid = {english}
}

@article{nl_amp,
  title = {Nondeterministic {{Noiseless Linear Amplification}} of {{Quantum Systems}}},
  author = {Ralph, T. C. and Lund, A. P.},
  year = {2009},
  month = apr,
  journal = {AIP Conference Proceedings},
  volume = {1110},
  number = {1},
  pages = {155--160},
  issn = {0094-243X},
  doi = {10.1063/1.3131295},
  urldate = {2026-01-18}
}

@misc{aralovPhotonCatalysis,
  title = {Photon Catalysis for General Multimode Multi-Photon Quantum State Preparation},
  author = {Aralov, Andrei and Gillet, {\'E}milie and Nguyen, Viet and Cosentino, Andrea and Walschaers, Mattia and Frigerio, Massimo},
  year = {2025},
  month = aug,
  number = {arXiv:2507.19397},
  eprint = {2507.19397},
  primaryclass = {quant-ph},
  publisher = {arXiv},
  doi = {10.48550/arXiv.2507.19397},
  urldate = {2026-01-19},
  archiveprefix = {arXiv}
}

@misc{meleSymplecticRank,
  title = {The Symplectic Rank of Non-{{Gaussian}} Quantum States},
  author = {Mele, Francesco Anna and Oliviero, Salvatore Francesco Emanuele and Upreti, Varun and Chabaud, Ulysse},
  year = {2025},
  month = apr,
  number = {arXiv:2504.19319},
  eprint = {2504.19319},
  primaryclass = {quant-ph},
  doi = {10.48550/arXiv.2504.19319},
  urldate = {2025-09-09},
  archiveprefix = {arXiv},
  langid = {english}
}

@misc{approx_stellar,
  title = {Assessing Non-{{Gaussian}} Quantum State Conversion with the Stellar Rank},
  author = {Hahn, Oliver and Ferrini, Giulia and Ferraro, Alessandro and Chabaud, Ulysse},
  year = {2025},
  month = may,
  number = {arXiv:2410.23721},
  eprint = {2410.23721},
  primaryclass = {quant-ph},
  publisher = {arXiv},
  doi = {10.48550/arXiv.2410.23721},
  urldate = {2026-01-19},
  archiveprefix = {arXiv}
}

@article{usugaNoisepoweredProbabilistic,
  title = {Noise-Powered Probabilistic Concentration of Phase Information},
  author = {Usuga, Mario A. and M{\"u}ller, Christian R. and Wittmann, Christoffer and Marek, Petr and Filip, Radim and Marquardt, Christoph and Leuchs, Gerd and Andersen, Ulrik L.},
  year = {2010},
  month = oct,
  journal = {Nature Physics},
  volume = {6},
  number = {10},
  pages = {767--771},
  publisher = {Nature Publishing Group},
  issn = {1745-2481},
  doi = {10.1038/nphys1743},
  urldate = {2026-01-25},
  copyright = {2010 Springer Nature Limited},
  langid = {english}
}

@article{zavattaHighfidelityNoiseless,
  title = {A High-Fidelity Noiseless Amplifier for Quantum Light States},
  author = {Zavatta, A. and Fiur{\'a}{\v s}ek, J. and Bellini, M.},
  year = {2011},
  month = jan,
  journal = {Nature Photonics},
  volume = {5},
  number = {1},
  pages = {52--56},
  publisher = {Nature Publishing Group},
  issn = {1749-4893},
  doi = {10.1038/nphoton.2010.260},
  urldate = {2026-01-25},
  copyright = {2010 Springer Nature Limited},
  langid = {english}
}

@article{kocsisHeraldedNoiseless,
  title = {Heralded Noiseless Amplification of a Photon Polarization Qubit},
  author = {Kocsis, S. and Xiang, G. Y. and Ralph, T. C. and Pryde, G. J.},
  year = {2013},
  month = jan,
  journal = {Nature Physics},
  volume = {9},
  number = {1},
  pages = {23--28},
  publisher = {Nature Publishing Group},
  issn = {1745-2481},
  doi = {10.1038/nphys2469},
  urldate = {2026-01-25},
  copyright = {2012 Springer Nature Limited},
  langid = {english}
}

@article{marekCoherentstatePhase,
  title = {Coherent-State Phase Concentration by Quantum Probabilistic Amplification},
  author = {Marek, Petr and Filip, Radim},
  year = {2010},
  month = feb,
  journal = {Physical Review A},
  volume = {81},
  number = {2},
  pages = {022302},
  issn = {1050-2947, 1094-1622},
  doi = {10.1103/PhysRevA.81.022302},
  urldate = {2026-01-25},
  copyright = {http://link.aps.org/licenses/aps-default-license},
  langid = {english}
}

@misc{fiurasekMaximumHeralding,
  title = {Maximum Heralding Probabilities of Non-Classical State Generation from Two-Mode {{Gaussian}} State via Photon Counting Measurements},
  author = {Fiur{\'a}{\v s}ek, Jarom{\'i}r},
  year = {2025},
  month = oct,
  number = {arXiv:2510.01951},
  eprint = {2510.01951},
  primaryclass = {quant-ph},
  publisher = {arXiv},
  doi = {10.48550/arXiv.2510.01951},
  urldate = {2026-01-25},
  archiveprefix = {arXiv}
}

@article{chabaudResourcesBosonic,
  title = {Resources for {{Bosonic Quantum Computational Advantage}}},
  author = {Chabaud, Ulysse and Walschaers, Mattia},
  year = {2023},
  month = mar,
  journal = {Physical Review Letters},
  volume = {130},
  number = {9},
  pages = {090602},
  publisher = {American Physical Society},
  doi = {10.1103/PhysRevLett.130.090602},
  urldate = {2026-01-25}
}
\appendix
\clearpage
\section{Mixed state generators}\label{sec:mixed_state}
In real-world experiments, non-Gaussian state generators inevitably suffer from imperfections such as loss, leading to mixed output states. The dominant errors in optical experiments, including loss, are Gaussian errors, which map a Gaussian state to another (generally mixed) Gaussian state. While we have focused on pure non-Gaussian state generators in the main text, the definition of the control-mode representation in Sec.~\ref{sec:idler_rep_multi} is sufficiently general to also cover mixed non-Gaussian generators affected by Gaussian noise. In fact, the control moments $(C,\bm\beta)$ can be defined in exactly the same way, simply by dropping the assumption of a pure Gaussian state. This broader formulation enables the experimental evaluation of non-Gaussian control parameters for realistic state generators, which is crucial for applications. The following theorem is key for discussing mixed non-Gaussian state generators.

\begin{apptheorem}\label{thm:gaussian_noise_signal}
For an arbitrary $l+k$-mode Gaussian state $\hat{\rho}_G$ and a $l'+k$-mode pure Gaussian state $\ket{G}$ having the same $k$-mode control moments $(C,\bm\beta)$, there exists a Gaussian channel $\mathcal{E}_\signal: \mathcal{D}(\mathcal{H}^{\otimes l'})\to \mathcal{D}(\mathcal{H}^{\otimes l})$ acting only on the signal modes such that
\begin{align}
\hat{\rho}_G=(\mathcal{E}_\signal\otimes \mathcal{I})(\ketbra{G}{G}).\label{eq:mixed_gaussian}
\end{align}
\end{apptheorem}
\begin{proof}
Because all pure states with the same $(C,\bm\beta)$ are related by a Gaussian unitary transformation from Theorem \ref{thm:purificationUniqueness}, it suffices to consider the case $l'=r$, where $r$ is the Schmidt rank of $C$. The proof is represented graphically in Fig.~\ref{fig:gaussian_noise_signal}. Let $\ket{G}_\mathrm{p}$ be a Gaussian purification of $\hat{\rho}_G$ with $l_\mathrm{p}+k$ modes. Since $\ket{G}_\mathrm{p}$ is also a purification of the control modes, Theorem \ref{thm:purificationUniqueness} gives
\begin{align}
    \ket{G}_\mathrm{p}=(\hat{U}_\signal\otimes\hat{I})\qty(\ket{0}^{\otimes l_\mathrm{p}-r}\otimes \ket{G}),
\end{align}
where $\hat{U}_\signal:\mathcal{H}^{\otimes l_p}\to \mathcal{H}^{\otimes l_p}$ is a unitary operator acting on the signal modes. By tracing out the ancillary modes, we obtain Eq.~\eqref{eq:mixed_gaussian}, where $\mathcal{E}_\signal$ is defined as
\begin{align}
\mathcal{E}_\signal(\rho)=\Tr_{\mathrm{a}}\qty[\hat{U}\qty(\ketbra{0}{0}^{\otimes l_\mathrm{p}-r}\otimes \rho)],
\end{align}
and $\Tr_{\mathrm{a}}$ denotes the partial trace over $l_{\mathrm{p}}-l$ ancillary modes.
\end{proof}

\begin{figure*}
    \centering
    \input{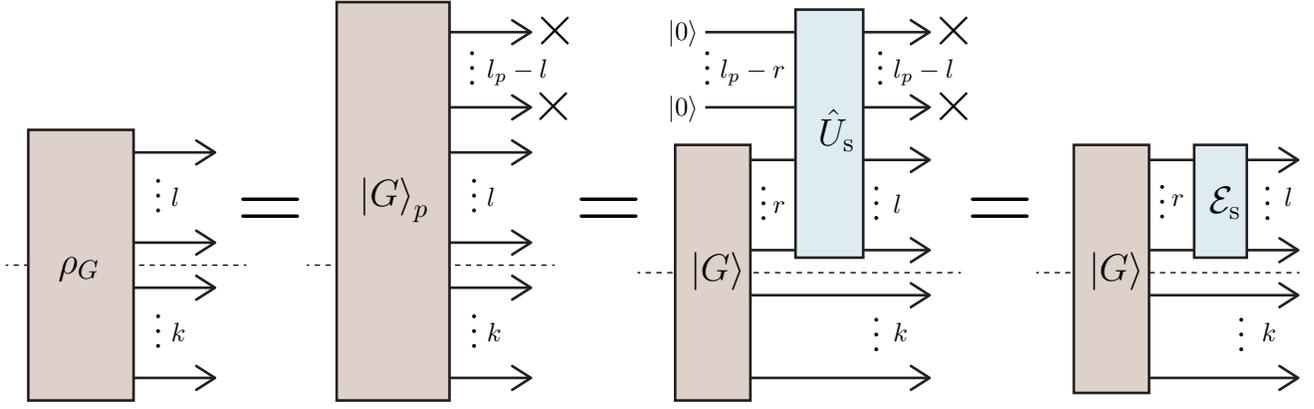}
    \caption{Proof of Theorem \ref{thm:gaussian_noise_signal}. The symbol $\times$ indicates tracing out the mode.}
    \label{fig:gaussian_noise_signal}
\end{figure*}

Mixed non-Gaussian output states can generally result from Gaussian noise acting on either the signal modes or the control modes. However, Theorem \ref{thm:gaussian_noise_signal} shows that Gaussian noise on the control modes is equivalent to noise on the signal modes. Therefore, we can define the control-mode representation in the same way as in the pure-state case, as shown in the following corollary, corresponding to Corollary \ref{cor:idler_rep} for pure states.

\begin{appcor}\label{thm:cparam_uniqueness_mixed}
For a multi-mode (mixed) non-Gaussian state generator $(C, \bm{\beta}, \bm{n}, \hat{U}_\signal)$:
\begin{itemize}
    \item The output non-Gaussian state is determined by $(C,\bm{\beta},\bm{n})$, up to a Gaussian channel. Specifically, for an arbitrary output state $\hat{\rho}_{(C,\bm{\beta}),\bm{n}}$ and a pure output state $\ket{\psi}_{(C,\bm{\beta}),\bm{n}}$ with the same control moments, there exists a Gaussian channel $\mathcal{E}_\signal$ such that
\begin{align}
\hat{\rho}_{(C,\bm{\beta}),\bm{n}} = \mathcal{E}_\signal\qty(\ket{\psi}_{(C,\bm{\beta}),\bm{n}}\bra{\psi}_{(C,\bm{\beta}),\bm{n}}).
\end{align}
    Accordingly, we denote the state as $\hat{\rho}_{C,\bm{\beta},\bm{n},\mathcal{E}_\signal}$.
    \item The success probability of the state generation is determined solely by $(C,\bm{\beta},\bm{n})$, which we denote by $p_{\bm{n}}(C,\bm\beta)$.
\end{itemize}
\end{appcor}
\begin{proof}
    This follows directly from Theorem \ref{thm:gaussian_noise_signal}.
\end{proof}

Furthermore, for the case of two-mode state generators, we have the following theorem, corresponding to Theorem \ref{thm:gps_uniqueness}.

\begin{apptheorem}\label{thm:gps_uniqueness_mixed}
For $n \geq 2$, if 
$\hat{\rho}_{s_0,\delta_0,n,\mathcal{E}_\signal} 
= \hat{\rho}_{s_0',\delta_0',n,\mathcal{E}'_\signal}$, 
at least one of the following holds:
\begin{align}
    &s_0=s_0'=0 \quad\text{and}\quad |\delta_0|=|\delta_0'|,\\
    &(s_0,\delta_0)=(s_0',\pm\delta_0'),\\
    &(s_0,\delta_0)=(s_0',\pm\delta_0'^*).\label{eq:conjugate_possibility}
\end{align}
\end{apptheorem}
\begin{proof}
We show this by proving the uniqueness of the polynomial part of the characteristic function. For a state $\rho$, the characteristic function is defined as
\begin{align}
    \chi(\xi)=\Tr(\rho\hat{D}(\xi)).
\end{align}
The characteristic function of a state with stellar rank $n$ can be written as
\begin{align}
    \chi(\xi)=G(\xi,\xi^*)P(\xi,\xi^*),
\end{align}
where $G$ is a Gaussian function and $P$ is a polynomial in $\xi,\xi^*$ of total degree at most $2n$. 
Since the characteristic function of $\outerproduct{n}{m}$ for $n\leq m$ is
\begin{align}
\chi_{\outerproduct{n}{m}}(\xi)
&= \matrixel{m}{\hat{D}(\xi)}{n}\\
&= e^{-\tfrac{1}{2}\qty|\xi|^2}
\sqrt{\frac{n!}{m!}}
(-\xi^*)^{m-n}
L_{n}^{(m-n)}\qty(|\xi|^2),
\end{align}
where $L_{n}^{(\alpha)}$ is the generalized Laguerre polynomial, the (pure) particle form $\ket{\psi}_{s_0,\delta_0,n}^{(\mathrm{p})}$ has
\begin{align}
    G_{s_0,\delta_0,n}(\xi,\xi^*)&\propto \exp(-\tfrac12|\xi|^2),\\
    P_{s_0,\delta_0,n}(\xi,\xi^*)&=\qty(|\xi|^{2n}+\mathcal{O}(|\xi|^{2n-1})),\label{eq:poly_part_exp}
\end{align}
where $\mathcal{O}(|\xi|^{2n-1})$ denotes terms of total degree at most $2n-1$.

The action of a Gaussian map $\mathcal{E}$ on the characteristic function is
\begin{align}
    \chi(\xi) \mapsto G_{\mathcal{E}}(\xi,\xi^*)\chi(\alpha_{\mathcal{E}} \xi +\beta_{\mathcal{E}} \xi^*),
\end{align}
where $\alpha_{\mathcal{E}},\beta_{\mathcal{E}}\in\mathbb{C}$ and $G_{\mathcal{E}}(\xi,\xi^*)$ is a Gaussian prefactor \cite{gaussian_qi}. Suppose $\hat{\rho}_{s_0,\delta_0,n,\mathcal{E}}=\hat{\rho}_{s_0',\delta_0',n,\mathcal{E}'}$. Then the polynomial parts must match:
\begin{align}
    \begin{split}        
    &P_{s_0,\delta_0,n}(\alpha_{\mathcal{E}}\xi+\beta_{\mathcal{E}}\xi^*,(\alpha_{\mathcal{E}}\xi+\beta_{\mathcal{E}}\xi^*)^*)\\
    &\propto P_{s_0',\delta_0',n}(\alpha_{\mathcal{E}'}\xi+\beta_{\mathcal{E}'}\xi^*,(\alpha_{\mathcal{E}'}\xi+\beta_{\mathcal{E}'}\xi^*)^*).
    \end{split}
\end{align}
From Eq.~\eqref{eq:poly_part_exp}, comparing leading terms gives
\begin{align}
    \alpha_{\mathcal{E}'}\xi+\beta_{\mathcal{E}'}\xi^*=k(\alpha_{\mathcal{E}}\xi+\beta_{\mathcal{E}}\xi^*)
\end{align}
for some $k\in\mathbb{C}$, or
\begin{align}
    \alpha_{\mathcal{E}'}\xi+\beta_{\mathcal{E}'}\xi^*=k'(\alpha_{\mathcal{E}}\xi+\beta_{\mathcal{E}}\xi^*)^*\label{eq:conjugate_channel}
\end{align}
for some $k'\in\mathbb{C}$. In the first case,
\begin{align}
    P_{s_0,\delta_0,n}(\xi,\xi^*)\propto P_{s_0',\delta_0',n}(k\xi,(k\xi)^*),
\end{align}
so that
\begin{align}
\begin{split}
    \chi_{s_0,\delta_0,n}(\xi,\xi^*)
    &\propto \chi_{s_0',\delta_0',n}(k\xi,k^*\xi^*) \\
    &\quad \cdot \exp[-\tfrac{1}{2}(1-|k|^2)|\xi|^2]
\end{split}
\end{align}

This corresponds to the state $\ket{\psi}_{s_0,\delta_0,n}^{(\mathrm{p})}$ after a phase rotation by $\arg{k}$ and a loss channel with efficiency $|k|^2$. Since the left-hand side is a pure state and not the vacuum, we must have $|k|=1$. Then, by Theorem \ref{thm:gps_uniqueness}, either $s_0=s_0'=0$ and $|\delta_0|=|\delta_0'|$, or $(s_0,\delta_0)=(s_0',\pm\delta_0')$.

In the second case Eq.~\eqref{eq:conjugate_channel},
\begin{align}
    P_{s_0,\delta_0,n}(\xi,\xi^*)\propto P_{s_0',\delta_0',n}(k'\xi^*,k'^*\xi),
\end{align}
and
\begin{align}
    &\chi_{s_0,\delta_0,n}(\xi,\xi^*) \nonumber \\
    &\propto \chi_{s_0',\delta_0',n}(k'\xi^*,k'^*\xi)
      \exp[-\tfrac{1}{2}(1-|k|^2)|\xi|^2] \\
    &= \chi_{s_0',\delta_0'^*,n}(k'^*\xi,k'\xi^*)
      \exp[-\tfrac{1}{2}(1-|k|^2)|\xi|^2].
      \label{eq:conjugate_channel_char}
\end{align}

Here we used
\begin{align}
    \chi_{s_0,\delta_0,n}(\xi^*,\xi)=\chi_{s_0,\delta_0^*,n}(\xi,\xi^*),
\end{align}
since conjugating $\xi\mapsto\xi^*$ corresponds to transposition $\rho\mapsto\rho^T$ in the Fock basis, which is equivalent to $\delta_0\mapsto\delta_0^*$ in the particle form. From Eq.~\eqref{eq:conjugate_channel_char}, as before we obtain either $s_0=s_0'=0$ and $|\delta_0|=|\delta_0'|$, or $(s_0,\delta_0)=(s_0',\pm\delta_0'^*)$. Combining both cases yields the condition stated in the theorem.
\end{proof}

Note that the third condition \eqref{eq:conjugate_possibility}, absent in the pure-state case, arises because of the existence of Gaussian channels that implement time-reversal with added noise, such as the phase-conjugating channel \cite{phase_conj_channel}. The proof of Theorem \ref{thm:cparam_uniqueness_mixed} shows that the polynomial part of the characteristic function acts as a ``fingerprint'' of the non-Gaussian state, allowing the extraction of the control parameters $s_0,\delta_0$ from experimentally accessible data.

Theorem~\ref{thm:cparam_uniqueness_mixed} demonstrates that the non-Gaussian control parameters, determined by the shape of the characteristic function, are invariant under Gaussian errors acting on the signal modes. By contrast, Gaussian errors on the control modes generally modify the control moments $(C,\bm\beta)$, and thus change $(s_0,\delta_0)$. However, if the error channel is known, one can pre-compensate the control mode so that it yields the desired $(s_0,\delta_0)$ after the errors. 

As an illustrative example, consider a loss channel acting on a control mode close to the vacuum state, with 
\begin{align}
    C &= \mqty(1+\Delta c & 0 \\ 0 & 1+\Delta d), \\
    \bm\beta &= \mqty(\Delta \beta_x & 0 \\ 0 & \Delta \beta_p),
\end{align}
where $\Delta c, \Delta d, \Delta \beta_x, \Delta \beta_p \ll 1$. In this regime we have
\begin{align}
    s_0 &\sim \frac{|\Delta c - \Delta d|}{\Delta c + \Delta d}, \\
    \delta_0 &\sim \frac{2}{\sqrt{\Delta c + \Delta d}}\,
                (\Delta \beta_x - i \Delta \beta_p).
\end{align}

The effect of a loss channel with efficiency $\eta$ on $(C,\bm\beta)$ is
\begin{align}
    C &\mapsto \eta C + (1-\eta) I, \\
    \bm\beta &\mapsto \sqrt{\eta}\,\bm\beta.
\end{align}
Equivalently, in terms of the small deviations,
\begin{align}
    \Delta c &\mapsto \eta \Delta c, \\
    \Delta d &\mapsto \eta \Delta d, \\
    \Delta \beta_x &\mapsto \sqrt{\eta}\,\Delta \beta_x, \\
    \Delta \beta_p &\mapsto \sqrt{\eta}\,\Delta \beta_p.
\end{align}
Thus, both $s_0$ and $\delta_0$ remain invariant under loss. This observation reproduces the well-known fact that the effect of loss on the control mode can be mitigated in the ``weak-pump'' limit.
\section{Approximation of CPS}\label{sec:cps_derive}
When $s_0=0$ and $|\delta_0|>0$, the state $\ket{\psi}_{0,\delta_{0},n}^{(\mathrm{w})}$ is a good approximation of CPS~\cite{gkp,cps_konno}, defined as
\begin{equation}
    \ket{\mathrm{CPS}} = e^{-i\hat{x}^3} \ket{p=0}.
\end{equation}

Due to the symmetry at $s_0=0$, we can assume $\delta_{0x}>0$ and $\delta_{0p}=0$ without loss of generality. As shown in Ref.~\cite{gkp}, using the higher-order approximation of the Fock wavefunction $\phi_n(x)$ valid for $\tfrac{x}{2\sqrt{n+1/2}}\ll 1$~\cite{hermite_approx}, one has
\begin{align}
    \begin{split}        
    &\phi_n(x)\\
    &\appropto \cos(\sqrt{n+1/2}\,x - \frac{x^3}{24\sqrt{n+1/2}} + \mathcal{O}(x^5) - \frac{n}{2}\pi).
    \end{split}
\end{align}
Substituting this into the wave form Eq.~\eqref{eq:wave_form} yields
\begin{align}
    &\braket{x}{\psi}_{0,\delta_{0x},n}^{(\mathrm{w})}\\
    &\propto \phi_n(x+i\delta_{0x}) \\
    &\appropto \exp(-i\frac{p_0}{2}x)\exp(-kx^2)\exp(-i\gamma x^3),
\end{align}
where
\begin{align}
    p_0 &= 2\sqrt{n+1/2} + \frac{\delta_{0x}^2}{4\sqrt{n+1/2}},\\
    k   &= \frac{\delta_{0x}}{8\sqrt{n+1/2}},\\
    \gamma &= \frac{1}{24\sqrt{n+1/2}}.
\end{align}

If $\delta_{0x}\gg \frac{1}{2\sqrt{n+1/2}}$, this approximation is valid within the Gaussian envelope $\exp(-kx^2)$. In this regime, we obtain
\begin{align}
    \ket{\psi}_{0,\delta_{0x},n}^{(\mathrm{w})}
    \appropto \hat{D}(\tfrac{i}{2}p_0)\hat{S}(\tfrac{1}{3}\ln\gamma)\exp(-k\gamma^{-2/3}x^2)\ket{\mathrm{CPS}},
\end{align}
where the envelope coefficient is
\begin{align}
    k\gamma^{-2/3} = \frac{9^{1/3}}{2}(n+1/2)^{-1/6}\delta_{0x}.
\end{align}
This factor vanishes asymptotically as $n\to\infty$, which implies that
\begin{align}
    \ket{\psi}_{0,\delta_{0x},n}^{(\mathrm{w})}
    \appropto \hat{D}(\tfrac{i}{2}p_0)\hat{S}(\tfrac{1}{3}\ln\gamma)\ket{\mathrm{CPS}},\label{eq:cps_approx_x}
\end{align}
in the large-$n$ limit.

When we instead assume $\delta_{0x}=0$ and $\delta_{0p}>0$, the same conclusion can be seen in terms of the Fourier-transformed wavefunction of the CPS~\cite{cps_wave_airy}:
\begin{align}
    \braket{p}{\mathrm{CPS}}\propto \mathrm{Ai}\qty(\frac{p}{2\cdot3^{1/3}})
\end{align}
where $\mathrm{Ai}$ is the Airy function of the first kind. Using the asymptotic formula for the Fock wavefunction around the turning point $x=2\sqrt{n+1/2}$~\cite{hermite_approx,hermite_airy}, we have
\begin{align}
    \phi_n(x)\appropto \mathrm{Ai}\qty(n^{1/6}\qty(2\sqrt{n+1/2}-x)).
\end{align}
Then, from Eq.~\eqref{eq:wave_form}, within the range $|2\sqrt{n+1/2}-x| \ll 1/\delta_{0p}$,
\begin{align}
    \braket{x}{\psi}_{0,i\delta_{0p},n}^{(\mathrm{w})}\appropto \mathrm{Ai}\qty(n^{1/6}\qty(2\sqrt{n+1/2}-x)).
\end{align}
This indicates
\begin{align}
    \begin{split}        
    &\ket{\psi}_{0,i\delta_{0p},n}^{(\mathrm{w})}\\&\appropto \hat{D}(2\sqrt{n+1/2})\hat{S}\qty(-\tfrac{1}{3}\ln{24 \sqrt{n}})\hat{R}(\pi/2)\ket{\mathrm{CPS}},
    \end{split}
\end{align}
which is equivalent to Eq.~\eqref{eq:cps_approx_x} up to a $\pi/2$ rotation in the limit $n\to\infty$.
\section{\choijam{} isomorphism for Gaussian maps}\label{sec:choi_mat}

Any completely positive (CP), but not necessarily trace-preserving Gaussian map $M: \mathcal{D}(\mathcal{H}^{\otimes k}) \to \mathcal{D}(\mathcal{H}^{\otimes j})$ has one-to-one correspondence with a $j+k$-mode Gaussian state $\rho_M$ (Choi--Jamiołkowski isomorphism~\cite{nielsen_chuang}). The characteristic function of $\rho_M$ is given by
\begin{align}
    \chi_M(\bm{q}_1,\bm{q}_2)&=\Tr(\rho_M\qty(\hat{D}(\bm{q}_1)\otimes\hat{D}(\bm{q}_2)))\\
    &\propto\exp(-\frac12 (\bm{q}_1^T,\bm{q}_2^T)\Sigma_M\mqty(\bm{q}_1\\\bm{q}_2)+i(\bm{q}_1^T,\bm{q}_2^T)\bm\gamma_M),
\end{align}
where $\bm{q}_1, \bm{q}_2$ are $2j$- and $2k$-dimensional real quadrature vectors, respectively, and  
\begin{align}
    \Sigma_M = \mqty(J & L^T \\ L & K), \quad \bm{\gamma}_M = \mqty(\bm{\epsilon} \\ \bm{\delta}).
\end{align}

Now consider an $l+k$-mode Gaussian input state with covariance matrix and mean  
\begin{align}
    \Sigma = \mqty(A & B^T \\ B & C), \quad \bm{\gamma} = \mqty(\bm{\alpha} \\ \bm{\beta}).
\end{align}
For convenience, we call the first $l$ modes the \emph{signal} modes and the remaining $k$ modes the \emph{control} modes, and denote their quadrature vectors by $\bm{q}_\signal, \bm{q}_\idler$ respectively.

When the map $\mathcal{M}$ acts on the signal modes, the characteristic function of the output state is given by  
\begin{align}
    \chi'(\bm{q}_\signal,\bm{q}_\idler)
    = \int d\bm{q}_{\signal}' \, \chi_M(\bm{q}_\signal, Z^{\otimes l} \bm{q}_{\signal}') \, \chi(\bm{q}_{\signal}', \bm{q}_\idler).
\end{align}
Then the integrand has covariance matrix and mean:  
\begin{align}
    \Sigma_{\mathrm{all}}
    &= \mqty(
        A + Z^{\otimes l} K Z^{\otimes l} & Z^{\otimes l} L & B^T \\
        L^T Z^{\otimes l} & J & 0 \\
        B & 0 & C
    ),\\
    \bm{\gamma}_{\mathrm{all}} &= \mqty(\bm{\alpha} + Z^{\otimes l} \bm{\delta} \\ \bm{\epsilon} \\ \bm{\beta}),
\end{align}
with respect to $(\bm{q}_{s'}, \bm{q}_\signal, \bm{q}_\idler)$.

Since the integration is Gaussian, the output state's covariance matrix and mean can be computed analytically as:  
\begin{align}
    \Sigma' = \mqty(A' & B'^T \\ B' & C'), \quad \bm{\gamma}' = \mqty(\bm{\alpha}' \\ \bm{\beta}'),
\end{align}
\begin{align}
    A' &= K - L^T Z^{\otimes l} (A + Z^{\otimes l} K Z^{\otimes l})^{-1} Z^{\otimes l} L \label{eq:A_trans} \\
    B' &= -B (A + Z^{\otimes l} K Z^{\otimes l})^{-1} Z^{\otimes l} L \\
    C' &= C - B (A + Z^{\otimes l} K Z^{\otimes l})^{-1} B^T \\
    \bm{\alpha}' &= \bm{\epsilon} - L^T Z^{\otimes l} (A + Z^{\otimes l} K Z^{\otimes l})^{-1} (\bm{\alpha} + Z^{\otimes l} \bm{\delta}) \label{eq:alpha_trans} \\
    \bm{\beta}' &= \bm{\beta} - B (A + Z^{\otimes l} K Z^{\otimes l})^{-1} (\bm{\alpha} + Z^{\otimes l} \bm{\delta}) \label{eq:beta_trans}
\end{align}

A particularly important example is the \emph{damping operator}:
\begin{align}
    \hat{\Gamma}(\bm{\lambda}) = \bigotimes_m e^{-\lambda_m \hat{n}_m}
\end{align}
which corresponds to a Gaussian map with covariance matrix and mean:
\begin{align}
    \Sigma_{\Gamma}
    = \mqty(T & \sqrt{T^2 - 1} Z^{\otimes k} \\ \sqrt{T^2 - 1} Z^{\otimes k} & T), \quad \bm{\gamma}_\Gamma = 0\label{eq:damping_cov_mean}
\end{align}
where
\begin{align}
    T = \mathrm{diag}(t_1, t_1, \dots, t_k, t_k) 
\end{align}
and
\begin{align}
    t_m = \coth{\lambda_m}. \label{eq:ta_relation2}
\end{align}
This is equivalent to a tensor product of thermal states.

In the special case where $\rho_M$ is pure, the map $M$ is called a Gaussian filter~\cite{gaussian_filter}, which maps pure states to pure states. Applying Theorem~\ref{thm:purificationUniqueness} to $\hat{\rho}_M$, we find that any Gaussian filter can be expressed as a transformation of the form:
\begin{align}
    \ket{\psi} \to \hat{F} \ket{\psi} = \hat{U}_1 \hat{\Gamma}(\bm{\lambda}) \hat{U}_2 \ket{\psi}
\end{align}
where $\hat{U}_1, \hat{U}_2$ are unitary operators, assuming $j = k$. Even in the case $j \neq k$, the same form can be obtained by adding auxiliary vacuum modes $\ket{0}$.

Using the Baker--Campbell--Hausdorff formula
\begin{align}
    e^{\hat{A}} \hat{B} e^{-\hat{A}} = \hat{B} + [\hat{A}, \hat{B}] + \frac{1}{2}[\hat{A},[\hat{A},\hat{B}]] + \dots ,
\end{align}
the damping operator $\hat{\Gamma}(\lambda) = e^{-\lambda \hat{n}}$ acts on the annihilation and creation operators $\hat{a}, \hat{a}^\dagger$ as:
\begin{align}
    \hat{\Gamma}(\lambda)
    \mqty(\hat{a} \\ \hat{a}^\dagger)
    \hat{\Gamma}(\lambda)^{-1}
    = \mqty(k & 0 \\ 0 & k^{-1}) \mqty(\hat{a} \\ \hat{a}^\dagger), \quad
    k = e^\lambda.
\end{align}
Therefore, $\hat{F}$ also admits a matrix representation with respect to the operator vector $\hat{\bm{a}} = (\hat{a}_1, \hat{a}_1^\dagger, \dots, \hat{a}_k, \hat{a}_k^\dagger)^T$ \cite{gaussian_filter}:
\begin{align}
    \hat{F} \hat{\bm{a}} \hat{F}^{-1} = S \hat{\bm{a}} + \bm{b}.
\end{align}
\section{Derivation of the canonical forms}\label{sec:proof_Cmat}
Here we give proofs of the following theorems in the main text.
\thmpurificationUniquenessGps*
\thmpurificationUniqueness*

Note that, although Refs.~\cite{duan_simon_duan,duan_simon_simon} give a principle proof of Theorem \ref{thm:purificationUniqueness}, the uniqueness of $\hat{U}_\idler$ for a given $C,\bm\beta$ has not been explicitly stated in these proof. Here we provide a complete proof for rigorousity.

Since Thm.~\ref{thm:purificationUniqueness} is a generalized version of Thm.~\ref{thm:purificationUniquenessGps}, it is sufficient to give a proof for Thm.~\ref{thm:purificationUniqueness}.
\begin{proof}[Proof \textup{(Theorem~\ref{thm:purificationUniqueness})}]
By adding extra vacuum modes if necessary, we assume $l=k$ without loss of generality. Because $\bm{\alpha},\bm{\beta}$ can be arbitrary controlled by including displacements in $\hat{U}_\signal,\hat{U}_\idler$, it is sufficient to prove the case of $\bm\alpha=\bm\beta=0$.

From Williamson's theorem \cite{williamson_decomp}, $C$ can be decomposed as
\begin{align}
    C=SDS^T.\label{eq:williamson_decomp_C}
\end{align}
where $S$ is a symplectic matrix and the diagonal matrix with paired elements
\begin{align}
    D=\mathrm{diag}(a_1,a_1, \dots, a_k,a_k)
\end{align}
is called symplectic eigenvalues. Using that the covariance matrix of the two-mode squeezed state $\ket{\mathrm{TMSS}(a)}$ is given by
\begin{align}
    \Sigma_a=\mqty(
    aI& \sqrt{a^2-1}Z\\
    \sqrt{a^2-1}Z& aI
    ),
\end{align}
and using the symplectic matrices $S_\signal,S_\idler$ correponding to the Gaussian unitary operators $\hat{U}_\signal,\hat{U}_\idler$, the decomposition in Eq.~\eqref{eq:standard_form} can be written as
\begin{align}
    \begin{split}
    \Sigma=\mqty(S_\signal^T & 0\\0 & S_\idler^T)\mqty(
    D& \sqrt{D^2-1}Z^{\otimes k}\\
    \sqrt{D^2-1}Z^{\otimes k}& D
    )\mqty(S_\signal & 0\\0 & S_\idler).
    \end{split}\label{eq:standard_form_mat}
\end{align}
We will prove that this decomposition is always possible. By including the Gaussian unitary corresponding to the $S$ in Eq.~\eqref{eq:williamson_decomp_C} to $S_\idler$, we can assume $C=D$. In the same way, $A$ can also be assumed to be diagonal. From Eq.~\eqref{eq:tmss_def}, the symplectic eigenvalues $a_j$ have one-to-one correspondence to the Schmidt coefficients. From general properties of purification, the signal modes should share the Schmidt coefficients with the control modes, hence $A$ also has a Williamson's decomposition with the same symplectic eigenvalues. Thus, we can assume
\begin{align}
    A=C=D.\label{eq:thermal_assumption}
\end{align}

From the condition that $\ket{G}$ is a pure Gaussian state, $\Sigma$ should be a symplectic matrix, hence the block matrices $A,B,C$ should satisfy
\begin{align}
    A\Omega^{\otimes k} A+B^T\Omega^{\otimes k} B&=\Omega^{\otimes k},\label{eq:symp_cond1}\\
    C\Omega^{\otimes k} C+B\Omega^{\otimes k} B^T&=\Omega^{\otimes k},\label{eq:symp_cond2}\\
    B\Omega^{\otimes k} A+C\Omega^{\otimes k} B&=0,\label{eq:symp_cond3}
\end{align}
where $\Omega$ is the symplectic form
\begin{align}
    \Omega=\mqty(0&1\\-1&0)
\end{align}
and
\begin{align}
    \Omega^{\otimes k}=\mqty(\Omega& & \\&\ddots&\\&&\Omega)
\end{align}
From Eqs.~\eqref{eq:thermal_assumption} and \eqref{eq:symp_cond2}, we have
\begin{align}
    B\Omega^{\otimes k} B^T=-(D^2-1)\Omega^{\otimes k}.
\end{align}
By multiplying $(D^2-1)^{-1/2}Z^{\otimes k}$ and its transpose from both sides, where
\begin{align}
    Z=\mqty(1&0\\0&-1),
\end{align}
we obtain
\begin{align}
    (D^2-1)^{-1/2}Z^{\otimes k}B\Omega^{\otimes k} ((D^2-1)^{-1/2}Z^{\otimes k}B)^T=\Omega^{\otimes k},
\end{align}
which means $(D^2-1)^{-1/2}Z^{\otimes k}B$ is symplectic, thus
\begin{align}
    B=Z^{\otimes k}(D^2-1)^{1/2}S\label{eq:b_symplectic}
\end{align}
for some symplectic matrix $S$. Here we used the relation $Z\Omega Z=-\Omega$.

Meanwhile, from Eqs.~\eqref{eq:thermal_assumption} and \eqref{eq:symp_cond3}, we have
\begin{align}
    BD\Omega^{\otimes k}=-D\Omega^{\otimes k}B.\label{eq:B_omega_cond}
\end{align}
By multiplying $Z^{\otimes k}$ from right, we have
\begin{align}
    (BZ^{\otimes k})(D\Omega^{\otimes k})=(D\Omega^{\otimes k})(BZ^{\otimes k}).
\end{align}
This indicates that $BZ^{\otimes k}$ has the same eigenspace as $D\Omega^{\otimes k}$. Writing degeneracy of the symplectic eigenvalues $a_1,\dots,a_k$ as $d_1,\dots,d_m$ ($d_1+\dots+d_m=k$), we can write $B$ in a block-diagonal form as
\begin{align}
    B=\mathrm{diag}(B_1,\dots,B_m)
\end{align}
where, from Eq.~\eqref{eq:B_omega_cond}, each $B_j$ is a $2d_j\times 2d_j$ matrix satisfying
\begin{align}
    \Omega^{\otimes d_j} B_j+B_j\Omega^{\otimes d_j}=0.
\end{align}
Thus $S$ in Eq.~\eqref{eq:b_symplectic} is also block diagonal
\begin{align}
    S=\mathrm{diag}(S_1,\dots,S_m)
\end{align}
and each $S_j$ satisfies
\begin{align}
    \Omega^{\otimes d_j} S_j Z^{\otimes d_j}+S_j Z^{\otimes d_j}\Omega^{\otimes d_j}=0.
\end{align}
By multiplying $S_j^T$ from the left and $\Omega^{\otimes d_j} Z^{\otimes d_j}$ from the right, we obtain
\begin{align}
    S_j^TS_j&=I,
\end{align}
which means $S_j$ is an orthogonal symplectic matrix, thus we have $S^TDS=D$.
Therefore, from Eq.~\eqref{eq:b_symplectic}, we obtain
\begin{align}
    \Sigma=\mqty(S^T & 0\\0 & I)\mqty(
    D& \sqrt{D^2-1}Z^{\otimes k}\\
    \sqrt{D^2-1}Z^{\otimes k}& D
    )\mqty(S & 0\\0 & I).
\end{align}
This proves the possibility of the decomposition in Eq.~\eqref{eq:standard_form_mat}. From the construction, $\hat{U}_\idler$ only depends on $(C,\bm\beta)$.
\end{proof}
\section{Derivation of dampling transformation}\label{sec:proof_fock_equivalence_damp}
Here we give proofs of the following theorems in the main text.
\thmfockEquivalenceDampGps*
\thmfockEquivalenceDamp*

The domain of $\bm{t}$ is given by the following equivalent conditions:

(i)
\begin{align*}
T_+>1,T_-<-(\Pi_-(C+(T_+\oplus 0))^{-1}\Pi_-)^{-1}.
\end{align*}

(ii)
\begin{align*}
|T|&>1,\Pi_-(C+T)^{-1}\Pi_-<0
\end{align*}

(iii)
\begin{align*}
    |T|>1, \mathcal{D}_{\bm{t}}(C)>0.
\end{align*}

Here, $T$ is divided into a positive part and a negative part as
\begin{align}
    T=T_{+}\oplus T_{-},T_{+}>0, T_{-}<0
\end{align}
and $\Pi_-$ is a projection operator to the support of $T_-$.

Thm.~\ref{thm:fock_equivalence_damp} is a generalized version of Thm.~\ref{thm:fock_equivalence_damp_gps}. Thus, it is sufficient to give a proof for Thm.~\ref{thm:fock_equivalence_damp}.

\begin{proof}[Proof \textup{(Theorem~\ref{thm:fock_equivalence_damp})}]
We will first show that the damping operation transforms $(C,\bm\beta)$ according to Eq.~\eqref{eq:damping_trans}, then derive the condition on $T$. For $\lambda_j\geq 0$, suppose the tensor product of damping operators
\begin{align}
    \bigotimes_j e^{-\lambda_j\hat{n}_j}
\end{align}
is applied to the control modes.
By putting $t_m=\coth \lambda_m$ and using the \choijam{} representation of damping operation in Appendix \ref{sec:choi_mat} (Eqs.~\eqref{eq:A_trans}-\eqref{eq:beta_trans}, \eqref{eq:damping_cov_mean}), we obtain the transformation of $(C,\bm\beta)$ as
\begin{align}
\begin{split}
    (C,\bm{\beta})\to &\mathcal{D}_{T}(C,\bm{\beta})=(\mathcal{D}_{T}(C),\mathcal{D}_{T}(\beta)) \\
    :=(&T-\sqrt{T^2-1}(C+T)^{-1}\sqrt{T^2-1},\\&\sqrt{T^2-1}(C+T)^{-1}\bm{\beta}).
\end{split}\label{eq:damping_trans2}
\end{align}

If $T> I$, the transformation corresponds to a physical operation satisfying $\lambda_j>0$. The range of $T$ can be further extended by using that,
for arbitrary $|T_1|,|T_2|>I$ and $(C,\bm\beta)$, there is a relation
\begin{align}
    D_{T_1}(D_{T_2}(C,\bm\beta))&=D_{T_1\circ T_2}(C,\bm\beta)\\
    T_1\circ T_2&:=D_{T_1}(T_2)=\frac{T_1T_2+1}{T_1+T_2}.
\end{align}

For an arbitrary $T$ satisfying $|T|>I$, we can find a sufficiently small $\epsilon>0$ and $T'=(1+\epsilon)I$ such that 
\begin{align}
    T\circ T'>I.
\end{align}
Therefore, if $\mathcal{D}_{T}(C,\bm\beta)$ represents a physical control-mode representation satisfying Eq.~\eqref{eq:uncertainty}, $(C,\bm\beta)$ and $\mathcal{D}_{T}(C,\bm\beta)$ both have the same output state as $D_{T\circ T'}(C,\bm\beta)$.

Let us derive a closed-form condition for $\mathcal{D}_{T}(C,\bm\beta)$ satisfies the phyisical condition Eq.~\eqref{eq:uncertainty}. The case $T<-I$ corresponds to unphysical damping operations where $\lambda_j<0$. From the condition $\mathcal{D}_{T}(C)>0$, $T<-C$ is necessary. Conversely, when $T<-C$, because $C$ satisfies Eq.~\eqref{eq:uncertainty}, we have
\begin{align}
0<-C-T\leq -i\Omega^{\otimes k}-T,
\end{align}
thus
\begin{align}
\mathcal{D}_{T}(C)&\geq T-\sqrt{T^2-1}(i\Omega^{\otimes k}+T)^{-1}\sqrt{T^2-1}\\
&=i\Omega^{\otimes k}.
\end{align}
To summarize, under the assumption $T>0$ or $T<0$, the condition on $T$ is $T>I$ or $T<-C$.

Now we consider the condition on $T$ for general cases. $T$ can be divided into a positive part and a negative part as
\begin{align}
    T=T_{+}\oplus T_{-},
\end{align}
where $T_{+}>0$ and $T_{-}<0$.
We use that
\begin{align}
    T=(-\infty\oplus T_{-})\circ (T_{+}\oplus \infty).
\end{align}
This corresponds to applying the physical and unphysical damping operators separately. The control-mode representation after applying the physical part can be written as
\begin{align}
    C_-&:=\Pi_-D_{T_{+}\oplus \infty}(C)\Pi_-\\&=(\Pi_-(C+(T_+\oplus 0))^{-1}\Pi_-)^{-1},
\end{align}
where we write projection matrices to the support of $T_{\pm}$ as $\Pi_{\pm}$. From the discussion in the last paragraph,
\begin{align}
    T_{-}<-C_{-} \label{eq:T_cond}
\end{align}
is the necessary and sufficient condition for the physical $\mathcal{D}_{T}(C)$. By using the formula for inversion of a block matrix, this is equivalent to
\begin{align}
    \Pi_-(C+T)^{-1}\Pi_-&<0.\label{eq:damping_cond}
\end{align}
It is clear from Eq.~\eqref{eq:damping_trans2} that, Eq.~\eqref{eq:damping_cond} needs to be satisfies for having $\mathcal{D}_{T}(C)>0$. This shows that a weaker condition $\mathcal{D}_{T}(C)>0$ is enough for verifying the physicality condition.

\end{proof}
\section{Derivation of the Cayley-transformed form of the damping operation}\label{sec:cayley-transform}
Here we show that Under the Cayley transform (Eqs.~\eqref{eq:cayley_1}-\eqref{eq:cayley_3}), the damping transformation (Eq.~\eqref{eq:damping_trans}) can be written as Eqs.~\eqref{eq:damping_cayley_1} and \eqref{eq:damping_cayley_2}.

We have
\begin{align}
    C'\pm I&=T\pm I-\sqrt{T^2-I}(C+T)^{-1}\sqrt{T^2-I}\\
    &=\sqrt{T^2-I}\qty[(T\mp I)^{-1}-(C+T)^{-1}]\sqrt{T^2-I}\\
    &=\sqrt{T^2-I}(C+T)^{-1}\qty(C\pm I)(T\mp I)^{-1}\sqrt{T^2-I}\\
    &=\sqrt{T^2-I}(C+T)^{-1}\qty(C\pm I)\sqrt{\frac{T\pm I}{T \mp I}}.
\end{align}
Using this, we obtain
\begin{align}
\tilde{C}'&=(C'+I)^{-1}(C'-I)\\
&=\sqrt{\frac{T-I}{T+I}}\frac{C-I}{C+I}\sqrt{\frac{T-I}{T+I}}\\
&=\sqrt{\tilde{T}}\tilde{C}\sqrt{\tilde{T}}.
\end{align}
Also,
\begin{align}
\tilde{\beta}'&=(C'+I)^{-1}\beta'\\
&=(C'+I)^{-1}\sqrt{T^2-I}(C+T)^{-1}\beta\\
&=\sqrt{\frac{T-I}{T+I}}(C+I)^{-1}\beta\\
&=\sqrt{\tilde{T}}\tilde{\beta}.
\end{align}
\section{Proof of Theorem \ref{thm:gps_irreversible} and its generalization}\label{sec:proof_conversion_cond}
Here we give a proof for the following theorem in the main text.
\gpsIrreversible*
For proving Theorem \ref{thm:gps_irreversible}, we first show the following more general theorem.
\begin{apptheorem}\label{thm:conversion_cond}
    A pure Gaussian state with a control-mode representation $(C,\bm\beta)$ can be transformed to another Gaussian state with $(C',\bm\beta')$ via a Gaussian map acting on the signal modes, if and only if
    \begin{align}
        C'\leq C\label{eq:conv_cond_C}
    \end{align}
    and
    \begin{align}
        \bm{\beta}-\bm{\beta'}\in \mathrm{Im}(C-C').\label{eq:conv_cond_beta}
    \end{align}
\end{apptheorem}
\begin{proof}[Proof \textup{(Theorem~\ref{thm:conversion_cond})}]
Let the number of signal and control modes be $l$ and $k$, respectively. The covariance matrix and mean vector of the total system are given by
\begin{align}
\Sigma=\mqty(A & B^T\\B & C),\bm{\gamma}=\mqty(\bm{\alpha}\\\bm{\beta})
\end{align}

Suppose we apply a Gaussian map $M$ described by
\begin{align}
    \Sigma_M=\mqty(J&L^T\\L&K), \bm{\gamma}_M=\mqty(\bm{\epsilon}\\\bm{\delta})
\end{align}
(see Section \ref{sec:choi_mat}) to the signal mode and obtain a new control-mode representation $(C',\bm\beta')$ as a result. Here we will show that the necessary and sufficient condition for the existence of such an $M$ is given by
\begin{align}
    C'\leq C\label{eq:conv_cond_C_appendix}
\end{align}
and
\begin{align}
    \bm{\beta}-\bm{\beta'}\in \mathrm{Im}(C-C').\label{eq:conv_cond_beta_appendix}
\end{align}

From Theorem \ref{thm:purificationUniqueness}, by adding extra vacuum modes if necessary, we can assume without loss of generality that $l=k$, and
\begin{align}
    A&=D,\\
    B&=S(D^2-1)^{1/2}Z^{\otimes k},\\
    \bm{\alpha}&=0.
\end{align}
Here,
\begin{align}
    C=SDS^T
\end{align}
is the Williamson decomposition of $C$. Then, the transformation of the control-mode representation by the map $M$ is expressed as
\begin{align}
    C-C'&=S\sqrt{D^2-1}(D+K)^{-1}\sqrt{D^2-1}S^T\label{eq:C_trans2}\\
    \bm{\beta}-\bm{\beta}'&=S\sqrt{D^2-1}(D+K)^{-1}
    \bm{\delta}\label{eq:beta_trans2}
\end{align}
(see Section \ref{sec:choi_mat}).

Suppose the map $M$ exists. From Eq.~\eqref{eq:C_trans2}, the first condition~\eqref{eq:conv_cond_C_appendix} clearly holds. Also, since $(D+K)^{-1}>0$ and $\sqrt{D^2-1}\geq 0$, we have
\begin{align}
\begin{split}
    &\mathrm{Im}\qty(\sqrt{D^2-1}(D+K)^{-1}\sqrt{D^2-1})\\&=\mathrm{Im}\qty(\sqrt{D^2-1}(D+K)^{-1}).
\end{split}
\end{align}
Since $S$ is invertible, the second condition~\eqref{eq:conv_cond_beta_appendix} also holds.

Conversely, suppose Eqs.~\eqref{eq:conv_cond_C_appendix} and \eqref{eq:conv_cond_beta_appendix} hold. We define
\begin{align}
    \tilde{C}'=S^{-1}C'S^{-T}.
\end{align}
Then,
\begin{align}
    \tilde{C}'\leq S^{-1}CS^{-T}=D
\end{align}
Also, Eq.~\eqref{eq:C_trans2} can be written as
\begin{align}
    \tilde{C}'&=D-\sqrt{D^2-1}(D+K)^{-1}\sqrt{D^2-1}\\
    &=D_{D}(K)\label{eq:C_trans_damping}
\end{align}
where $\mathcal{D}_{T}(\cdot)$ is the damping transformation (Theorem \ref{thm:fock_equivalence_damp}). We want to a $K$ satisfying this. Assume $D$ can be written as
\begin{align}
    D&=\mathrm{diag}(t_1,t_1,\dots,t_m,t_m,1,1,\dots,1,1),\\
    &=\mathrm{diag}(D_{+},I_{2(k-m)})
\end{align}
where $D_+>1$. Using the same division, we write $\tilde{C}'$ in block form as
\begin{align}
\tilde{C}'&=\mqty(\tilde{C}_+'&\tilde{C}'^T_{+1}\\\tilde{C}_{+1}'&\tilde{C}_1')
\end{align}
From $\tilde{C}'\leq D$, it follows that $\tilde{C}_1'\leq I_{2(k-m)}$. Because $\tilde{C}_1'$ must satisfy the uncertainty principle, we conclude $\tilde{C}_1'=I_{2(k-m)}$, which is the covariance matrix of a pure state. This implies $\tilde{C}'^T_{+1}=0$, and thus
\begin{align}
    \tilde{C}'&=\mathrm{diag}(C'_+,I_{2(k-m)}),\\ &C'_+\leq D_+.
\end{align}
Therefore,
\begin{align}
    K&=\mathrm{diag}(K_+,I_{2(k-m)})\\
    \begin{split}
    K_+&=\sqrt{D_{+}^2-1}\qty(D_+-C'_+)^{-1}\sqrt{D_{+}^2-1}-D_+\\
    &=D_{-D_{+}}(C_{+}),
    \end{split}
\end{align}
satisfies Eq.~\eqref{eq:C_trans_damping}, and hence Eq.~\eqref{eq:C_trans2}. From Theorem \ref{thm:fock_equivalence_damp}, this $K$ represents a physical Gaussian map.

Moreover, from Eq.~\eqref{eq:conv_cond_beta_appendix},
\begin{align}
    (C-C')\tilde{\delta}=\bm{\beta}-\bm{\beta}'
\end{align}
holds for some $\tilde{\delta}$, thus we can set
\begin{align}
    \delta=\sqrt{D^2-1}S^T\tilde{\delta}
\end{align}
to satisfy Eq.~\eqref{eq:beta_trans2}. Thus, Eqs.~\eqref{eq:conv_cond_C_appendix} and \eqref{eq:conv_cond_beta_appendix} are necessary and sufficient conditions.
\end{proof}
\begin{proof}[Proof \textup{(Theorem~\ref{thm:gps_irreversible})}]
Suppose $n \geq 2$ and consider the conversion from $\ket{\psi}_{s_0, \delta_0, n, \hat{U}_\signal}$ to $\ket{\psi}_{s_0', \delta_0', n, \hat{U}_\signal'}$. From Theorems~\ref{thm:gps_uniqueness} and \ref{thm:conversion_cond}, this conversion is possible if and only if there exist control-mode representations $(C, \bm{\beta})$ and $(C', \bm{\beta}')$ corresponding to the non-Gaussian phase-sensitivity parameters $(s_0, \delta_0)$ and $(s_0', \delta_0')$, respectively, such that Eqs.~\eqref{eq:conv_cond_C} and \eqref{eq:conv_cond_beta} are satisfied.

Without loss of generality, we restrict our attention to diagonal control matrices $C = \mathrm{diag}(c, d)$ with $c > d$ and $C' = \mathrm{diag}(c', d')$ with $c' > d'$. Figure~\ref{fig:s0_cd}, which illustrates the relation between $c$, $d$, and $s_0$, provides an intuitive understanding of the convertibility conditions.

From Eq.~\eqref{eq:conv_cond_C}, we require $c' \leq c$ and $d' \leq d$. Therefore:
\begin{itemize}
    \item When $s_0 \geq 1$ and $s_0' < s_0$, this condition cannot be satisfied.
    \item When $0 \leq s_0 < 1$, it is always possible to choose control parameters such that $1 < d < c$. By selecting $c' = d' = 1 + \epsilon$ with sufficiently small $\epsilon > 0$ and $\bm{\beta}' = 0$, both Eqs.~\eqref{eq:conv_cond_C} and \eqref{eq:conv_cond_beta} are satisfied. This implies that the state can be converted to a Fock state (corresponding to $s_0 = 0, \delta_0 = 0$). Since $c = d = N$ for arbitrarily large $N$ also represents the same Fock state, this Fock state can further be converted to any other $(C', \bm{\beta}')$.
    \item When $s_0' > s_0$, since the damping transformation brings $(c, d)$ to $(\infty, 1/s_0)$, we can always find $C$ and $C'$ corresponding to $s_0$ and $s_0'$ that satisfy $C' < C$. Hence, both Eqs.~\eqref{eq:conv_cond_C} and \eqref{eq:conv_cond_beta} can be satisfied.
    \item When $s_0 = s_0' > 1$, Eq.~\eqref{eq:conv_cond_C} requires $C = C'$. Then, Eq.~\eqref{eq:conv_cond_beta} requires $\bm{\beta} = \bm{\beta}'$, which implies that the two states are identical: $(s_0, \delta_0) = (s_0', \delta_0')$.
    \item When $s_0 = s_0' = 1$, Eq.~\eqref{eq:conv_cond_C} requires $C = \mathrm{diag}(c, 1)$, $C' = \mathrm{diag}(c', 1)$ with $c' \leq c$. In this case, Eq.~\eqref{eq:conv_cond_beta} requires $\delta_{0x} = 0$.
\end{itemize}

In summary, the necessary and sufficient condition for the convertibility is that one of the following holds:
\begin{align}
\begin{cases}
0 \leq s_0 < 1, \\
s_0' > s_0, \\
s_0 = s_0' = 1 \ \text{ and } \ \delta_{0x} = 0.
\end{cases}
\end{align}
\end{proof}
\section{Derivation of particle and wave forms}\label{sec:proof_fock_wave_form}
Here we give proofs of the following theorems in the main text.
\particleForm*
\waveForm*

In addition, we show the following theorem, describing the relation between these forms.
\begin{apptheorem}[Particle form $\to$ Wave form]\label{thm:fock_to_wave}
The particle form and the wave form are related by a Gaussian unitary operator $\hat{U}^{(\mathrm{p}\to\mathrm{w})}_{s_0,\delta_0,n}$:
\begin{align}
    \ket{\psi}_{s_0,\delta_0,n}^{(\mathrm{w})} 
    = \hat{U}^{(\mathrm{p}\to\mathrm{w})}_{s_0,\delta_0,n}\,
      \ket{\psi}_{s_0,\delta_0,n}^{(\mathrm{p})}.
\end{align}
The action of $\hat{U}^{(\mathrm{p}\to\mathrm{w})}_{s_0,\delta_0,n}$ on the quadratures is given by
\begin{align}
\begin{split}
&\hat{U}^{(\mathrm{p}\to\mathrm{w})\dagger}_{s_0,\delta_0,n}
\begin{pmatrix}
\hat{x} \\[4pt]
\hat{p}
\end{pmatrix}
\hat{U}^{(\mathrm{p}\to\mathrm{w})}_{s_0,\delta_0,n}\\
&=
\begin{pmatrix}
0 & \tfrac{1}{\sqrt{s_0+1}} \\
-\sqrt{s_0+1} & 0
\end{pmatrix}
\begin{pmatrix}
\hat{x} \\
\hat{p}
\end{pmatrix}
+
\begin{pmatrix}
-\delta_{0p} \\
\delta_{0x}
\end{pmatrix}.
\end{split}\label{eq:particle_wave_trans}
\end{align}
\end{apptheorem}
\begin{proof}[Proof \textup{(Theorem~\ref{thm:particle_form})}]
Applying \choijam{} isomorphism (Appendix~\ref{sec:choi_mat}) to the standard form (Theorem \ref{thm:purificationUniquenessGps}), the output state of a two-mode non-Gaussian state generator can be represented as
\begin{align}
    \ket{\psi}_{C,\bm\beta,n,\hat{U}_\signal} = \hat{K} \ket{n},
\end{align}
using a Gaussian filter
\begin{align}
    \hat{K} = \hat{U}_p \hat{\Gamma}(\lambda) \hat{S}(r) \hat{D}(\bm{\beta}).
\end{align}

Here, $\hat{S}(r)= \exp \left[ \frac{1}{2} \left( r^* \hat{a}^2 - r \hat{a}^{\dagger 2} \right) \right]$, $\hat{D}(\bm{\beta})=\exp \left( \beta \hat{a}^\dagger - \beta^* \hat{a} \right)$, and $\hat{\Gamma}(\lambda)=\exp(-\lambda \hat{n})$ are the squeezing, displacement and damping operators, respectively, and $\hat{U}_p$ is a Gaussian unitary operator. Without loss of generality, we assume $C$ is diagonal, with eigenvalues $c \geq d$. Then, the parameters $\lambda,r$ are determined as
\begin{align}
    \lambda = \coth^{-1}{\sqrt{cd}}, r = \frac{1}{4} \log\left(\frac{c}{d}\right).
\end{align}

If $\hat{U}_p$ can be chosen such that
\begin{align}
    \hat{K} \hat{a} \hat{K}^{-1} \propto \hat{a}, \label{eq:fock_form_condition}
\end{align}
then
\begin{align}
    \hat{K} \ket{0} \propto \ket{0}, \label{eq:fock_form_condition_ket}
\end{align}
and hence the output state can be expressed as
\begin{align}
    \ket{\psi}_{C,\bm\beta,n,\hat{U}_\signal} &\propto \hat{K} \hat{a}^{\dagger n} \ket{0} \\
    &= \left(\hat{K} \hat{a}^\dagger \hat{K}^{-1}\right)^n \hat{K} \ket{0} \\
    &\propto \left(\hat{K} \hat{a}^\dagger \hat{K}^{-1}\right)^n \ket{0},\label{eq:particle_fock_sp}
\end{align}
which is a superposition of Fock states up to photon number $n$. In what follows, we determine such a $\hat{U}_p$ and derive the particle form Eq.~\eqref{eq:fock_form}.

We first consider the case $\bm{\beta} = 0$. In this case, we can choose a squeezing operator $\hat{S}(r')$ as $\hat{U}_p$. Then the action of $\hat{K}$ becomes
\begin{align}
&\hat{K} \mqty(\hat{a} \\ \hat{a}^\dagger) \hat{K}^{-1}\\
&= 
\mqty(\cosh r & \sinh r \\
      \sinh r & \cosh r)
\mqty(k & 0 \\
      0 & k^{-1})
\mqty(\cosh r' & \sinh r' \\
      \sinh r' & \cosh r')
\mqty(\hat{a} \\ \hat{a}^\dagger) \\
&= 
\mqty(k \cosh r & k^{-1} \sinh r \\
      k \sinh r & k^{-1} \cosh r)
\mqty(\cosh r' & \sinh r' \\
      \sinh r' & \cosh r')
\mqty(\hat{a} \\ \hat{a}^\dagger),
\end{align}
where
\begin{align}
    k = e^{\lambda}.
\end{align}
By choosing $r'$ such that
\begin{align}
    \tanh r' = -k^{-2} \tanh r,
\end{align}
we get:
\begin{align}
&\hat{K} \mqty(\hat{a} \\ \hat{a}^\dagger) \hat{K}^{-1}\\
\begin{split}        
&= 
\frac{1}{\sqrt{\eta}}
\mqty(k \cosh r & k^{-1} \sinh r \\
      k \sinh r & k^{-1} \cosh r)\\
&\quad\quad\cdot\mqty(k \cosh r & -k^{-1} \sinh r \\
      - k^{-1} \sinh r & k \cosh r)
\mqty(\hat{a} \\ \hat{a}^\dagger)
\end{split}\\
&= 
\mqty(\eta^{1/2} & 0 \\
      \eta^{-1/2} s_0 & \eta^{-1/2})
\mqty(\hat{a} \\ \hat{a}^\dagger),
\end{align}
where 
\begin{align}
    s_0 = \frac{c-d}{cd - 1}
\end{align}
is the non-Gaussian phase sensitivity, and 
\begin{align}
    \eta &= k^2 \cosh^2 r - k^{-2} \sinh^2 r \\
         &= \frac{(c+1)(d+1)}{cd - 1}.
\end{align}

Since this satisfies condition Eq.~\eqref{eq:fock_form_condition}, we obtain
\begin{align}
    \ket{\psi}_{C,n} &\propto \hat{K} \hat{a}^{\dagger n} \ket{0} \\
    &\propto \left(\hat{K} \hat{a}^\dagger \hat{K}^{-1}\right)^n \ket{0} \\
    &\propto \left(\hat{a}^\dagger +s_0 \hat{a}\right)^n \ket{0}.
\end{align}
This gives the particle form for $\bm{\beta} = 0$.

Next, we consider the case with finite displacement $\bm{\beta} \neq 0$. We modify $\hat{U}_p$ to include a displacement:
\begin{align}
    \hat{U}_p = \hat{D}(\beta') \hat{S}(r').
\end{align}
Then,
\begin{align}
\hat{K} \mqty(\hat{a} \\ \hat{a}^\dagger) \hat{K}^{-1}
&= 
\mqty(\eta^{1/2} & 0 \\
      \eta^{-1/2} s_0 & \eta^{-1/2})
\mqty(\hat{a} + \beta' \\ \hat{a}^\dagger + \beta'^*)
+
\mqty(\beta \\ \beta^*).
\end{align}
By choosing
\begin{align}
    \beta' = -\eta^{-1/2} \beta,
\end{align}
we get
\begin{align}
\hat{K} \mqty(\hat{a} \\ \hat{a}^\dagger) \hat{K}^{-1}
&=
\mqty(\eta^{1/2} & 0 \\
      \eta^{-1/2} s_0 & \eta^{-1/2})
\mqty(\hat{a} \\ \hat{a}^\dagger)
+
\mqty(0 \\ \eta^{-1/2} \delta_0),
\end{align}
where
\begin{align}
    \delta_0 &= -\eta^{-1/2} s_0 \beta + (\eta^{1/2} - \eta^{-1/2}) \beta^*\\
    &=\frac{2}{\sqrt{cd - 1}} \qty(\sqrt{\frac{d + 1}{c + 1}} \beta_x - i \sqrt{\frac{c + 1}{d + 1}} \beta_p)
\end{align}
Hence, the output state is
\begin{align}
    \ket{\psi}_{s_0,\delta_0,n} \propto \qty(\hat{a}^\dagger +s_0 \hat{a} + \delta_0)^n \ket{0}. \label{eq:fock_form_appendix}
\end{align}
\end{proof}
\begin{proof}[Proof \textup{(Theorem~\ref{thm:wave_form})}]
We consider the operator
\begin{align}
\hat{K}=\exp(-b_p \hat{p})\exp(-b_x \hat{x})\exp(-k\hat{x}^2).
\end{align}
and calculate the wavefuntion representation of $\hat{K}\ket{0}$. Because
\begin{align}
    \begin{split}        
    &\bra{x}\exp(-b_x \hat{x})\exp(-k\hat{x}^2)\ket{0}\\&\propto\exp[-\qty(k+\frac{1}{4})x^2-b_x x],
    \end{split}
\end{align}
by a Fourier transform we obtain
\begin{align}
\begin{split}        
    &\bra{p}\exp(-b_x \hat{x})\exp(-k\hat{x}^2)\ket{0}\\&\propto\exp[-\frac{1}{4(4k+1)}p^2+\frac{ib_x}{4k+1}p].
\end{split}
\end{align}
Thus, $p$-wave function of $\hat{K}\ket{0}$ can be calculated as
\begin{align}
    \bra{p}\hat{K}\ket{0}&\propto\exp[-\frac{1}{4(4k+1)}p^2+\frac{ib_x}{4k+1}p-b_pp].
\end{align}
This can be written as
\begin{align}
    \hat{K}\ket{0}\propto \hat{U}_w\ket{0}.
\end{align}
where
\begin{align}
    \hat{U}_w=\hat{S}\qty(\frac12\log(4k+1))D\qty(\frac{2b_x}{\sqrt{4k+1}},2\sqrt{4k+1}b_p).\label{eq:U_w_def}
\end{align}

If we define
\begin{align}
    \hat{K}'=\hat{U}_w^\dagger \hat{K},
\end{align}
we have
\begin{align}
    \hat{K}'\ket{0}\propto\ket{0},
\end{align}
thus
\begin{align}
    \hat{K}'\ket{n}\propto \left(\hat{K}' \hat{a}^\dagger \hat{K}'^{-1}\right)^n \ket{0},
\end{align}
similarly to the case of the particle form (Eq.~\eqref{eq:particle_fock_sp}).
The transformation of the quadrature operators is calculated as
\begin{align}
&\hat{K}\mqty(\hat{x}\\\hat{p})\hat{K}^{-1}\\
&=\mqty(1 & 0\\-4ik&1)\mqty(\hat{x}+2ib_p\\\hat{p}-2ib_x)\\
&=\mqty(1 & 0\\-4ik&1)\mqty(\hat{x}\\\hat{p})+\mqty(2ib_p\\8ab_p-2ib_x).
\end{align}
Therefore,
\begin{align}
    &\hat{K}'\mqty(\hat{x}\\\hat{p})\hat{K}'^{-1}\\
    \begin{split}        
    &=\mqty(1 & 0\\-4ik&1)\mqty((4k+1)^{-1/2}\hat{x}-\frac{2b_x}{4k+1}\\(4k+1)^{1/2}\hat{p}-2(4k+1)b_p)\\&\quad+\mqty(2ib_p\\8ab_p-2ib_x)
    \end{split}\\
    &=\mqty(\frac{1}{\sqrt{4k+1}} & 0\\-\frac{4ik}{\sqrt{4k+1}}&\sqrt{4k+1})\mqty(\hat{x}\\\hat{p})+\qty(-\frac{2b_x}{4k+1}+2ib_p)\mqty(1\\i).
\end{align}
Hence, $\hat{a}^\dagger$ is transformed as
\begin{align}
    &\hat{K}'\hat{a}^\dagger\hat{K}'^{-1}\\
    &=\frac12\mqty(1&-i)\mqty(\frac{1}{\sqrt{4k+1}} & 0\\-\frac{4ik}{\sqrt{4k+1}}&\sqrt{4k+1})\mqty(\hat{x}\\\hat{p})-\frac{2b_x}{4k+1}+2ib_p\\
    &=\frac12\qty(\frac{1-4k}{\sqrt{1+4k}}\hat{x}-i\sqrt{4k+1}\hat{p})-\frac{2b_x}{4k+1}+2ib_p\\
    &=\frac{1}{\sqrt{1+4k}}\hat{a}^\dagger-\frac{4k}{\sqrt{1+4k}}\hat{a}-\frac{2b_x}{4k+1}+2ib_p\\
    &\propto \hat{a}^\dagger-4k\hat{a}-\frac{2b_x}{\sqrt{4k+1}}+2i\sqrt{4k+1}b_p.
\end{align}
For comparing this to the particle form, we add an extra $-\pi/2$ phase rotation, obtaining
\begin{align}
    \begin{split}        
    &\hat{R}\qty(-\frac{\pi}{2})\hat{K}'\hat{a}^\dagger\hat{K}'^{-1}\hat{R}^\dagger\qty(-\frac{\pi}{2})\\&\propto \hat{a}^\dagger+4k\hat{a}+2\sqrt{4k+1}b_p+\frac{2ib_x}{\sqrt{4k+1}}.
    \end{split}\label{eq:wave_form_rot}
\end{align}
From this, we have the correspondence
\begin{align}
    k&=\frac{s_0}{4}\\
    b_x&=-\frac{1}{2}\sqrt{s_0+1}\delta_{0p}\\
    b_p&=\frac{1}{2\sqrt{s_0+1}}\delta_{0x},
\end{align}
leading to the wave form Eq.~\eqref{eq:wave_form}.
\end{proof}
\begin{proof}[Proof \textup{(Theorem~\ref{thm:fock_to_wave})}]
From the proof of Theorem \ref{thm:wave_form}, the conversion from the particle form to the wave form is given by the Gaussian unitary operator $\hat{U}_w\hat{R}\qty(\frac{\pi}{2})$ (Eqs.~\eqref{eq:U_w_def} and \eqref{eq:wave_form_rot}). Redefining this operator as $\hat{U}^{(\mathrm{p}\to\mathrm{w})}_{s_0, \delta_0, n}$, it can be written as
\begin{align}
    \hat{U}^{(\mathrm{p}\to\mathrm{w})}_{s_0, \delta_0, n}=\hat{S}\qty(\frac12\log(s_0+1))\hat{D}\qty(-\delta_{0p},\delta_{0x})\hat{R}\qty(\frac{\pi}{2}),
\end{align}
and its action is given by Eq.~\eqref{eq:particle_wave_trans}.
\end{proof}
\section{Proof of Theorem \ref{thm:gps_uniqueness}}\label{sec:proof_uniqueness}
Here we give a proof of the following theorem.
\gpsUniqueness*
\begin{proof}[Proof \textup{(Theorem~\ref{thm:gps_uniqueness})}]
In the particle form (Theorem \ref{thm:particle_form}), the coefficients for $\ket{n},\ket{n-1},\ket{n-2}$ can be explicitly written as
\begin{align}
    \braket{n}{\psi}_{s_0,\delta_0,n}^{(\mathrm{p})}&\propto  \sqrt{n!},\\
    \braket{n-1}{\psi}_{s_0,\delta_0,n}^{(\mathrm{p})}&\propto \delta_0 \sqrt{n \cdot n!},\\
    \braket{n-2}{\psi}_{s_0,\delta_0,n}^{(\mathrm{p})}&\propto \frac{1}{2} (s_0+\delta_0^2)\sqrt{n(n-1)\cdot n!}.
\end{align}
Hence,
\begin{align}
    \frac{\braket{n-1}{\psi}_{s_0,\delta_0,n}^{(\mathrm{p})}}{\braket{n}{\psi}_{s_0,\delta_0,n}^{(\mathrm{p})}}&= \delta_0 \sqrt{n},\label{eq:fock_coeff_nm1}\\
    \frac{\braket{n-2}{\psi}_{s_0,\delta_0,n}^{(\mathrm{p})}}{\braket{n}{\psi}_{s_0,\delta_0,n}^{(\mathrm{p})}}&= (s_0+\delta_0^2) \frac{\sqrt{n(n-1)}}{2}\label{eq:fock_coeff_nm2}.
\end{align}
Now, suppose
\begin{align}
\ket{\psi}_{s_0,\delta_0,n,\hat{U}_\signal}=\ket{\psi}_{s_0',\delta_0',n,\hat{U}'_\signal}.
\end{align}
We can assume one of the state is in the particle form, without loss of generality. Since phase rotation is the only Gaussian unitary operation leaving the state in finite superposition of Fock states \cite{stellar_rank}, we have
\begin{align}
\ket{\psi}_{s_0,\delta_0,n}^{(\mathrm{p})}=e^{i\theta\hat{n}}\ket{\psi}_{s_0',\delta_0',n}^{(\mathrm{p})}.
\end{align}
Thus, from Eqs.~\eqref{eq:fock_coeff_nm1} and \eqref{eq:fock_coeff_nm2}, we have
\begin{align}
    \delta_0&=\delta_0' e^{i\theta},\\
    s_0+\delta_0^2&=(s_0'+\delta_0'^2) e^{2i\theta}.
\end{align}
The solutions to these equations are either $s_0=s_0'=0$ and $|\delta_0|=|\delta_0'|$, or $s_0=s_0'> 0$ and $\delta_0=\pm \delta_0'$.
\end{proof}
\section{Proof of Theorem~\ref{thm:reduced_cparam}}\label{sec:c0_proof}

Here we provide the proof of Theorem~\ref{thm:reduced_cparam} in the main text.

\reducedCparam*

\begin{proof}[Proof \textup{(Theorem~\ref{thm:reduced_cparam})}]
Using the rotation transformation from Theorem~\ref{thm:fock_equivalence_rot_gps}, we can diagonalize $C$ as:
\begin{align}
    C = O^T \begin{pmatrix} c & 0 \\ 0 & d \end{pmatrix} O,
\end{align}
where, without loss of generality, we assume $c \geq d$. Thus, it suffices to consider the case where $C$ is diagonal.

Let us introduce the Cayley transform of the variables:
\begin{align}
    \tilde{c} &= \frac{c - 1}{c + 1}, \quad \tilde{d} = \frac{d - 1}{d + 1}, \\
    \tilde{\beta}_x &= \frac{1}{c + 1} \beta_x, \quad \tilde{\beta}_p = \frac{1}{d + 1} \beta_p.
\end{align}

Then, from Eqs.~\eqref{eq:cayley_1}--\eqref{eq:cayley_3}, the damping transformation in Theorem~\ref{thm:fock_equivalence_damp_gps} can be written as:
\begin{align}
    (\tilde{c}, \tilde{d}, \tilde{\beta}_x, \tilde{\beta}_p) \to \qty( \tilde{t} \tilde{c}, \, \tilde{t} \tilde{d}, \, \sqrt{ \tilde{t} } \tilde{\beta}_x, \, \sqrt{ \tilde{t} } \tilde{\beta}_p ).
\end{align}

Since this is simply a scaling transformation, the following three parameters are invariant:
\begin{align}
    \tilde{s}_0 &=\frac{ \tilde{d} }{ \tilde{c} }, \\
    \tilde{\beta}_{0x} &= \frac{ \tilde{\beta}_x }{ \sqrt{ \tilde{c} } }, \\
    \tilde{\beta}_{0p} &= \frac{ \tilde{\beta}_p }{ \sqrt{ \tilde{c} } }.
\end{align}

$s_0$ (Eq.~\eqref{eq:reduced_cparam_c0}) and $\delta_0$ (Eq.~\eqref{eq:reduced_cparam_beta0})can be expressed in terms of these invariants:
\begin{align}
    s_0 &= \frac{1 - \tilde{s}_0}{1 + \tilde{s}_0}, \\
    \delta_0 &= 2 \sqrt{ \frac{2}{1 + \tilde{s}_0} } \left( \tilde{\beta}_{0x} - i \tilde{\beta}_{0p} \right).
\end{align}

Thus, $s_0$ and $\delta_0$ are invariant under the damping transformation.

Conversely, if different sets of parameters $(c, d, \beta_x, \beta_p)$ yield the same $(s_0, \delta_0)$, they necessarily share the same $( \tilde{s}_0, \tilde{\beta}_{0x}, \tilde{\beta}_{0p} )$, which guarantees that there exists a $\tilde{t}$ that transforms one parameter set into the other.

\end{proof}
\section{Detail of approximation with lower stellar rank}\label{sec:approx_with_disp}
\begin{figure*}[htb]
    \centering
    \input{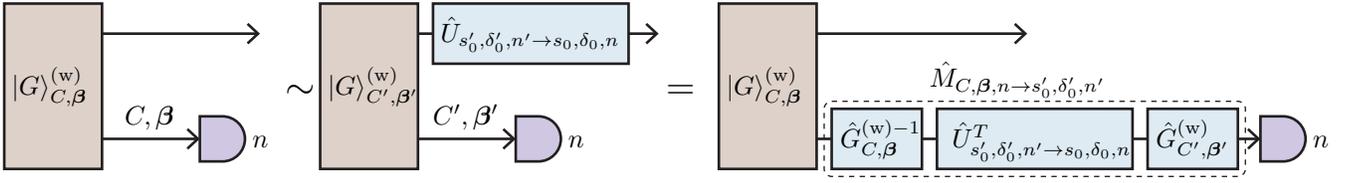}
    \caption{Graphical illustration of how the approximation of a two-mode non-Gaussian state generator can be represented as a Gaussian filter acting on the control mode.}
    \label{fig:photon_number_reduction_gps}
\end{figure*}
In Sec.~\ref{sec:approx_less_n}, we have shown how to obtain a two-mode non-Gaussian state generator with a smaller measured photon number $n'$, from the constants $k,d$ satisfying
\begin{align}
    \phi_n(x) \appropto \phi_{n'}(kx - d),
\end{align}
where $\phi_n(x)$ is the wavefunction of Fock state $\ket{n}$ and the approximation holds around
\begin{align}
    x = x_0 = \frac{\sqrt{s_0+1}}{s_0}\delta_{0p}.
\end{align}
In this section, we explain how to obtain such $k,d$. We also show that the transformation of the non-Gaussian state generator can be expressed by application of a Gaussian filter to the control mode.

\subsection{When \texorpdfstring{$x_0 = 0$}{x₀ = 0} and \texorpdfstring{$n \equiv n' (\mathrm{mod}\ 2)$}{n ≡ n' (mod 2)}}
First, we consider the simplest case with no displacement ($x_0 = 0$) and with $n$ and $n'$ of the same parity. In this case, we may set $d=0$. Let the wavefunction of the Fock state $\ket{n}$ be $\phi_n(x)$. This satisfies the Schrödinger equation
\begin{align}
   \phi_n''(x) = -(4n+2 - x^2) \phi_n(x).
\end{align}
Using the ``local momentum''
\begin{align}
    p(x) = \sqrt{4n+2 - x^2},
\end{align}
this can be rewritten as
\begin{align}
    \phi_n''(x) = -\qty[p(x)]^2 \phi_n(x).
\end{align}
Near $x=0$, and taking parity into account, the semiclassical approximation~\cite{wkb,gkp_kan,hermite_approx} gives
\begin{align}
    \phi_n(x) \appropto \cos\left(\int_0^x p(x') \, dx' - \frac{n}{2} \pi\right). \label{eq:wkb}
\end{align}

Now consider the scaled wavefunction
\begin{align}
    \tilde{\phi}_{n'}(x) \coloneqq \phi_{n'}(kx).
\end{align}
It satisfies
\begin{align}
    \tilde{\phi}_{n'}''(x) &= -\qty[\tilde{p}(x)]^2 \tilde{\phi}_{n'}(x),
\end{align}
with
\begin{align}
    \tilde{p}(x) &= k^2(k^2 x^2 - 4n' - 2).
\end{align}
Hence, choosing
\begin{align}
    k = \sqrt{\frac{2n+1}{2n'+1}}
\end{align}
ensures that near $x = 0$,
\begin{align}
    p(x) \sim \tilde{p}(x).
\end{align}
Since $n$ and $n'$ have the same parity, the initial conditions at $x = 0$ also match:
\begin{align}
    (\phi_n(0), \phi_n'(0)) \propto (\tilde{\phi}_{n'}(0), \tilde{\phi}_{n'}'(0)).
\end{align}
Therefore, from Eq.~\eqref{eq:wkb}, the approximation
\begin{align}
    \phi_n(x) \appropto \tilde{\phi}_{n'}(x)
\end{align}
holds around $x = 0$. 

\subsection{Case of \texorpdfstring{$x_0 \neq 0$}{x₀ ≠ 0}}
In the case with displacement ($x_0\neq0$), the Schrödinger equation for the squeezed and displaced Fock state wavefunction
\begin{align}
    \tilde{\phi}_{n'}(x) := \phi_{n'}(kx - d) \label{eq:fock_disp_squeeze}
\end{align}
is
\begin{align}
    \tilde{\phi}_{n'}''(x) &= -[\tilde{p}(x)]^2 \tilde{\phi}_{n'}(x),\\
    \tilde{p}(x) &= k^2\qty[(kx - d)^2 - (4n' + 2)].
\end{align}

Thus, the approximation
\begin{align}
    \phi_n(x) \appropto \tilde{\phi}_{n'}(x)
\end{align}
holds near $x = x_0$, if the local momentum matches:
\begin{align}
    p(x_0) = \tilde{p}(x_0), \label{eq:p_match}
\end{align}
and the initial conditions also match:
\begin{align}
    (\phi_n(x_0), \phi_n'(x_0)) \propto (\tilde{\phi}_{n'}(x_0), \tilde{\phi}_{n'}'(x_0)). \label{eq:init_match}
\end{align}

By numerically solving Eqs.~\eqref{eq:p_match} and \eqref{eq:init_match} simultaneously, the values of $(k, d)$ can be narrowed down to a finite set. Among these, selecting the one that minimizes the difference in the first derivative of the local momentum,
\begin{align}
    |p'(x_0) - \tilde{p}'(x_0)|,
\end{align}
provides a better approximation over a wider range (Method 1).

This method gives good approximations when $x_0$ lies inner of the turning point $x_t$ (where $p(x_t) = 0$), but when this is not the case, a better approximation can be obtained by solving the conditions
\begin{align}
    p(x_0) &= \tilde{p}(x_0),\\
    p'(x_0) &= \tilde{p}'(x_0),
\end{align}
simultaneously to find $k$ and $d$. In this case, the equations reduce to a cubic equation for $k$, which can be solved analytically. (Method 2). We employ a heuristic method based on numerical experiments as follows. Using the largest root $x_z > 0$ of $\phi_n(x)$,
\begin{itemize}
    \item When $|x_0| < x_z$, use Method 1.
    \item When $x_z \leq |x_0| < x_t$, use Method 2 with $x_0 = x_t$.
    \item When $x_t \leq |x_0|$, use Method 2.
\end{itemize}

\subsection{Representation by a Gaussian filter}
So far, we have shown that the output state of a two-mode non-Gaussian state generator can be approximated up to Gaussian unitary degrees of freedom. We now derive the explicit form of the transformation of the state generators, in the form of a Gaussian filter acting on the control mode.

The approximation Eq.~\eqref{eq:approx_wave} can be explicitly written as a transformation of non-Gaussian state generators:
\begin{align}
{}_\idler\braket{n}{G}_{C,\bm\beta}^{(\mathrm{w})}\appropto {}_\idler\bra{n'}\qty(\hat{U}_{s_0',\delta_0',n'\to s_0,\delta_0,n})_\signal\ket{G}_{C',\bm\beta'}^{(\mathrm{w})}
\end{align}
where $\ket{G}_{C,\bm\beta}^{(\mathrm{w})}$ is a Gaussian state having the control-mode representation $(C,\bm\beta)$ and corresponding to the wave form $\ket{\psi}_{s_0,\delta_0,n}^{(\mathrm{w})}$, and $\hat{U}_{s_0',\delta_0',n'\to s_0,\delta_0,n}$ is the Gaussian unitary operator defined in Eq.~\eqref{eq:approx_unitary_def}. 
Using the \choijam{} isomorphism, we define the operator corresponding to $\ket{G}_{C,\bm\beta}^{(\mathrm{w})}$ as $\hat{G}_{C,\bm\beta}^{(\mathrm{w})}$ (see Appendix~\ref{sec:choi_mat}). Then this transformation becomes
\begin{align}
    \ket{G}_{C,\bm\beta}^{(\mathrm{w})}&\to\qty(\hat{U}_{s_0',\delta_0',n'\to s_0,\delta_0,n})_\signal\ket{G}_{C',\bm\beta'}^{(\mathrm{w})}\\
    &=\qty[\hat{G}_{C',\bm\beta'}^{(\mathrm{w})}\hat{U}_{s_0',\delta_0',n'\to s_0,\delta_0,n}^T\qty(\hat{G}_{C,\bm\beta}^{(\mathrm{w})})^{-1}]_\idler\ket{G}_{C,\bm\beta}^{(\mathrm{w})}\\
    &=:\qty(\hat{M}_{C,\bm\beta,n\to s_0',\delta_0',n'})_\idler\ket{G}_{C,\bm\beta}^{(\mathrm{w})}\label{eq:wkb_idler}
\end{align}
Thus, the transformation can be expressed as applying a Gaussian filter $\hat{M}_{C,\bm\beta,n\to s_0',\delta_0',n'}$ on the control mode (see Fig.~\ref{fig:photon_number_reduction_gps}). Note that this operator is not necessarily be physical. Since an arbitrary non-Gaussian state generator $\ket{G}_{C,\bm\beta}$ can be expressed as
\begin{align}
    \ket{G}_{C,\bm\beta}=\hat{U}_\signal^{(\mathrm{w})}\ket{G}_{C,\bm\beta}^{(\mathrm{w})}
\end{align}
using a Gaussian unitary operator $\hat{U}_\signal^{(\mathrm{w})}$ acting on the signal mode, the transformation of $\ket{G}_{C,\bm\beta}$ is also written as
\begin{align}
    \ket{G}_{C,\bm\beta}\to\qty(\hat{M}_{C,\bm\beta,n\to s_0',\delta_0',n'})_\idler\ket{G}_{C,\bm\beta}
\end{align}
using the same $\hat{M}_{C,\bm\beta,n\to s_0',\delta_0',n'}$.

The choice of $(C',\bm\beta')$ corresponding to $(s_0',\delta_0')$ is not unique, and has the degree of freedom of applying a damping operation. One possible choice is the $C'$ that has the same symplectic eigenvalues as $C$, which we employ for the algorithm implemented in Sec.~\ref{sec:optimization}.
\section{Invariant non-Gaussian control parameters}\label{sec:invariant_s0}
In Sec.~\ref{sec:multi-mode}, we defined the non-Gaussian control parameters $\bm{s}_0,\bm{\delta}_0$ for multi-mode non-Gaussian state generators. Although this definition is intuitive and plays a key role in the optimization algorithm proposed in this work, it suffers from the drawback that it is not invariant under the damping transformation. In this section, we introduce an alternative definition---the \emph{invariant non-Gaussian control parameters}---which are invariant under damping. These quantities are useful for a quantitative analysis of the properties of the output states, and will be used in the discussion of the GKP state in Appendix~\ref{sec:opt_detail}.

Previously, the partial moments $C_m,\bm{\beta}_m$ were defined as the $2\times 2$ submatrices and subvectors of $C,\bm{\beta}$ [Eqs.~\eqref{eq:partial_cparam_C} and \eqref{eq:partial_cparam_beta}]. Instead, consider projecting all control modes except the $m$-th mode onto vacuum:
\begin{align}
    \ket{G}_m \coloneqq \Biggl(\prod_{\substack{l \leq m' < l+k \\ m' \neq m}} {}_{m'}\!\bra{0}\Biggr)\ket{G}.
\end{align}
The \emph{invariant partial moments} $\tilde{C}_m, \tilde{\bm{\beta}}_m$ are defined as the control moments of $\ket{G}_m$. By construction, $\tilde{C}_m$ and $\tilde{\bm{\beta}}_m$ are invariant under any damping transformation applied to modes other than the $m$-th mode. Consequently, the \emph{invariant non-Gaussian control parameters} $\bm{\tilde{s}}_{0}, \bm{\tilde{\delta}}_{0}$ are obtained by replacing $C_m,\bm{\beta}_m$ with $\tilde{C}_m,\tilde{\bm{\beta}}_m$ in the definitions of the standard control parameters. These parameters are therefore invariant under damping.

The physical interpretation of $\bm{\tilde{s}}_{0},\bm{\tilde{\delta}}_{0}$ is that they represent the values of $s_0,\delta_0$ when all other control modes are projected onto the vacuum state ($n=0$). Heuristically, $\bm{\tilde{s}}_{0}$ tends to take larger values than $\bm{s}_{0}$, due to the additional Gaussian filtering induced by these vacuum projections. These invariant quantities are particularly useful for quantitative comparisons of the output states of multi-mode non-Gaussian state generators, such as in the GKP state example discussed in Appendix~\ref{sec:opt_detail}.

\section{Details of the optimization examples}\label{sec:opt_detail}
We now provide details of the application of the algorithm in Sec.~\ref{sec:optimization} to the examples presented in Sec.~\ref{sec:optimization_example}. 
The procedure consists of two steps: photon-number reduction and probability enhancement. 

In Figs.~\ref{fig:cps_opt}, \ref{fig:cps_opt}, \ref{fig:gkp_opt}, and \ref{fig:random}, we show the resulting state generators after both steps (``Final''), together with the intermediate result after the first step (``Reduced''). Each figure includes schematics of the state generator and its parameters, as well as the control-mode representation $(C,\bm{\beta})$. For clarity, we adopt an ordering convention for $(C,\bm{\beta})$ given by $x_1,\dots,x_k,p_1,\dots,p_k$, instead of the $x_1,p_1,\dots,x_k,p_k$ convention used in the main text.

All squeezing values $r$ are reported in decibel (dB) scale, according to the conversion formula
\begin{equation}
    r_{\mathrm{dB}} = \frac{20}{\ln{10}}\, r \approx 8.686\, r.
\end{equation}

We remark that, in all examples some probability enhancement is already observed in the first step, reflecting the general tendency that larger photon numbers correspond to lower success probabilities.

\subsection{Cat state}
We consider the GPS scheme~\cite{GPS} for generating cat states, illustrated in Fig.~\ref{fig:cat_opt}(a).  
The initial squeezing is set to $r_1=-r_2=\SI{5}{dB}$, and the beamsplitter reflectance to $0.1$, both of which are reasonable experimental assumptions.  
We begin with $n=15,16$ photon detections for odd and even cat states, respectively, and set the target photon numbers to $n'=5,6$ so that the resulting states maintain high quality.  
To quantify the performance, we use the $x^2$-squeezing~\cite{cat_nls} defined in Eq.~\eqref{eq:cat_sqz_def} in the main text.  
The detailed results are presented in Fig.~\ref{fig:cat_opt}.

\subsection{CPS state}
We consider the scheme of Ref.~\cite{gkp} for generating CPS states using a displaced two-mode squeezed state, illustrated in Fig.~\ref{fig:cps_opt}(a).  
The initial squeezing is set to $r_1=-r_2=\SI{5}{dB}$, with a displacement $\alpha_2=1$ applied to the control mode.  
A displacement is also applied to the signal mode, although not essential, so as to center the output state in phase space.  
We begin with $n=20$ and set the target photon number to $n'=7$.  
To quantify the performance, we employ the cubic nonlinear squeezing~\cite{cps_nls} defined in Eq.~\eqref{eq:cps_sqz_def} in the main text.  
The detailed results are presented in Fig.~\ref{fig:cps_opt}.

\subsection{GKP state}\label{sec:opt_gkp_detail}
We consider the cat-breeding protocol~\cite{gkp_breeding} with $x=0$ conditioning, shown in Fig.~\ref{fig:gkp_opt}(a). Since homodyne detection is Gaussian, exchanging the order of photon-number measurement and conditioning maps this protocol to the Gaussian breeding circuit~\cite{gaussian_breeding} in Fig.~\ref{fig:gkp_opt}(b). A detailed derivation of this equivalence can be found in the Supplemental Material of Ref.~\cite{xanadu_architecture}.

The circuit obtained from the breeding has the form
\begin{align}
C &= \mqty( C_{xx} & 0 \\ 0 & d I),\\
C_{xx} &=
\mqty(
 a & a-\tfrac{1}{d} & \cdots & a-\tfrac{1}{d} \\
 a-\tfrac{1}{d} & a & \cdots & a-\tfrac{1}{d} \\
 \vdots & \vdots & \ddots & \vdots \\
 a-\tfrac{1}{d} & a-\tfrac{1}{d} & \cdots & a),
\end{align}
where $d>1$ and $ad>1$. The invariant non-Gaussian phase sensitivity $\tilde{s}_{0m}$ (see Appendix~\ref{sec:invariant_s0}) is given by
\begin{align}
    \tilde{s}_{0m} = \tilde{s}_{0} = \frac{d-\tilde{c}}{\tilde{c}d-1},
\end{align}
with
\begin{align}
    \tilde{c} = \frac{1}{d}\cdot\frac{a+(k-1)(a-1/d)}{a+(k-2)(a-1/d)}-1.
\end{align}

Breeding $k$ cat states with non-Gaussian phase sensitivity $s_0$ yields
\begin{align}
    \tilde{s}_0 = ks_0 + k-1. \label{eq:s_para_evolution}
\end{align}
Using this $\tilde{s}_{0}$, the wavefunction of the generated state can be written explicitly as
\begin{align}
    \psi(x) \propto \qty[\phi_n(x)]^n \exp\!\left[-\frac{\tilde{s}_{0}-k+1}{4}x^2\right].
\end{align}

For the initial state generator in our optimization, we use the circuit corresponding to the critical condition $s_0=1$ for each cat state, which is the maximum $s_0$ achieved by the commonly used photon-subtraction scheme (see Sec.~\ref{sec:classification}).  
The corresponding parameters for the cat-state generator in Fig.~\ref{fig:cat_opt}(a) are $r_1=-r_2=\SI{8.00}{dB}$ and $R=0.137$, giving the invariant non-Gaussian phase sensitivity $\tilde{s}_0=5$.

We evaluate the performance using the GKP non-Gaussian squeezing~\cite{gkp_squeezing}, defined as
\begin{align}
    \begin{split}        
    &\xi_{\mathrm{GKP}}=\\ &\min_{\lambda,\phi_1,\phi_2}\,
    \ev{2\cos^2\!\qty(\lambda\tfrac{\sqrt{\pi}}{2}\hat{x}+\phi_1) 
    + 2\cos^2\!\qty(\tfrac{1}{\lambda}\tfrac{\sqrt{\pi}}{2}\hat{p}+\phi_2) }.
    \end{split}
\end{align}
As in the cat-state case in Sec.~\ref{sec:overview}, there exists an optimal $\tilde{s}_0$ that yields the best GKP squeezing, as shown in Fig.~\ref{fig:gkp_evaluate}.  
This optimal value is smaller than that obtained from cat breeding, whether under the critical condition ($s_0=1$, Sec.~\ref{sec:classification}) or under the condition for optimal $x^2$-squeezing (Fig.~\ref{fig:gps_evaluate}(c)).  

In our case, after optimization, $\tilde{s}_0$ decreases from $\tilde{s}_0=5$ to $\tilde{s}_0=3.05$, approaching the optimal value.  
This reduction improves the photon number without degrading state quality, highlighting the non-optimality of the unmodified breeding scheme based on conventional photon-subtracted cat states with $s_0<1$.

\subsection{Random state}
Finally, we apply the algorithm to random Gaussian states.  
A random Gaussian unitary $U$ is sampled using the \texttt{random\_symplectic} function in MrMustard~\cite{mrmustard}, constructed as $U=WS(r)V$ from two Haar-random symplectic orthogonal matrices $W,V$ and a random squeezing vector $r\in[0,r_\mathrm{max}]$.  
With a displacement $d\in[0,d_\mathrm{max}]^{\otimes 2(k+1)}$, the state is
\begin{align}
    \ket{G}=D(d)U\ket{0}.
\end{align}
We use $r_\mathrm{max}=1$ and $d_\mathrm{max}=0.5$.  
Fig.~\ref{fig:random} shows one example, showing improvements in both probability and the measured photon numbers.

\begin{figure*}[htb]
  \centering
  \subfloat[Setup for cat state generation.]{
    \begin{minipage}{0.45\textwidth}
    \centering
    \input{figures/optimization/cat_system_detail_annotated_labels_overlay.tex}
    \end{minipage}
  }\\
  \subfloat[Parameters of odd and even cat state generators. 
``Original'': before optimization; 
``Reduced'': after the first step (photon-number reduction); 
``Final'': after the second step (probability enhancement).]{
    \begin{minipage}{0.95\linewidth}
    \centering
    \begin{tabular}{@{}lccccccc ccccccc@{}}
        \toprule
        & \multicolumn{7}{c}{Odd cat state} & \multicolumn{7}{c}{Even cat state} \\
        \cmidrule(lr){2-8}\cmidrule(lr){9-15}
        & $n$ & $r_1$ & $r_2$ & $R$ & $p$ & $C$ & $\bm{\beta}$ 
        & $n$ & $r_1$ & $r_2$ & $R$ & $p$ & $C$ & $\bm{\beta}$ \\
        \midrule
        Original 
        & $15$ & $5.0$ & $-5.0$ & $0.10$ & $\eformat{1.77e-6}$ &
        $\begin{bmatrix}0.60&0.00\\0.00&2.88\end{bmatrix}$ &
        $\begin{bmatrix}0.00\\0.00\end{bmatrix}$ &
        $16$ & $5.0$ & $-5.0$ & $0.10$ & $\eformat{8.29e-7}$ &
        $\begin{bmatrix}0.60&0.00\\0.00&2.88\end{bmatrix}$ &
        $\begin{bmatrix}0.00\\0.00\end{bmatrix}$ \\[1.2em]

        Reduced
        & $5$ & $3.46$ & $-4.22$ & $0.24$ & $\eformat{3.55e-4}$ &
        $\begin{bmatrix}0.97&0.00\\0.00&1.78\end{bmatrix}$ &
        $\begin{bmatrix}0.00\\0.00\end{bmatrix}$ &
        $6$ & $3.52$ & $-4.34$ & $0.22$ & $\eformat{1.21e-4}$ &
        $\begin{bmatrix}0.94&0.00\\0.00&1.84\end{bmatrix}$ &
        $\begin{bmatrix}0.00\\0.00\end{bmatrix}$ \\[1.2em]

        Final
        & $5$ & $14.33$ & $-5.96$ & $0.22$ & $\eformat{4.58e-2}$ &
        $\begin{bmatrix}0.91&0.00\\0.00&21.09\end{bmatrix}$ &
        $\begin{bmatrix}0.00\\0.00\end{bmatrix}$ &
        $6$ & $15.00$ & $-5.90$ & $0.21$ & $\eformat{3.84e-2}$ &
        $\begin{bmatrix}0.83&0.00\\0.00&25.17\end{bmatrix}$ &
        $\begin{bmatrix}0.00\\0.00\end{bmatrix}$ \\
        \bottomrule
    \end{tabular}
    \end{minipage}
  }\\
  \subfloat[Odd cat state before optimization.]{
    \begin{minipage}{0.45\textwidth}
    \centering
    \includegraphics[scale=1]{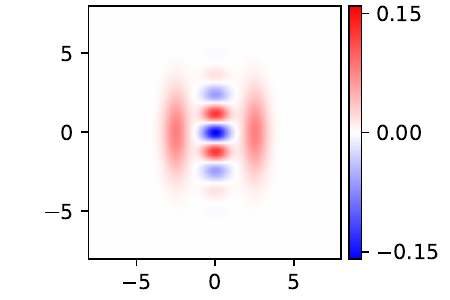}
    \end{minipage}
  }
  \subfloat[Odd cat state after optimization.]{
    \begin{minipage}{0.45\textwidth}
    \centering
    \includegraphics[scale=1]{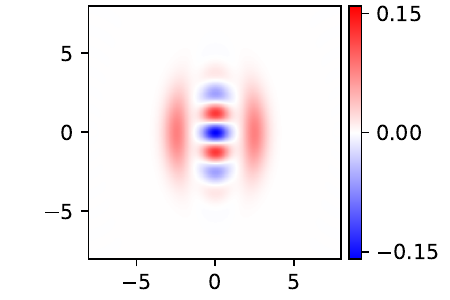}
    \end{minipage}
  }\\

  \subfloat[Even cat state before optimization.]{
    \begin{minipage}{0.45\textwidth}
    \centering
    \includegraphics[scale=1]{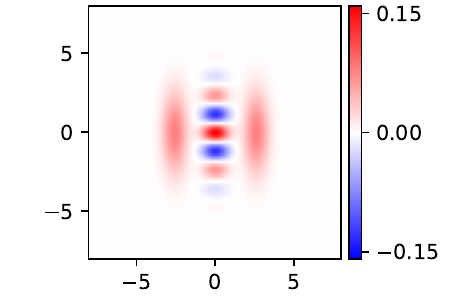}
    \end{minipage}
  }
  \subfloat[Even cat state after optimization.]{
    \begin{minipage}{0.45\textwidth}
    \centering
    \includegraphics[scale=1]{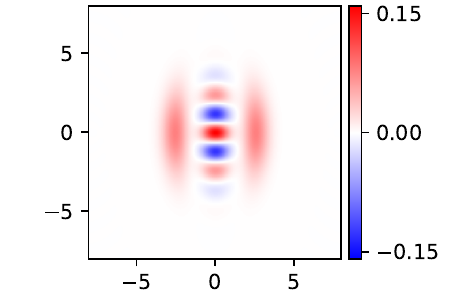}
    \end{minipage}
  }\\
  \hfill
 
  \caption{Optimization of odd and even cat states.}
  \label{fig:cat_opt}
\end{figure*}

\begin{figure*}[htb]
  \centering
  \subfloat[Setup for CPS generation.]{
    \begin{minipage}{0.3\textwidth}
    \centering
    \input{figures/optimization/cps_system_detail_annotated_labels_overlay.tex}
    \end{minipage}
  }
  \subfloat[Parameters of CPS generators.
``Original'': before optimization; 
``Reduced'': after the first step (photon-number reduction); 
``Final'': after the second step (probability enhancement).]{
    \begin{minipage}{0.6\textwidth}
    \centering
    \begin{tabular}{@{}cccccccccc@{}}
    \toprule
     & $n$ & $r_1$ & $r_2$ & $R$ & $\alpha_1$ & $\alpha_2$ & $p$ & $C$ & $\bm{\beta}$ \\
    \midrule
    Original 
    & $20$ & $5.0$ & $-5.0$ & $0.50$ & $3.40$ & $1.00$ & $\eformat{2.19e-8}$ &
    $\begin{bmatrix}
      1.74 & 0.00 \\ 0.00 & 1.74
    \end{bmatrix}$ &
    $\begin{bmatrix} 2.00 \\ 0.00 \end{bmatrix}$ \\[1.2em]
    
    Reduced
    & $7$ & $5.80$ & $-4.35$ & $0.60$ & $1.54$ & $0.85$ & $\eformat{2.49e-3}$ &
    $\begin{bmatrix}
      1.74 & 0.00 \\ 0.00 & 1.74
    \end{bmatrix}$ &
    $\begin{bmatrix} 1.69 \\ 0.00 \end{bmatrix}$ \\[1.2em]
    
    Final 
    & $7$ & $9.82$ & $-8.36$ & $0.59$ & $0.37$ & $2.34$ & $\eformat{7.43e-2}$ &
    $\begin{bmatrix}
      4.06 & 0.00 \\ 0.00 & 4.06
    \end{bmatrix}$ &
    $\begin{bmatrix} 4.68 \\ 0.00 \end{bmatrix}$ \\
    \bottomrule
    \end{tabular}
    \end{minipage}
  }\\
  \subfloat[CPS before optimization.]{
    \begin{minipage}{0.45\textwidth}
    \centering
    \includegraphics[scale=1]{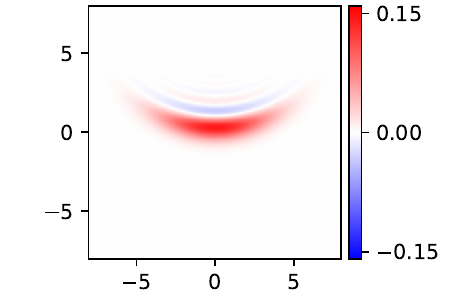}
    \end{minipage}
  }
  \subfloat[CPS after optimization.]{
    \begin{minipage}{0.45\textwidth}
    \centering
    \includegraphics[scale=1]{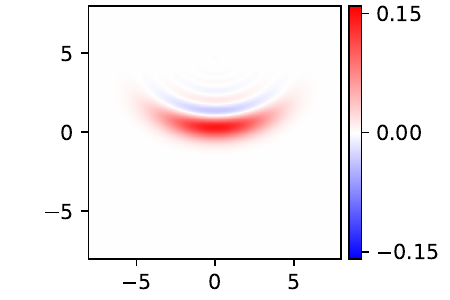}
    \end{minipage}
  }\\

  \caption{Optimization of CPS.}
  \label{fig:cps_opt}
\end{figure*}
\begin{figure*}[htb]
  \centering
  \subfloat[Breeding for generating GKP state.]{
    \begin{minipage}{0.5\textwidth}
    \centering
    \input{figures/optimization/gkp_breeding_system_detail_annotated_labels_overlay.tex}
    \end{minipage}
  }
  \subfloat[Setup for GKP state generation. Obtained from (a) by exchanging the order of photon-number measurements and conditionings \cite{xanadu_architecture}.]{
    \begin{minipage}{0.45\textwidth}
    \centering
    \input{figures/optimization/gkp_system_detail_annotated_labels_overlay.tex}
    \end{minipage}
  }\\
  \subfloat[Parameters of GKP state generators.
``Original'': before optimization; 
``Reduced'': after the first step (photon-number reduction); 
``Final'': after the second step (probability enhancement). The order of quadratures for $C$ and $\bm\beta$ is $x_1,\dots,x_k,p_1,\dots,p_k$.]{
\begin{minipage}{0.98\textwidth}
\centering
\begin{tabular}{@{}lccccccccc c c@{}}
\toprule
 & $n$ & $r_1$ & $r_2$ & $r_3$ & $r_4$ & $R_1$ & $R_2$ & $R_3$ & $p$ 
 & $C$ & $\bm{\beta}$ \\
\midrule
Original 
& $20$ & \phantom{-}6.53 & $-8.29$ & $-7.38$ & $-7.38$
& 0.80 & 0.33 & 0.50 & $\eformat{1.75e-12}$ &
$\begin{bmatrix}
0.45 & 0.27 & 0.27 & 0.00 & 0.00 & 0.00 \\
0.27 & 0.45 & 0.27 & 0.00 & 0.00 & 0.00 \\
0.27 & 0.27 & 0.45 & 0.00 & 0.00 & 0.00 \\
0.00 & 0.00 & 0.00 & 5.47 & 0.00 & 0.00 \\
0.00 & 0.00 & 0.00 & 0.00 & 5.47 & 0.00 \\
0.00 & 0.00 & 0.00 & 0.00 & 0.00 & 5.47
\end{bmatrix}$ &
$\begin{bmatrix}
0.00 \\ 0.00 \\ 0.00 \\ 0.00 \\ 0.00 \\ 0.00
\end{bmatrix}$ \\[1.2em]

Reduced 
& $7$ & \phantom{-}5.29 & $-6.31$ & $-4.17$ & $-4.17$
& 0.58 & 0.33 & 0.50 & $\eformat{5.54e-9}$ &
$\begin{bmatrix}
0.77 & 0.39 & 0.39 & 0.00 & 0.00 & 0.00 \\
0.39 & 0.77 & 0.39 & 0.00 & 0.00 & 0.00 \\
0.39 & 0.39 & 0.77 & 0.00 & 0.00 & 0.00 \\
0.00 & 0.00 & 0.00 & 2.62 & 0.00 & 0.00 \\
0.00 & 0.00 & 0.00 & 0.00 & 2.62 & 0.00 \\
0.00 & 0.00 & 0.00 & 0.00 & 0.00 & 2.62
\end{bmatrix}$ &
$\begin{bmatrix}
0.00 \\ 0.00 \\ 0.00 \\ 0.00 \\ 0.00 \\ 0.00
\end{bmatrix}$ \\[1.2em]

Final 
& $7$ & \phantom{-}8.09 & $-16.20$ & $-14.00$ & $-14.00$
& 0.60 & 0.33 & 0.50 & $\eformat{1.44e-4}$ &
$\begin{bmatrix}
0.89 & 0.85 & 0.85 & 0.00 & 0.00 & 0.00 \\
0.85 & 0.89 & 0.85 & 0.00 & 0.00 & 0.00 \\
0.85 & 0.85 & 0.89 & 0.00 & 0.00 & 0.00 \\
0.00 & 0.00 & 0.00 & 25.12 & 0.00 & 0.00 \\
0.00 & 0.00 & 0.00 & 0.00 & 25.12 & 0.00 \\
0.00 & 0.00 & 0.00 & 0.00 & 0.00 & 25.12
\end{bmatrix}$ &
$\begin{bmatrix}
0.00 \\ 0.00 \\ 0.00 \\ 0.00 \\ 0.00 \\ 0.00
\end{bmatrix}$ \\
\bottomrule
\end{tabular}
\end{minipage}
}\\

  \subfloat[GKP state before optimization.]{
    \begin{minipage}{0.45\textwidth}
    \centering
    \includegraphics[scale=1]{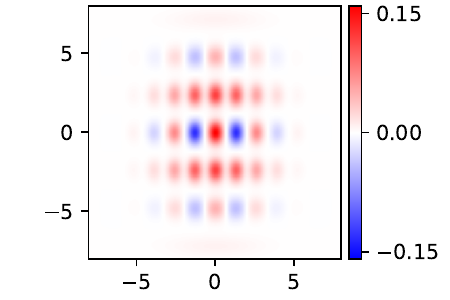}
    \end{minipage}
  }
  \subfloat[GKP state after optimization.]{
    \begin{minipage}{0.45\textwidth}
    \centering
    \includegraphics[scale=1]{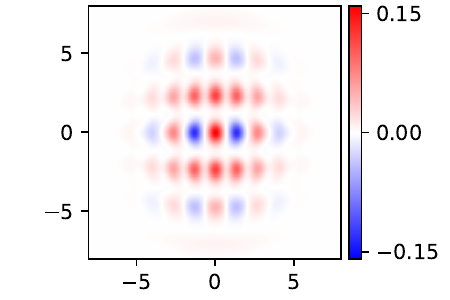}
    \end{minipage}
  }\\
  \caption{Optimization of GKP state.}
  \label{fig:gkp_opt}
\end{figure*}
\begin{figure*}[htb]
    \centering
    \input{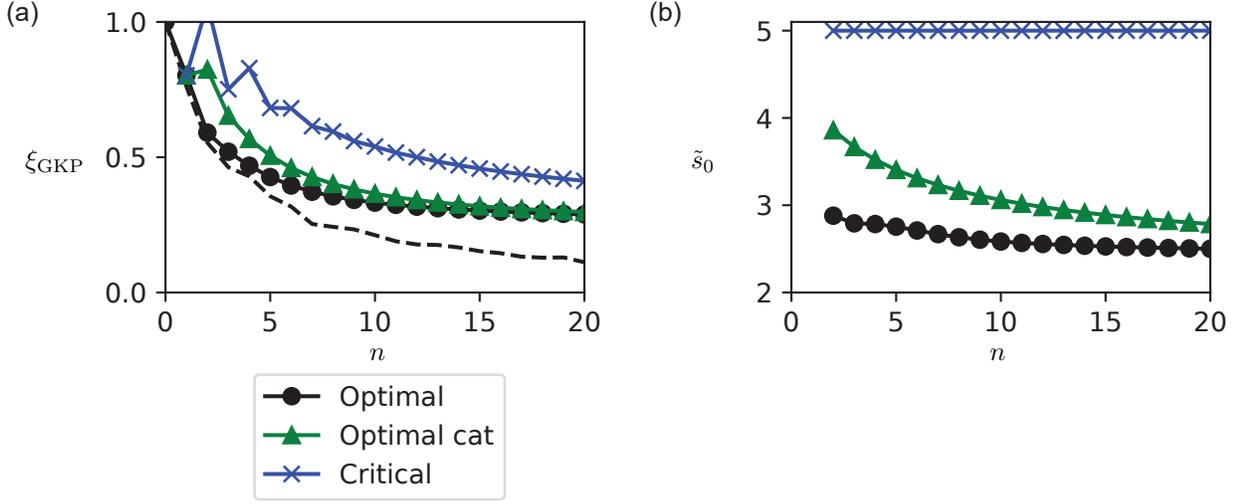}
    \caption{(a) GKP squeezing $\xi_{\mathrm{GKP}}$ as a function of $\tilde{s}_0$ for different $n$. 
“Optimal”: optimized over $\tilde{s}_0$. 
“Optimal cat”: breeding using cat states with $s_0$ taken from Fig.~\ref{fig:gps_evaluate}(c). 
“Critical”: breeding using cat states under the critical condition $s_0=1$. 
Dashed line: stellar-rank limit \cite{gkp_squeezing}. 
(b) Corresponding values of $\tilde{s}_0$ for panel (a). Values for $n=0,1$ are not shown, as the state is independent of $\tilde{s}_0$ up to Gaussian unitaries.
}

    \label{fig:gkp_evaluate}
\end{figure*}
\begin{figure*}[htb]
\centering
\subfloat[Parameters of Random state generators.
``Original'': before optimization; 
``Reduced'': after the first step (photon-number reduction); 
``Final'': after the second step (probability enhancement).
The order of quadratures for $C$ and $\bm\beta$ is $x_1,\dots,x_k,p_1,\dots,p_k$.]{
\begin{minipage}{0.95\textwidth}
\centering
\begin{tabular}{@{}c c c c c@{}}
\toprule
 & $n$ & $p$ & $C$ & $\bm{\beta}$ \\
\midrule

Original 
& $(14,14,14,14)$ 
& $\eformat{1.79e-30}$ 
& $\begin{bmatrix}
\phantom{-}1.51 & -0.14 & -0.15 & \phantom{-}0.20 & \phantom{-}0.26 & -0.22 & \phantom{-}0.12 & -0.28 \\
-0.14 & \phantom{-}1.23 & -0.06 & -0.28 & -0.11 & -0.02 & -0.22 & -0.21 \\
-0.15 & -0.06 & \phantom{-}0.74 & -0.05 & \phantom{-}0.00 & -0.08 & -0.48 & \phantom{-}0.01 \\
\phantom{-}0.20 & -0.28 & -0.05 & \phantom{-}0.95 & -0.12 & -0.32 & \phantom{-}0.05 & -0.21 \\
\phantom{-}0.26 & -0.11 & \phantom{-}0.00 & -0.12 & \phantom{-}0.79 & \phantom{-}0.11 & \phantom{-}0.15 & -0.10 \\
-0.22 & -0.02 & -0.08 & -0.32 & \phantom{-}0.11 & \phantom{-}1.13 & \phantom{-}0.15 & \phantom{-}0.39 \\
\phantom{-}0.12 & -0.22 & -0.48 & \phantom{-}0.05 & \phantom{-}0.15 & \phantom{-}0.15 & \phantom{-}1.78 & \phantom{-}0.07 \\
-0.28 & -0.21 & \phantom{-}0.01 & -0.21 & -0.10 & \phantom{-}0.39 & \phantom{-}0.07 & \phantom{-}1.41
\end{bmatrix}$ 
& $\begin{bmatrix}
-0.27 \\ -0.06 \\ \phantom{-}0.38 \\ -0.11 \\ \phantom{-}0.22 \\ \phantom{-}0.02 \\ \phantom{-}0.30 \\ -0.17
\end{bmatrix}$ \\[1.2em]

Reduced 
& $(9,14,2,12)$ 
& $\eformat{1.25e-22}$ 
& $\begin{bmatrix}
\phantom{-}0.78 & \phantom{-}0.06 & \phantom{-}0.10 & -0.16 & -0.06 & -0.21 & \phantom{-}0.08 & \phantom{-}0.04 \\
\phantom{-}0.06 & \phantom{-}1.19 & \phantom{-}0.14 & \phantom{-}0.38 & -0.25 & -0.08 & \phantom{-}0.16 & \phantom{-}0.05 \\
\phantom{-}0.10 & \phantom{-}0.14 & \phantom{-}0.93 & \phantom{-}0.00 & \phantom{-}0.08 & \phantom{-}0.03 & \phantom{-}0.00 & \phantom{-}0.05 \\
-0.16 & \phantom{-}0.38 & \phantom{-}0.00 & \phantom{-}0.90 & -0.07 & \phantom{-}0.11 & \phantom{-}0.08 & \phantom{-}0.00 \\
-0.06 & -0.25 & \phantom{-}0.08 & -0.07 & \phantom{-}1.45 & -0.15 & -0.17 & \phantom{-}0.30 \\
-0.21 & -0.08 & \phantom{-}0.03 & \phantom{-}0.11 & -0.15 & \phantom{-}1.19 & -0.13 & -0.50 \\
\phantom{-}0.08 & \phantom{-}0.16 & \phantom{-}0.00 & \phantom{-}0.08 & -0.17 & -0.13 & \phantom{-}1.17 & \phantom{-}0.01 \\
\phantom{-}0.04 & \phantom{-}0.05 & \phantom{-}0.05 & \phantom{-}0.00 & \phantom{-}0.30 & -0.50 & \phantom{-}0.01 & \phantom{-}1.46
\end{bmatrix}$
& $\begin{bmatrix}
-0.13 \\ -0.31 \\ \phantom{-}0.23 \\ \phantom{-}0.20 \\ -0.22 \\ \phantom{-}0.08 \\ -0.13 \\ \phantom{-}0.01
\end{bmatrix}$ \\[1.2em]

Final 
& $(9,14,2,12)$ 
& $\eformat{4.50e-6}$ 
& $\begin{bmatrix}
\phantom{-}2.27 & \phantom{-}1.06 & \phantom{-}0.79 & -0.88 & \phantom{-}2.87 & -6.24 & \phantom{-}1.55 & \phantom{-}6.02 \\
\phantom{-}1.06 & \phantom{-}14.71 & \phantom{-}1.96 & \phantom{-}8.28 & -15.94 & \phantom{-}2.76 & \phantom{-}8.06 & -5.77 \\
\phantom{-}0.79 & \phantom{-}1.96 & \phantom{-}1.21 & \phantom{-}0.62 & -0.25 & -1.38 & \phantom{-}1.06 & \phantom{-}1.37 \\
-0.88 & \phantom{-}8.28 & \phantom{-}0.62 & \phantom{-}5.87 & -11.40 & \phantom{-}5.83 & \phantom{-}3.92 & -7.46 \\
\phantom{-}2.87 & -15.94 & -0.25 & -11.40 & \phantom{-}30.91 & -19.03 & -7.54 & \phantom{-}23.99 \\
-6.24 & \phantom{-}2.76 & -1.38 & \phantom{-}5.83 & -19.03 & \phantom{-}23.92 & -1.46 & -25.77 \\
\phantom{-}1.55 & \phantom{-}8.06 & \phantom{-}1.06 & \phantom{-}3.92 & -7.54 & -1.46 & \phantom{-}6.39 & -0.50 \\
\phantom{-}6.02 & -5.77 & \phantom{-}1.37 & -7.46 & \phantom{-}23.99 & -25.77 & -0.50 & \phantom{-}29.54
\end{bmatrix}$
& $\begin{bmatrix}
-1.40 \\ -1.79 \\ -0.20 \\ \phantom{-}0.16 \\ -1.70 \\ \phantom{-}4.02 \\ -1.54 \\ -3.82
\end{bmatrix}$ \\
\bottomrule
\end{tabular}
\end{minipage}
}\\
\subfloat[Random state before optimization.]{
    \begin{minipage}{0.45\textwidth}
    \centering
    \includegraphics[scale=1]{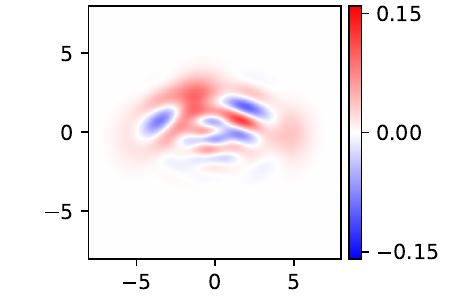}
    \end{minipage}
  }
  \subfloat[Random state after optimization.]{
    \begin{minipage}{0.45\textwidth}
    \centering
    \includegraphics[scale=1]{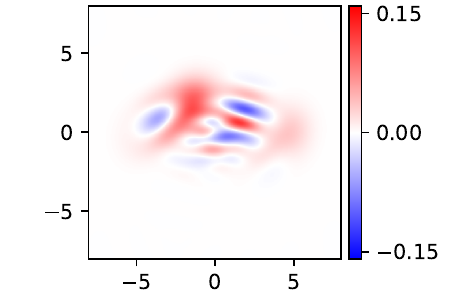}
    \end{minipage}
  }\\
\caption{Optimization of Random state.}
\label{fig:random}
\end{figure*}
\end{document}